\documentclass[twocolumn,preprintnumbers,superscriptaddress]{revtex4-2}
\usepackage[utf8]{inputenc}
\usepackage{mathtools,amsthm,amssymb,setspace}
\usepackage{flushend}
\usepackage{extarrows}
\usepackage{algorithm, algpseudocode}
\usepackage[pass]{geometry}
\usepackage{mathrsfs}%
\usepackage{stmaryrd}
\usepackage{graphicx}
\usepackage{tabularx}%
\usepackage{adjustbox}
\usepackage{float}
\usepackage[normalem]{ulem}
\usepackage{enumitem}
\setenumerate[2]{label=(\alph*), ref=\theenumi{} (\alph*), align=left}

\usepackage{xcolor,tcolorbox}
\definecolor{pink}{rgb}{0.72,0.18,0.95}
\definecolor{celeste}{rgb}{0.09, 0.81, 0.9}
\definecolor{emerald}{rgb}{0, 0.677778, 0.555556}
\definecolor{lavender}{rgb}{0.64, 0.47, 0.87}
\definecolor{apricot}{rgb}{0.98, 0.71, 0.43}
\definecolor{skyblue}{rgb}{0.53, 0.70, 0.80}
\definecolor{darklavender}{rgb}{0.48, 0.35, 0.65} 
\definecolor{darkapricot}{rgb}{0.80, 0.55, 0.33} 
\definecolor{navyblue}{rgb}{0.1, 0.2, 0.5}  

\newcommand{\beginsupplement}{
    \setcounter{section}{0}
    \renewcommand{\thesection}{S\arabic{section}}
    \setcounter{equation}{0}
    \renewcommand{\theequation}{S\arabic{equation}}
    \setcounter{table}{0}
    \renewcommand{\thetable}{S\arabic{table}}
    \setcounter{figure}{0}
    \renewcommand{\thefigure}{S\arabic{figure}}
    \newcounter{SIfig}
    \renewcommand{\theSIfig}{S\arabic{SIfig}}
}

\usepackage{hyperref}  
\hypersetup{
    bookmarks=true,
    unicode=false,
    pdftoolbar=true,
    pdfmenubar=true,
    pdffitwindow=false,
    pdfstartview={FitH},
    pdftitle={},
    pdfauthor={},
    pdfsubject={},
    pdfcreator={},
    pdfproducer={},
    pdfkeywords={},
    pdfnewwindow=true,
    colorlinks=true,
    linkcolor=darklavender, 
    citecolor=darklavender, 
    filecolor=darklavender, 
    urlcolor=darklavender, 
    breaklinks=true,
}
\usepackage{cleveref}%
\usepackage{ytableau}
\usepackage{youngtab}

\usepackage{lipsum}    
\usepackage{tikz-cd}
\usepackage{tikz}
\usepackage{soul}
\usepackage{tikz-3dplot}
\usetikzlibrary{quantikz2}
\usetikzlibrary{spy}
\usetikzlibrary{arrows.meta}
\usetikzlibrary{calc}
\usetikzlibrary{decorations.pathreplacing,calligraphy}
\usepackage{intcalc}

\usepackage{soul}
\usepackage[utf8]{inputenc}

\usepackage[newcommands]{ragged2e}
\usepackage{comment}
\usepackage{caption}
\usepackage{subcaption}
\usepackage{intcalc}

\newcommand{\id}{\mathbb{I}}

\newcommand{\ketbra} [2] {|#1\rangle\langle#2| }
\newcommand{\yd}[1] {\left[#1\right]}
\newcommand{\yt}[1] {\left(#1\right)}
\newcommand{\wt}[1] {\ulcorner\!#1\!\urcorner}

\makeatletter
\renewcommand\listoffigures{\@starttoc{lof}}
\renewcommand\listoftables{\@starttoc{lot}}

\makeatother

\DeclarePairedDelimiter\ceil{\lceil}{\rceil}

\DeclareMathOperator{\Tr}{tr}

\newtheorem{theorem}{Theorem}[section]
\newtheorem{definition}[theorem]{Definition}
\newtheorem{lemma}[theorem]{Lemma}

\newtheorem{proposition}[theorem]{Proposition}

\crefname{lemma}{Lemma}{Lemmas}
\crefname{proposition}{Proposition}{Propositions}
\crefname{algorithm}{Alg.}{Algs.}
\crefname{definition}{Def.}{Defs.}
\crefname{theorem}{Thm.}{Thms.}
\crefname{claim}{Claim}{Claims}
\crefname{section}{Sec.}{Secs.}
\crefname{figure}{Fig.}{Figs.}
\crefname{equation}{Eq.}{Eqs.}
\crefname{table}{Tab.}{Tabs.}
\crefname{enumi}{}{}

\begin{document}

\preprint{MIT-CTP/5775}
\title{Optimal Quantum Purity Amplification}

\author{Zhaoyi Li}
\email{ladmon@mit.edu}
\affiliation{Department of Physics, Massachusetts Institute of Technology, Cambridge, MA 02139, USA}
\author{Honghao Fu}
\affiliation{Computer Science and Artificial Intelligence Lab, Massachusetts Institute of Technology, Cambridge, MA 02139, USA \\
Concordia Institute for Information Systems Engineering, Concordia University, Montreal, QC H3G 1S6, Canada}
\date{\today}
\author{Takuya Isogawa}
\affiliation{Department of Nuclear Science and Engineering, Massachusetts Institute of Technology, Cambridge, MA 02139, USA}
\date{\today}
\author{Caio Silva}
\affiliation{Department of Physics, Massachusetts Institute of Technology, Cambridge, MA 02139, USA}
\date{\today}
\author{Isaac Chuang}
\affiliation{Department of Physics, Massachusetts Institute of Technology, Cambridge, MA 02139, USA}

\begin{abstract}
Quantum purity amplification (QPA) provides a novel approach to counteracting the pervasive noise that degrades quantum states. We present the optimal QPA protocol for general quantum systems and global noise, resolving a two-decade open problem. Under strong depolarization, our protocol achieves an exponential reduction in sample complexity over the best-known methods. We provide an efficient implementation of the protocol based on generalized quantum phase estimation. Additionally, we introduce SWAPNET, a sparse and shallow circuit that enables QPA for near-term experiments. Simulations in both digital and analog quantum settings, along with experiments on superconducting quantum processors, confirm the protocol’s robustness and practical utility. Our findings suggest that QPA could improve the performance of quantum information processing tasks, particularly in the context of Noisy Intermediate-Scale Quantum (NISQ) devices, where reducing the effect of noise with limited resources is critical.
\end{abstract}

\maketitle
As we transition into the era of Noisy Intermediate-Scale Quantum (NISQ) technologies~\cite{P18}, managing errors and improving computational reliability becomes crucial for advancing practical quantum applications. Fault-tolerant quantum computing (FTQC)~\cite{dennis2002topological,CDT09, Li2023concatenation} with quantum error-correcting codes (QEC)~\cite{PhysRevA.57.127,Gottesman2014, aliferis2005quantum} provides a robust framework for error control in large-scale quantum devices. However, performing quantum error correction on NISQ devices is less practical, as coding introduces significant overhead, with benefits only becoming apparent with large block sizes or multiple levels of concatenation. Achieving a universal set of quantum operations on code-protected logical qubits also requires complex methods like magic state distillation and sophisticated engineering, which significantly increase hardware demands, making many of these protocols challenging for state-of-the-art physical platforms, including neutral atoms, trapped ions, and superconducting circuits~\cite{PhysRevA.86.032324}.

Quantum purity amplification (QPA), building on previous quantum state purification~\cite{CEM99,KW01,CFLL23}, error symmetrization~\cite{P99}, or stabilization~\cite{B97} protocols, offers a potential alternative to FTQC and QEC for NISQ applications. 
In contrast to QEC, which requires full knowledge of the computational process, QPA is process-agnostic and only assumes access to the output states. By consuming multiple copies of noise-corrupted states, QPA produces a state with higher purity. Therefore, it is particularly useful when rerunning the algorithm is straightforward and the results can be stored in quantum memory, but achieving high fidelity in the output is challenging. In fact, in most practical applications, having access to multiple copies of a resource state is entirely realistic. For those tasks, QPA may achieve outcomes similar to those of QEC, potentially using fewer resources. Compared to quantum error mitigation, which uses classical post-processing to extrapolate a noiseless scenario, or virtual state purification~\cite{HMOL20,Koczor21} which improves expectation value estimation of a certain operator, QPA directly outputs quantum states, making it an ideal subroutine for other quantum information tasks, particularly in quantum simulation~\cite{GAN14}, sensing~\cite{DRC17}, quantum state preparation, quantum cryptography~\cite{Gisin2002}, and quantum machine learning~\cite{Biamonte2017}.

The study of related protocols began around two decades ago, yet two key issues remain unsolved. First, from a theoretical standpoint, previous work has focused primarily on processing states with simple symmetries, such as single-qubit states and coherent states~\cite{KW01, CEM99, AFFJ05}, due to the challenges stemming from optimization problems that lack inherent symmetry. Single-qubit protocols are not capable of handling correlations in many-body systems, so applying them independently to each register is suboptimal. The protocol for coherent states in Ref.~\cite{AFFJ05} is also restricted, as it is unsuitable for tasks requiring superpositions essential for quantum information processing. Although some work has explored extensions to higher-dimensional systems, these methods do not account for all symmetries, leading to either suboptimal sample complexity or imperfect yield, as seen in streaming protocols~\cite{F16, CFLL23} and probabilistic protocols~\cite{YCH24, lindner2023optical, yang2024quantum}. Second, on the practical side, efforts toward implementation in realistic quantum systems remain limited. Although depolarizing noise has served as a useful starting point, little attention has been devoted to studying general noisy inputs~\cite{grier2025streaming}, or to benchmarking at the circuit level. Moreover, even for single-qubit protocols, its efficient implementation remains largely unexplored, with little guidance on gate complexity. Experimental demonstrations have confirmed the feasibility of such methods, though their scope remains limited to depolarized single-qubit inputs~\cite{RMCF04}. 

To address these problems, we develop QPA, an efficient and sample-optimal algorithm applicable to general quantum systems with generic noise. On the theoretical side, we resolve the mathematical difficulty of optimizing QPA by converting the problem into a more symmetric form, leveraging recent advancements in combinatorics~\cite{S16}, eventually leading to a simple three-step interpretation of our protocol, consisting of Schur sampling, correction, and trace-out. We also derive the optimal sample complexity in terms of infidelity for generic inputs, extending beyond depolarizing noise. The theory behind our protocol is based on $d$-dimensional states (qudits), providing a structured mathematical framework. To ensure practical applicability, it can be easily converted to \(k=\lceil \log(d) \rceil\) qubit systems, ensuring feasibility on standard qubit hardware. We present an efficient implementation of gate complexity \( \text{poly}(n,\log d) \) based on generalized quantum phase estimation (GQPE). For near-term applications, we introduce the lightweight SWAPNET algorithm optimized for low depth and gate efficiency. 
The robustness of QPA is demonstrated through simulations with realistic noise models, showing improved fidelity in Hamiltonian simulation and adiabatic evolution aligning with theoretical results. Experiments on superconducting quantum processors further validate QPA's effectiveness in state preparation.

In~\cref{sec:setup}, we formalize the problem of finding the optimal QPA protocol as an optimization task, and in~\cref{sec:optimal}, we present its solution, which leads to the operational interpretation. In~\cref{sec:implementation}, we present an efficient implementation and introduce the SWAPNET circuit. Finally, in~\cref{sec:applications}, we present the results of numerical studies demonstrating its
application in Hamiltonian simulations, adiabatic evolutions, and state preparation on IBM superconducting quantum processors.

\section{Problem Setup}
\label{sec:setup}
\begin{figure}
\includegraphics[scale=0.2]{"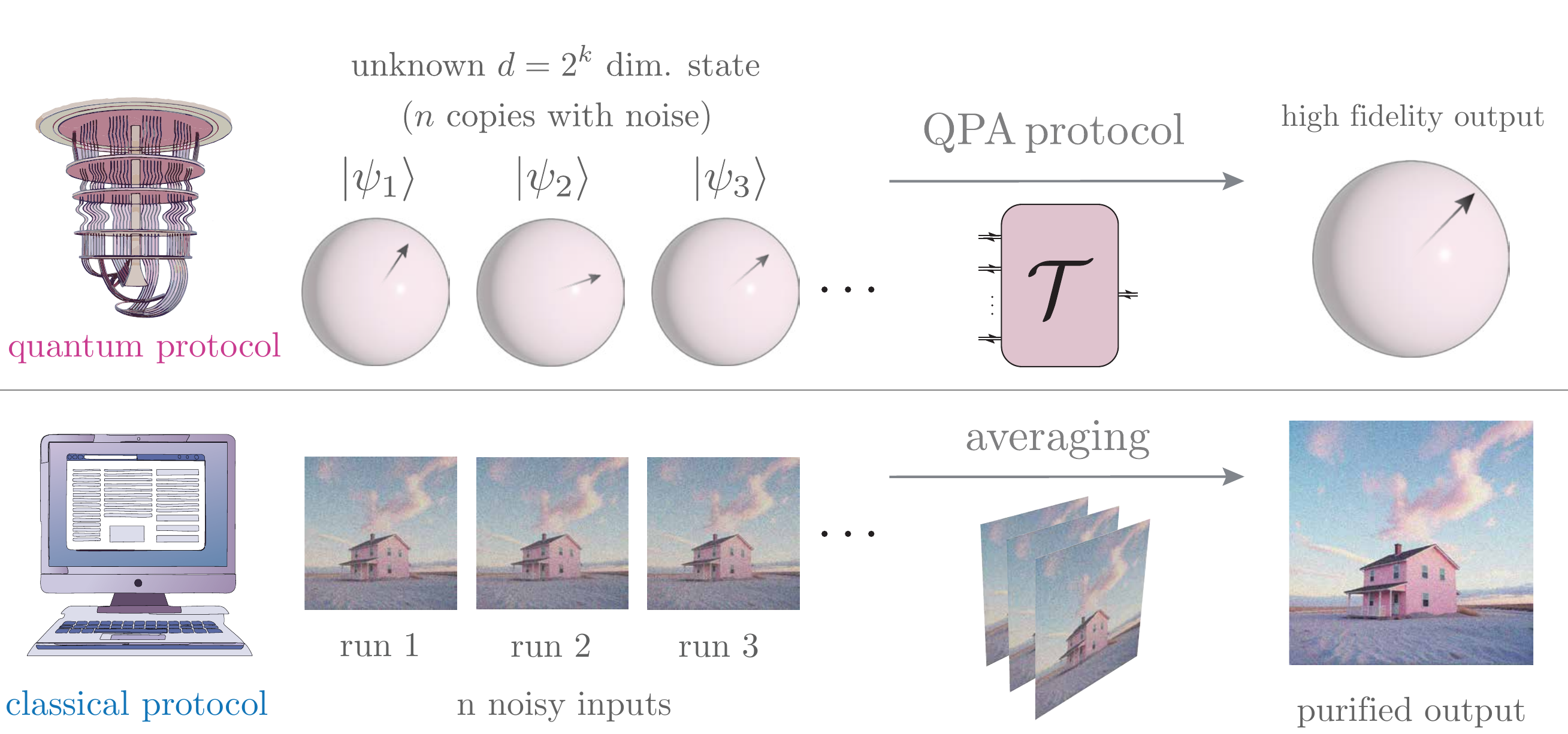"}
\caption{Classical analogy illustrating the QPA protocol $\mathcal{T}$ applied to noisy quantum states. Each of the $n$ copies represents a $d$-dimensional many-body state with low fidelity.
}
\label{fig:illustration}
\end{figure}

Consider a procedure $\mathcal{P}$, a digital or analog algorithm, designed to prepare an unknown quantum state $\ket{\psi}$. However, due to intrinsic noise in quantum operations, the actual output is not the ideal state $\ket{\psi}$ but instead a perturbed state $\ket{\psi_i}$ varying with each run, as illustrated in~\cref{fig:illustration}. To quantify our ignorance, we represent the stochastically corrupted state as a mixed state, $\rho$. Just as overlaying several blurry images sharpens a picture, we apply the QPA protocol $\mathcal{T}$ after multiple runs of $\mathcal{P}$ to synthesize copies of $\rho$. This transforms the inputs which initially lie within a high-dimensional Bloch sphere into a state closer to a target pure state $\sigma = \ketbra{\psi}{\psi}$ on the boundary, which can then be used as a resource for further computations. 

The optimality of a QPA protocol is determined by its ability to recover corrupted states. Fidelity, which compares the output state \(\rho\) to the target pure state \(\sigma\), quantifies this recovery.  To study this, we introduce the following definition:

\begin{definition}[QPA Protocol]  
Let \(\mathcal{R} \subseteq \mathcal{B}(W)\) be a compact set of noisy input states, where \(W = \mathbb{C}^d\). For each \(\rho \in \mathcal{R}\), let \(\sigma\) denote the corresponding ideal pure state. An \((n, \delta)\)-QPA protocol for \(\mathcal{R}\) and $\sigma(\rho)$ is a quantum channel  
\(\mathcal{T}: \mathcal{B}(W^{\otimes n}) \to \mathcal{B}(W)\)  
that takes \(n\) copies of \(\rho\) as input and produces an output state \(\rho'\) such that the fidelity satisfies  
\[
\mathcal{F}(\sigma, \rho') = \Tr(\sigma \rho') \geq 1 - \delta, \quad \forall \rho \in \mathcal{R}.
\]  \end{definition}
\noindent Previous protocols align with this definition, where \(\sigma\) corresponds to coherent states~\cite{AFFJ05} or graph states~\cite{KPDH06}.

As a starting point, we take $\mathcal{R}$ as the set of all depolarized $d$-dimensional pure states, though this condition will be lifted later. Specifically, for any pure state \(\sigma = |\psi\rangle\langle\psi|\) with \(|\psi\rangle \in W\), the depolarized state is given by \(\rho = \mathcal{D}_\lambda(\sigma) = (1 - \lambda)\sigma + \lambda\frac{\mathbb{I}_d}{d}\), where \(\lambda \in (0,1)\) is the depolarization strength. This implies that every \(\rho \in \mathcal{R}\) is associated with a unique pure state \(\sigma\). Taking into account all possible \(\sigma\), we arrive at the figure of merit given by the minimum fidelity~\cite{KW01} \(\min_{\ket{\psi} \in W} \Tr(\sigma \mathcal{T}(\rho^{\otimes n}))\). Importantly, this choice retains generality: twirling can transform any noisy channels into a depolarizing form~\cite{leditzky2017useful}. Moreover, the optimal protocol, defined this way, also optimally amplifies purity toward the principal eigenstate of generic mixed states and applies in the presence of circuit-level noise, as demonstrated later in the paper.

Due to the unitary invariance of the figure of merit, the expression can be reduced to a linear form by averaging \(\sigma\) over a Haar-random prior on \(W\). This allows us to cast the determination of the optimal QPA protocol as an SDP (semidefinite program)~\cite{CGW21}, in which we optimize over $T$, the Choi matrix of the channel $\mathcal{T}$. In particular, $T$ must be normalized (i.e., tracing out the output register yields the identity) and positive semi-definite. We also rewrite the figure of merit using the cost matrix \(C_{\mathrm{out},\mathrm{in}} = \int \sigma^\top_{\mathrm{out}} \otimes \rho^{\otimes n}_{\mathrm{in}}\).
The input and output registers are explicitly labeled as ``in'' and ``out'', though these labels will be omitted when the context is clear. The resulting SDP is:
\begin{equation}
    \label{eq:sdp}
    \begin{aligned}
    \text{find a matrix} \quad & T,\\
\text{maximize} \quad & \mathrm{tr}(C^\top T),\\
\text{subject to} \quad & \mathrm{tr_{out}}T=\mathbb{I}_{\mathrm{in}},\\
& T \succeq 0.
\end{aligned}
\end{equation}

\section{The optimal QPA protocol}
\label{sec:optimal}
We construct the optimal QPA protocol with reversibility in mind, ensuring that information is preserved as much as possible, and discarded only when necessary. The protocol follows three steps: First, we perform Schur sampling to project the input states onto irreducible representations (irreps), which correspond to sectors of different symmetries. Next, unlike previous work that discards certain ``bad states," our approach retains all data by processing every mixed symmetry sector, applying a correction to relocate the best-symmetrized state to the last register. Finally, since only one state can be output, the excess registers are discarded, leaving a purified output state with improved fidelity. This process is illustrated in~\cref{fig:flowchart}.

\begin{figure}[H]
\centering
\includegraphics[scale=0.28]{"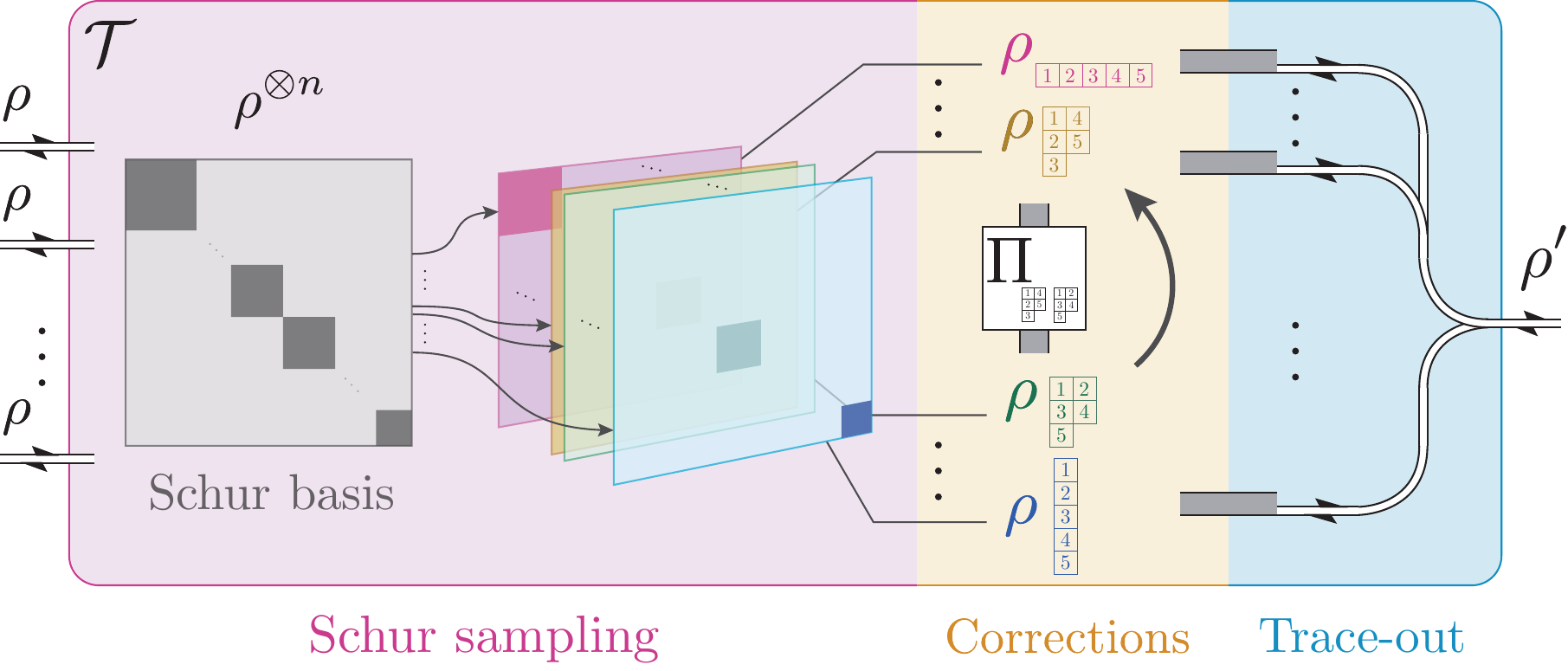"}
\caption{Illustration of the channel $\mathcal{T}$ for the optimal QPA protocol with $n=5$ and $d>5$. The procedure consists of three steps: \textbf{Step 1: Schur sampling.} The state $\rho^{\otimes n}$ attains a block diagonal form (gray) under the Schur basis. The first step involves projections onto irrep subspaces, each indicated by layers in different colors. Here, only four layers are shown for simplicity. \textbf{Step 2: Corrections.} Corrections are applied to arrange the irreps into the column-ordered form. \textbf{Step 3: Trace-out.} Finally, the last register is output as the state $\rho^\prime = \mathcal{T}(\rho^{\otimes 5})$ while all other registers are discarded.}
\label{fig:flowchart}
\end{figure}

In this section, we first analyze the symmetries of the optimal QPA protocol in~\cref{subsec:symmetry} and use them to simplify the SDP. We then define the optimal protocol in~\cref{subsec:optimality} based on the solution of the optimization. In~\cref{subsec:operational}, we relate this solution to the three-step formulation and study its sample complexity scaling.

\subsection{Symmetries of the Protocol}
\label{subsec:symmetry}
Symmetries play a key role in the first two steps of the protocol, specifically the exchange (or permutation) symmetry between the $n$ input registers and the global unitary equivalence symmetry~\cite{F16,BLM+22}, which acts on both input and output registers.

To study those symmetries, we need to introduce Young diagrams (YD), which are finite sequence of natural numbers in decreasing order, and can be represented graphically, for example, $\left[3,2\right]$ is represented as $\Yvcentermath1\scriptsize\yng(3,2)$. From Schur-Weyl duality, we know that under the (active) Schur transform $U_{\mathrm{Sch}}$, the input decomposes into a direct sum of irrep blocks, with each block uniquely labeled by a YD. These YDs are denoted by Greek letters in square brackets, such as $\yd{\varsigma}$ here:
\begin{equation} 
\label{equ:SW}
\rho^{\otimes n}=U_{\mathrm{Sch}}^\dag\left(\bigoplus_{\yd{\varsigma}\vdash n}\rho^{\yd{\varsigma}}\otimes\id_{g^{\yd{\varsigma}}}\right) U_{\mathrm{Sch}},
\end{equation} 
where $\yd{\varsigma}\vdash n$ indicates a valid YD with $n$ boxes and $g^{\yd{\varsigma}}$ refers to the multiplicity of the irrep associated with it. 

Similarly, the cost matrix $C_{\mathrm{out},\mathrm{in}}$ decomposes into a direct sum of $C^{\yd{\varsigma}}_{\mathrm{out},\mathrm{in}}$. Since random permutations on \( T \) can be absorbed into \( C \), we can, without loss of generality, average over them. 
As a result, \( T \) also decomposes accordingly into \(T^{\yd{\varsigma}}\)~\cite{BLM+22}.

Operationally, the block-diagonal form of the input can be viewed as a probability distribution over different irreps, corresponding to various outcomes of Schur sampling. Consequently, the initial step of the optimal QPA protocol always involves Schur sampling~\cite{CM23}. Following this, the respective branch of the optimal QPA protocol, $\mathcal{T}^{\yd{\varsigma}}$, corresponding to the Choi matrix $T^{\yd{\varsigma}}$, is applied to each outcome $\rho^{\yd{\varsigma}}$.
The optimal $T^{\yd{\varsigma}}$ can be obtained from an SDP (see~\cref{equ:subSDP}) similar to~\cref{eq:sdp}.

Additionally, $C^{\yd{\varsigma}}$ exhibits unitary covariance symmetry: it remains unchanged when the input register is transformed according to $\yd{\varsigma}$, and the output register is transformed according to the anti-fundamental representation, denoted by $\overline{\square}$. In other words, $\left[U^\ast_{\mathrm{out}}\otimes U^{\yd{\varsigma}}_{\mathrm{in}},C^{\yd{\varsigma}} \right]=0$ holds for all $U$ of $SU(d)$. Applying the unitaries $U^{\yd{\varsigma}}_{\mathrm CG}$ correspond to the dual Clebsch-Gordan (CG) transform, Schur's Lemma allows us to further decompose $C^{\yd{\varsigma}}$ into a direct sum of constant multiple of identities $\id^{\yd{\mu_i}}$ on the different blocks labeled by $\yd{\mu_i}$. We denote these coefficients by $c_i$:
\begin{align}
    U_{\mathrm{CG}}^{\yd{\varsigma}} C^{\yd{\varsigma}} U_{\mathrm{CG}}^{{\yd{\varsigma}}\dagger} =\bigoplus_{\text{feasible}\,i} C^{\yd{\mu_i}},\text{ and }
    C^{\yd{\mu_i}}=c_i \id^{\yd{\mu_i}}.
    \label{equ:DCG}
\end{align}
In the equation above, $\yd{\mu_i}$ are new sequences obtained by removing the last box on the $i$-th row. If this results in a valid YD, we say that $i$ is feasible. For instance, when $\yd{\varsigma}=\Yvcentermath1\scriptsize\yng(2,2)$, removing a box from the $i=2$-th row gives rise to the valid configuration  $\scriptsize\Yvcentermath1\yng(2,1)$ but removing from the $i=1$-st row does not. Since the Choi Matrix $T^{\yd{\varsigma}}$ shares the same symmetry as $C^{\yd{\varsigma}}$ again by averaging, such decomposition also applies and we denote its constant coefficients by $t_i$. Therefore, the SDP in~\cref{eq:sdp} reduces to a linear program (LP) in which $c_i$ are the coefficients and the parameters $t_i$ are to be optimized~\cite{GO23}.

\subsection{Proof of Optimality}
\label{subsec:optimality}
Next, we provide an explicit solution to this LP, thus characterizing \( T \) by the following~\cref{thm:choi}. We solve the problem by reducing it to a more symmetrized form and leveraging advanced combinatorial results, overcoming challenges unresolved by previous attempts~\cite{GO23}.
\begin{theorem}
    \label{thm:choi}
    The Choi matrix of the optimal QPA protocol is given by $T=U_{\mathrm{mSch}}^\dag\bigoplus_{\yd{\varsigma}}\frac{d^{\yd{\varsigma}}}{d^{\yd{\mu_{i^\ast}}}}\Pi^{\yd{\mu_{i^\ast}}}\otimes \mathbb{I}_{g^{\yd{\varsigma}}}U_{\mathrm{mSch}}$, where $i^\ast$ is the smallest feasible row index, and $
\Pi^{\yd{\mu_{i}}} = \bigoplus_{j < i} \mathbf{0}^{\yd{\mu_j}} \oplus \id^{\yd{\mu_{i}}} \oplus \bigoplus_{j > i} \mathbf{0}^{\yd{\mu_j}}
$ is the projector onto the block corresponding to the irrep $\yd{\mu_{i}}$.
\end{theorem}
Here, we have combined the two transforms together, resulting in the mixed Schur transform~\cite{Ngu23}: \begin{equation}
U_{\mathrm{mSch}}=\left(\bigoplus_{\yd{\varsigma}}U_{\mathrm{CG}}^{\yd{\varsigma}}\otimes \mathbb{I}_{g^{\yd{\varsigma}}}\right) \left(\id_\mathrm{out} \otimes U_{\mathrm{Sch}}\right).
\label{equ:Msch}
\end{equation}

This is to say, the Choi matrix of every branch of the optimal QPA protocol is proportional to a projector on to a certain irrep of the overall symmetry of the system. Additionally, the irrep is always obtained by removing the bottom-right box from $\yd{\varsigma}$. For instance, if $\yd{\varsigma}=\scriptsize\Yvcentermath1\yng(2,2,1)$, the $(2,2)$ box will be removed and $i^\ast=2$, resulting in $\yd{\mu_{i^\ast}}=\scriptsize\Yvcentermath1\yng(2,1,1)$. 

The proof is structured as follows, with the complete proof presented in Supplemental ~\cref{subsec:optimality_supp}: Using a projection argument, we relate $C^{\yd{\varsigma}}$ to a more symmetric operator $C^{\prime\yd{\varsigma}}$, which is a mixed-tensorial representation of $\rho$ averaged over the Haar measure. The new LP coefficients are given by $c_i^\prime$, which are proportional to the normalized Schur polynomial $-S^{\yd{\mu_i}}(\lambda/d, \ldots, \lambda/d, 1 - \lambda + \lambda/d)$ on $d$ variables corresponding to the $d$ eigenvalues of $\rho$. These polynomials are ordered according to the majorization relation between YDs~\cite{S16}. Specifically, $\yd{\mu}$ majorizes $\yd{\nu}$ when for $k \in \{1, \ldots, n\}$, $\sum_{a=1}^k \mu_a \geq \sum_{a=1}^k \nu_a$. Thus, we solve the LP by choosing the YD $\yd{\mu_{i^\ast}}$ corresponding to the largest coefficient $c^\prime_{i^\ast}$. 

With the optimal QPA channel, the figure of merit reaches its optimal value, from which we derive the corresponding optimal sample complexity.

\begin{theorem}
    \label{thm:optimal}
    Asymptotically, the optimal sample complexity required to produce an output with infidelity \(\delta = 1 - \mathcal{F} \) is  
    \[
    n = \frac{1}{\delta} \left( 1 - \frac{1}{d} \right) \frac{\lambda}{(1-\lambda)^2} + O\left(\log\left(\delta^{-1}\right)\right).
    \]
\end{theorem}

In Supplementary~\cref{subsec:fidelity_supp}, we prove this by leveraging the concentration of the Schur sampling probability distribution. For qubits and qutrits, our result matches the known expression~\cite{CEM99,F16}. 
Compared to the protocol proposed in Ref.~\cite{CFLL23}, the sample complexity exhibits an exponential reduction for $\lambda \geq \frac{1}{2}$ and maintains a similar asymptotic scaling in $n$ while achieving improved scaling in $d$ for $\lambda < \frac{1}{2}$.

Our protocol remains asymptotically optimal for generic noisy inputs. Specifically, for any quantum state \( \rho \) with a non-degenerate principal eigenstate, whose eigenvalues satisfy \( p_d > p_{d-1} \geq \cdots \geq p_1 \) without loss of generality, our protocol effectively concentrates the state onto its principal eigenstate, \( \sigma = |\psi_d\rangle\langle\psi_d| \). The scaling of the sample complexity is summarized in the following theorem, which is proven in ~\cref{subsec:gen_inputs}.

\begin{theorem}
    \label{thm:genericoptimal}
    For generic quantum states, the optimal sample complexity required to output a quantum state with infidelity $\delta$
    is given by the asymptotic expression:
    $$n=\frac{1}{\delta}\sum_{i=1}^{d-1}\frac{p_i}{(p_d-p_i)^2}+O(1).$$
\end{theorem}
\noindent 
A similar result was independently shown for the streaming protocol~\cite{grier2025streaming}.  In comparison, our protocol maintains optimal sample complexity, with asymptotic behavior similar to the depolarizing case. We supply the proof in Supplementary~\cref{subsec:gen_inputs}.

\subsection{Operational Interpretation}
\label{subsec:operational}
Now, we relate the Choi matrix $T$ to the three-step protocol previously introduced, to provide a clearer understanding of how $\mathcal{T}$ acts on the input states. To achieve this, we need to introduce some additional mathematical background. Young Tableaux (YTs) are number-filled YDs such that entries in each row and each column are increasing, such as $\Yvcentermath1\scriptsize\young(134,25)$. We label YTs with parenthesized letters, such as $\yt{s}$~\footnote{When referring to a box in a YD, we use the ``matrix'' indexing, starting from top to bottom, then left to right.}. YTs can be used to label the degenerate blocks arising from Schur-Weyl duality, as defined in~\cref{equ:SW}. We use $V^{\yt{s}}$ to denote the irrep subspace of $V^{\otimes n}$ associated with the YT $\yt{s}$. Moreover, for any two YTs, $\yt{m}$ and $\yt{n}$, corresponding to the same YD $\yd{\varsigma}$, there exists a transition operator \( \Pi^{\;\;\yd{\varsigma}}_{\yt{n}\yt{m}} \). That maps the subspace labeled by \( \yt{m} \) isomorphically onto that of \( \yt{n} \)~\cite{Alc18}.

Thus, Step 1 can alternatively be formulated as a projection onto the subspace of $\yd{\varsigma}$,  
\begin{equation}  
\rho^{\otimes n}=\sum_{\yd{\varsigma}\vdash n}\sum_{\yt{s}\vdash\yd{\varsigma}}\rho_{\yt{s}},
\end{equation}  
and Step 2 applies correction operators to transform the irrep's YT into the column-ordered form, as defined in~\cite{AW17}:  
\begin{definition}[Column-ordered YT]  
    A YT is column-ordered if it is filled column-wise from left to right. Given a YD, its corresponding column-ordered YT is denoted with the superscript~$\diamond$.  
\end{definition}
\noindent For instance, consider the YD $\scriptsize\Yvcentermath1\yng(3,2)$. Its column-ordered form is $\left.\scriptsize\Yvcentermath1\yng(3,2)\right.^\diamond=\scriptsize\Yvcentermath1\young(135,24)$.
The alignment follows from the intuition that, for a given Schur sampling outcome, the different registers are not equally symmetrized. In a YT, rows correspond to symmetrization, and columns correspond to antisymmetrization, making the register associated with the bottom-right box the most symmetrized. Thus, the goal is to move the state in this register to the $n$-th register and trace out all other registers. For a detailed proof showing the equivalence of the operational interpretation to the Choi matrix in~\cref{subsec:optimality}, we refer the readers to 
Supplemental~\cref{subsec:operational_supp}. To summarize, the optimal QPA protocol consists of the following three steps:

\begin{enumerate}[label=, leftmargin=0.5em, labelwidth=!,align=left]
    \item \textbf{Step 1: Schur sampling}\\ \hspace*{1.5em} Measure the inputs' symmetry configuration \( \yd{\varsigma} \).
    \item \textbf{Step 2: Corrections}\\ \hspace*{1.5em} Apply a correction \( \Pi^{\;\;\;\yd{\varsigma}}_{\yt{\varsigma^\diamond}\yt{s}} \) conditioned on \( \yd{\varsigma} \). 
    \item \textbf{Step 3: Trace-out} \\\hspace*{1.5em} Return the last register.
\end{enumerate}

\noindent Finally, note that when $d=2$, our optimal QPA protocol reduces to the previously studied optimal protocol for qubits~\cite{CEM99,KW01}.

\section{Implementation}
\label{sec:implementation}
\subsection{Efficient Algorithm}

To practically implement the three-step QPA protocol, we present an efficient algorithm in~\cref{alg:efficient_qpa}, with the corresponding circuit in~\cref{fig:GQPE_QPA}. Note that efficiency is characterized by gate complexity scaling polynomially in \( n \) and logarithmically in \( d \), with every qudit implementable with \(\ceil{\log(d)}\) qubits. 

Intuitively, the algorithm utilizes another control register (labeled ``ctrl") in addition to the data register (labeled ``data") that holds the states being processed. Since we aim to manipulate the symmetry sectors of the data without directly operating on the \( d^n \)-dimensional data register (which would potentially incur a $\text{poly}(d)$ gate cost), we employ GQPE as a ``probe." The GQPE entangles with the data register while transferring its symmetry information into the ctrl register. In Step 1, Schur sampling is achieved by indirectly measuring the YD information in the ctrl register. Subsequently, in Step 2, we reset the ctrl register to the computational basis state encoding the column-ordered YD. Then, we uncompute the GQPE, injecting the updated corrections back into the data register. The correctness of the algorithm is demonstrated in Supplementary~\cref{sec:GQPE-QPA}.

\begin{algorithm}[H]
\caption{Efficient implementation}
\label{alg:efficient_qpa}
\textbf{Registers:} Quantum data register, labeled as ``data", consisting of $n$ qudit sub-registers \( q_1, q_2, \dots, q_n \); quantum ancillae register, labeled as ``ctrl''; classical register \( A \) for storing measurement outcomes.\\
\textbf{Input:} $n$ qudits stored in data.\\
\textbf{Output:} A processed qudit.\\
\textbf{Runtime:} \( 4T_F + 2T_{C\!P}+T_\text{M-P}\), where \( T_F \) is the generalized quantum Fourier transform (GQFT) time, \( T_{C\!P} \) is the controlled permutation time, and \( T_\text{M-P} \) is the measure-and-prepare time.
\begin{algorithmic}[1]
\Statex \textbf{Step 1: Schur sampling}
\State Initialize \text{ctrl} in the trivial irrep state \( \ket{[n,0,\cdots,0]}_{\Lambda} \ket{(12\cdots n)}_L \ket{(12\cdots n)}_R \).\vspace{0.1cm}
\State Apply inverse GQFT \( F^{-1} \).
\State Apply controlled permutation \( C\!P_{\text{ctrl},\text{data}} \).
\State Apply GQFT \( F \).
\State Measure \( \Lambda \) to obtain \( \yd{\varsigma} \) and store result in $A$.

\Statex \textbf{Step 2: Corrections}
\State Reinitialize the $R$ register to \( \ket{(\varsigma^\diamond)}_R\) with ``Prep''.
\State Apply inverse GQFT \( F^{-1} \).
\State Apply inverse controlled permutation \( C\!P^{-1}_{\text{ctrl},\text{data}} \).
\State Apply GQFT \( F \).
\Statex \textbf{Step 3: Trace-out}
\State \Return \( q_n \).
\end{algorithmic}
\end{algorithm}

In terms of gate complexity, the primary dependence on \( d \) arises from the controlled permutations, which scale logarithmically with \( d \)~\cite{H05}. Since GQFT can be implemented in \( \text{poly}(n)\) gates~\cite{Bea97}, GQPE also requires \( \text{poly}(n,\log d) \) gates. The measure-and-prepare step consists of measuring only the \( \Lambda \) register and preparing a specific computational basis state in the \( R \) register, requiring at most \( \text{poly}(n) \) gates. Therefore, the overall gate complexity is \( \text{poly}(n,\log d) \), making it an efficient implementation.

\begin{figure}[H]
\centering
\includegraphics[scale=0.19]{"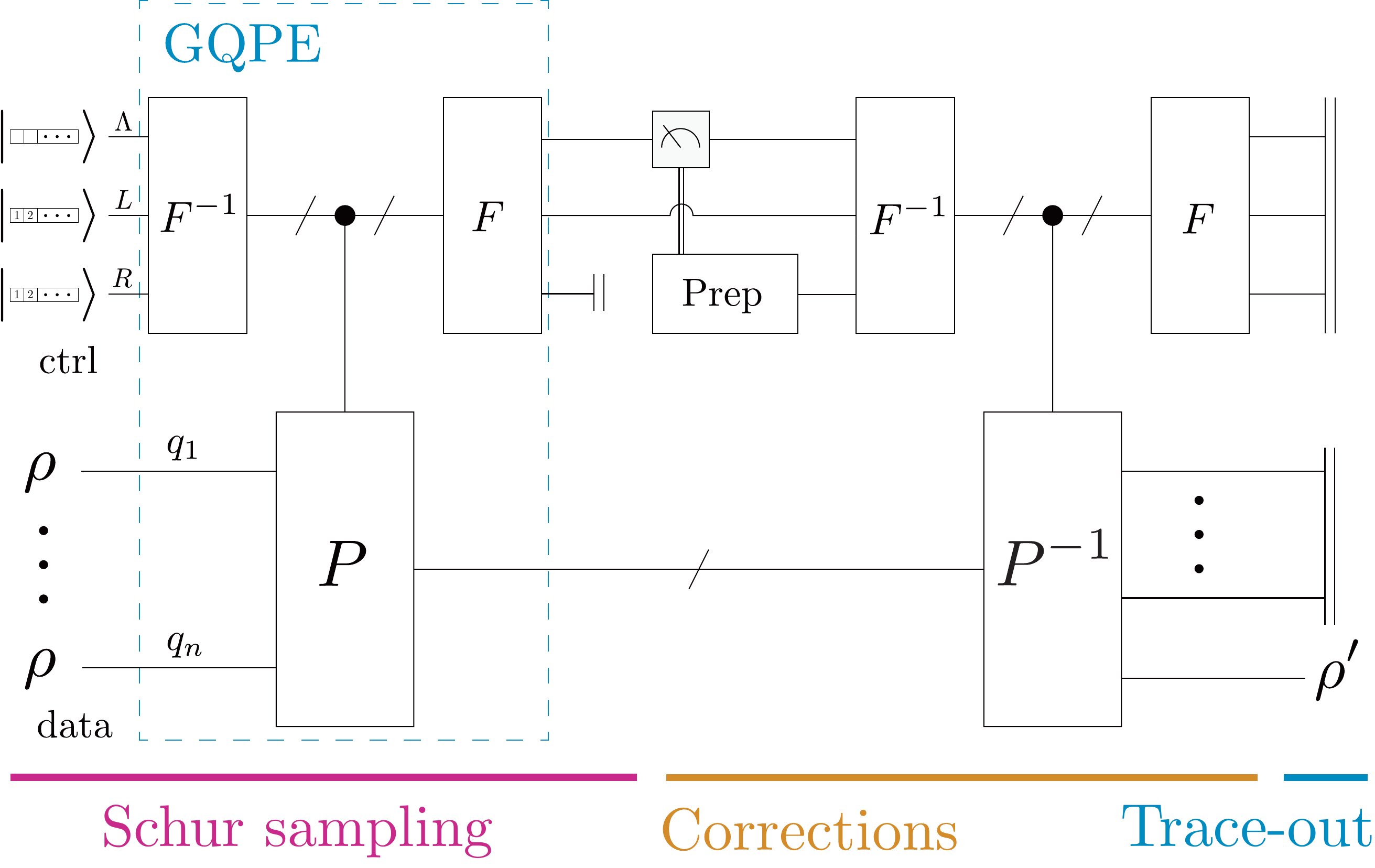"}
\caption{Circuit diagram of~\cref{alg:efficient_qpa}. \textbf{Step 1: Schur sampling} performs weak Schur sampling using GQPE. \textbf{Step 2: Corrections} are done using inverse GQPE. \textbf{Step 3: Trace-out} is performed at the end of the process, the final register with the state $\rho^\prime$ is returned, while all other registers are discarded.}
\label{fig:GQPE_QPA}
\end{figure}

\subsection{SWAPNET}
For NISQ experiments, where the ability to implement deep circuits remains challenging, we developed a low-depth circuit, SWAPNET, which bypasses the need for complex gate compilation. The algorithm is presented in~\cref{alg:SWAPNET}, and its circuit diagram is shown in~\cref{fig:swapnet}a.
\begin{algorithm}[H]
\caption{SWAPNET}
\textbf{Registers:} Quantum data register consisting of three (effective) qudit sub-registers \( q_1, q_2, q_3 \); quantum control register consisting of single ancilla qubit \( q_0 \); classical register \( A \) for storing the measurement outcome $z$.\\
\textbf{Input:} Qudits stored in \( q_1, q_2, q_3 \), number of trials \( N_{\mathrm{trials}} \).\\
\textbf{Output:} A processed qudit.
\begin{algorithmic}[1]
\For{$N = 1$ \textbf{to} $N_{\mathrm{trials}}$}
    \State Initialize ancilla qubit \( q_0 \) in the state \( \ket{0} \).
    \State Perform SWAP test between qudits \( q_1 \) and \( q_2 \) using ancilla \( q_0 \).
    \State Measure ancilla and obtain outcome \( z \).
    \If{$z = 1$}
        \Return \( q_3 \).
    \Else
        \State Swap qudits \( q_2 \) and \( q_3 \).
    \EndIf
\EndFor\\
\Return \( q_3 \).
\end{algorithmic}
\label{alg:SWAPNET}
\end{algorithm}
\begin{figure*}
\begin{center}
\includegraphics[scale=0.11]{"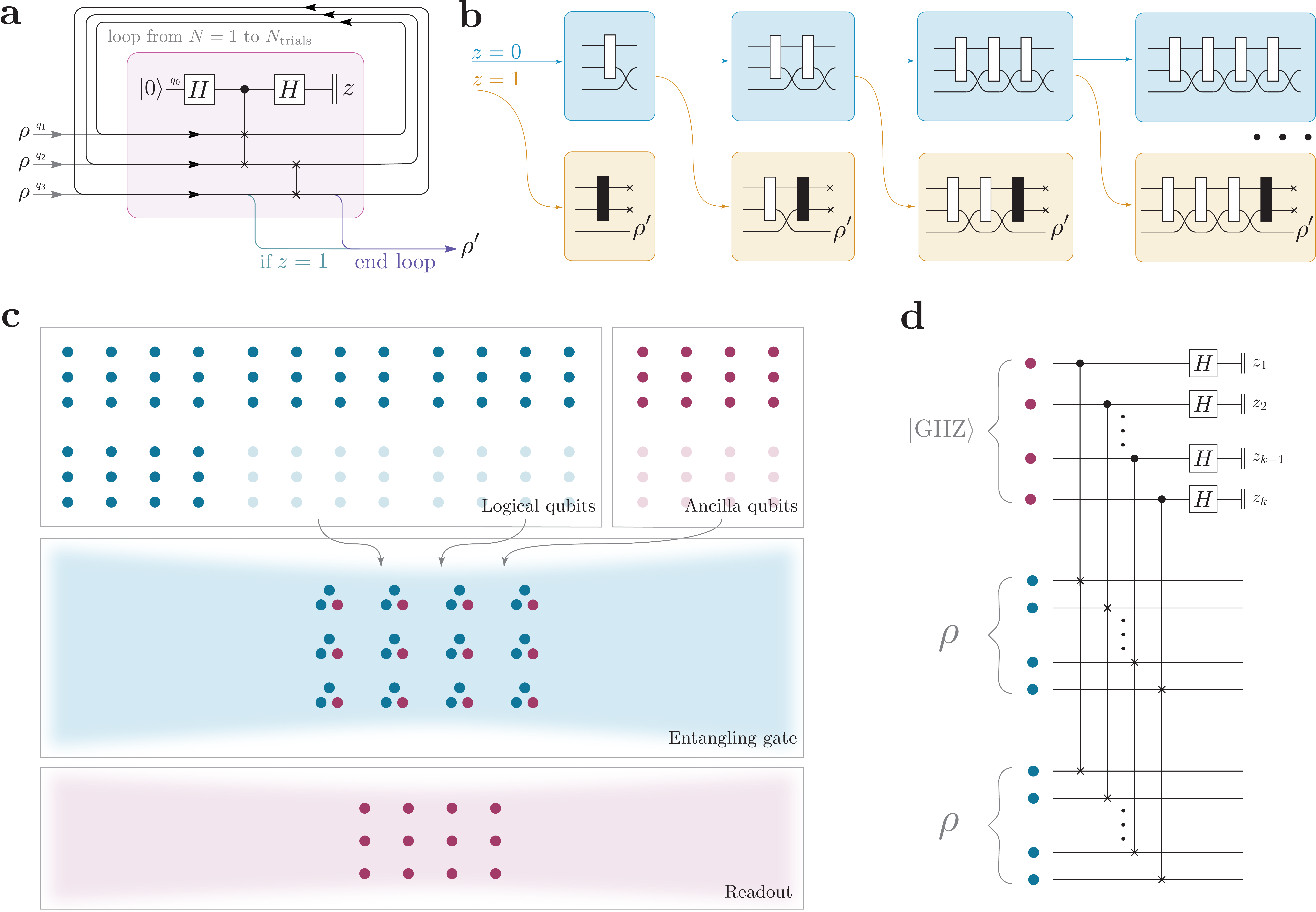"}\vspace{0.5cm}
\caption{{\bf a} Circuit diagram of the SWAPNET algorithm. The ancilla qubit is measured at the end, and based on the outcome $z$, the process either repeats or outputs the purified state $\rho^\prime$. If $z=1$ is never measured, the process terminates after $N_\mathrm{trials}$ iterations, yielding the purified state. {\bf b} Unrolling the circuit over time, the $z=0(1)$ outcomes of the SWAP tests amount to applying (anti)symmetrizers, represented by white(black) boxes. The outcomes, $\rho^\prime$, are labeled at the returned registers, with the traced-out registers crossed out. {\bf c} Physical layout for implementation with Rydberg atoms. Resource qubits (including both data qubits and ancillae) are initially stored in the uppermost zone. They are then transferred to the middle zone, where a 3 qubit CSWAP gate is performed using a global pulse. Afterwards, the ancillae qubits are moved to the lowermost zone for readout. {\bf d} Circuit of the logical CSWAP gate on 3 multiqubit states, achieved by applying transversal CSWAP gates to respective qubits.}
\label{fig:swapnet}
\end{center}
\vspace{-0.2cm}%
\end{figure*}
\noindent SWAPNET is designed to execute the optimal QPA protocol by interlacing a network of SWAP tests. In Supplementary~\cref{subsec:circuit}, we prove the correctness of the algorithm. Intuitively, it provides a reversible generalization of tree-structured SWAP test protocols, up to the trace-out step. 

Experimentally, because our protocol can be implemented transversally by using post-selected high-fidelity GHZ ancillae as a resource, it is particularly well-suited for Rydberg atom arrays~\cite{scholl2023erasure}. In Rydberg quantum simulators, strong interactions are used to model complex systems, with information digitally encoded in the hyperfine levels of the atoms. CSWAP gates can be implemented natively using techniques such as Rydberg anti-blockade or Rydberg pumping~\cite{WWHF+21, SCMW+24}, or with a sequence of CZ and H gates, as outlined in Ref.~\cite{liu2022multi}.
As shown in~\cref{fig:swapnet}c, we follow the zoned architecture of the experimental setup~\cite{BEGL24}. In this setup, suppose the subroutine $\mathcal{P}$ prepares copies of $k$-qubit data states. We first run the subroutine $\mathcal{P}$ in either digital or analog mode, then switch the Rydberg system to its analog mode for subsequent operations. In the entangling zone, the data qubits and the ancillae are brought together and global laser pulses are used to implement the CSWAP gates, as shown in the circuit diagram in~\cref{fig:swapnet}d. Afterward, the ancillae are moved to the read-out zone for measurement. After the process, the data qubits with higher fidelity can be used for subsequent tasks.

\section{Practical Applications}
\label{sec:applications}

\begin{figure*}
\begin{center}
\includegraphics[scale=0.13]{"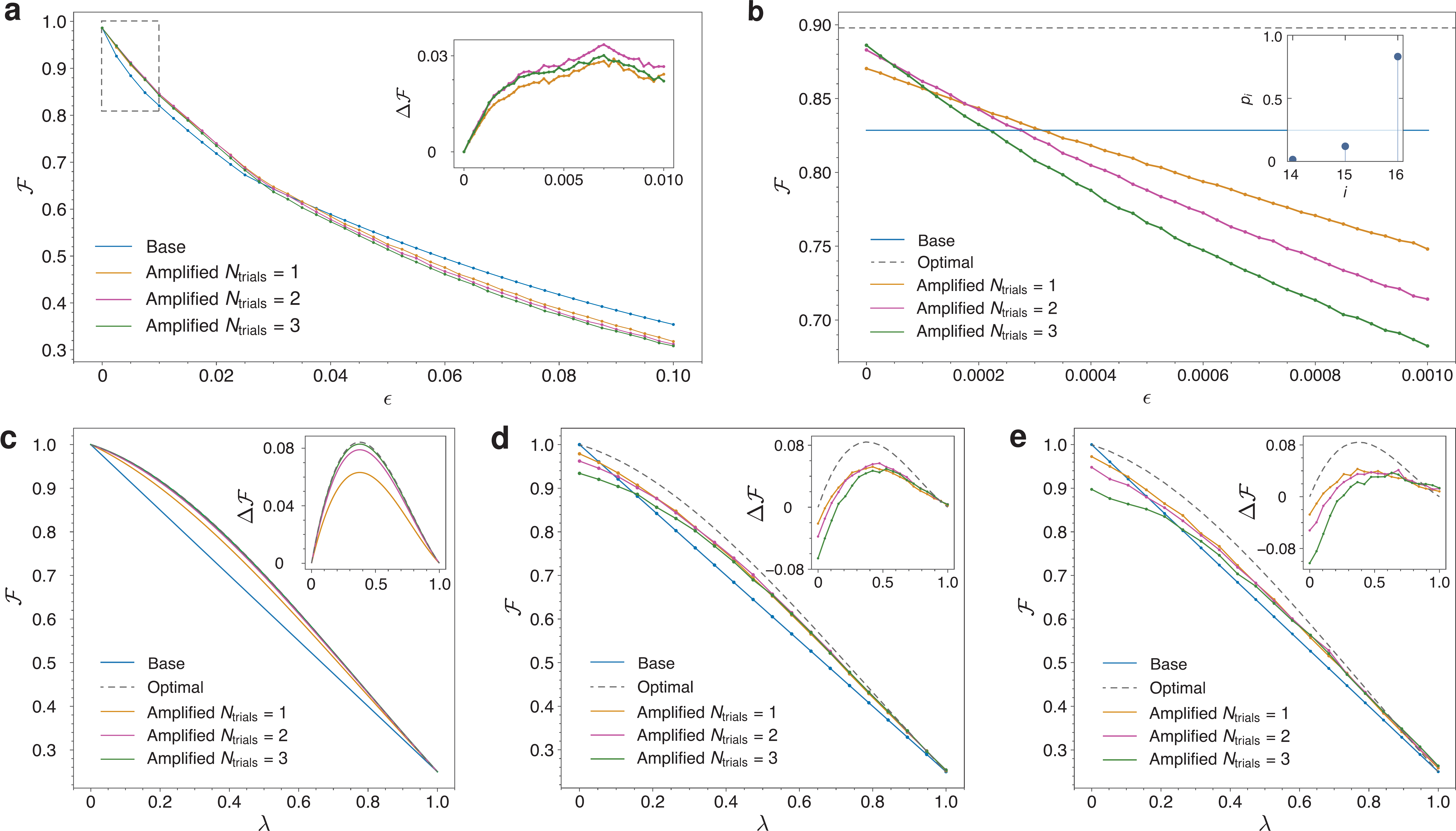"}\vspace{0.5cm}
\caption{{\bf a} Base and amplified fidelities for $N_\mathrm{trials} = 1, 2, 3$ as functions of gate error rate $\epsilon$ from Hamiltonian simulation. Inset shows fidelity gain. {\bf b}
Base, amplified, and optimal fidelities for $N_{\mathrm{trials}}=1,2,3$ as functions of gate error rate $\epsilon$ from adiabatic state preparation. Inset shows the top three eigenvalues of the input state. {\bf c} Base, amplified, and optimal fidelities for $N_{\mathrm{trials}} = 1, 2, 3$ versus depolarization strength $\lambda$ from numerical simulation with global noise only. {\bf d, e} Base and amplified fidelities for $N_{\mathrm{trials}} = 1, 2, 3$ versus  $\lambda$, shown for (d) simulated processor and (e) experimental processor.}
\label{fig:curves}
\end{center}
\vspace{-0.2cm}%
\end{figure*}

We evaluate the performance of QPA across various scenarios using the architecture shown in~\cref{fig:swapnet}, comparing the fidelity gain, defined as the difference between the amplified fidelity (output vs. ideal state) and the base fidelity (input vs. ideal state).

\textbf{Hamiltonian simulation.} Consider a many-body Hamiltonian $H$; our goal is to prepare the time-evolved state $\ket{\psi} = e^{-iHt}\ket{0}$. For illustration, we will choose $H$ to be a transverse-field Ising model on four sites. Simulating the first-order Trotterized state-preparation circuit $\mathcal{P}$ yields the three inputs to SWAPNET. To assess the robustness of QPA to circuit-level noise, we consider a phenomenological noise model, where each single/multi-qubit gate, both in $\mathcal{P}$ and in the QPA circuit, is modeled as a perfect application followed by a single/multi-qubit depolarizing channel with error rate $\epsilon$. The base fidelity is optimized over the number of Trotter steps $N_{\mathrm{Trot}}$ (See \ref{sec:numerical}).

The QPA protocol remains robust (Fig.~\ref{fig:curves}a) with a pseudothreshold emerging near $\epsilon \approx 0.03$, below which fidelity improvements are observed, falling within reach of near-term experimental capabilities. Two iterations already yield significant gains, while additional rounds degrade performance due to the increased noise from deeper circuits.

\textbf{Adiabatic preparation of $Z_2$-ordered entangled state.}
To evaluate the applicability of QPA beyond digital protocols, we benchmark it in the analog setting by applying it to adiabatic state preparation using the native Hamiltonian of Rydberg simulators. The Hamiltonian is adiabatically tuned to prepare the one-dimensional $Z_2$ state with $k = 4$, given by $\frac{1}{\sqrt{2}}\left( \ket{grgr} + \ket{rgrg} \right)$, where $\ket{g}$ and $\ket{r}$ denote the ground and Rydberg states, respectively. The adiabatic ramp profile is optimized to ensure high-fidelity preparation of the input state for the QPA protocol. The noise arises from Lindbladian dynamics, dominated by dephasing, and is obtained via numerical integration of the master equation. QPA is simulated by sampling eigenstates based on the input density matrices. We model post-selected GHZ-state preparation as an ideal GHZ state followed by independently applied single-qubit depolarizing noise with error probability $\epsilon$, and apply measurement errors with the same probability. Three-qubit CSWAP gates are assigned an error probability of $100\epsilon$~\cite{evered2023high}. The parameters and optimization region are chosen to model typical conditions on Rydberg platforms~\cite{ebadi2022quantum}.

The noise in the input states $\rho$ leads to a spectrum in which the three dominant eigenstates, shown in the inset panel of~\cref{fig:curves}b, differ from one another by roughly an order of magnitude. The state is far from depolarized; however, QPA remains effective in this regime, enhancing fidelity for small $\epsilon$ and exhibiting a pseudothreshold around $0.0003$ (see ~\cref{fig:curves}b). The enhancement is attributed to a small bias, as the fidelity between the principal eigenstate and the ideal target state exceeds $0.997$ (see \ref{sec:numerical}). The amplified fidelities asymptotically approach values near the theoretical optimum around $0.90$, and additional improvements are expected with more input copies $n$, validating the applicability of QPA to analog procedures.

\textbf{Experimental state preparation.} We apply QPA to depolarized inputs on a superconducting quantum processor to further validate the theoretical predictions. Global depolarization acting on ququarts ($d=4$) is emulated using stochastically sampled Pauli strings, followed by execution of the transpiled SWAPNET circuit for various \(N_\mathrm{trials}\). Given current hardware's lack of complex classical control, we run separate circuits for all measurement outcomes of the control register and classically combine the results. As a baseline, we include numerical simulations with only global depolarization (\cref{fig:curves}c) and results from a simulated processor excluding time-dependent fluctuations (\cref{fig:curves}d); the experimental results are shown in~\cref{fig:curves}e. The optimal fidelity for $n=3$ is shown as dashed lines (\ref{sec:numerical})).

We observe a fidelity increase across most of the \(\lambda\) range, particularly under strong global noise. The experimental processor data is consistent with both the simulated processor and theoretical predictions. At larger \(\lambda\), the gain from QPA persists despite circuit-level noise, but deviations emerge at smaller \(\lambda\), where deeper circuits increase susceptibility to temporally correlated noise. The effectiveness of shallow circuits is further supported by the observation that just two iterations (\(N_{\mathrm{trials}} = 2\)) are sufficient to closely match the theoretically optimal fidelity for \(n = 3\) at large $\lambda$. The data suggest that QPA can complement QEC by reducing global noise with shallow circuits, particularly in high-noise regimes where QEC is less effective.

\section{Outlook}
In this work, we have established the optimal QPA protocol and demonstrated its efficient implementation. By extending QPA to larger quantum systems with generic noise models and demonstrating its broader applicability through numerical studies, we open new avenues for practical applications in quantum computing, particularly in the NISQ era. Looking ahead, several exciting directions for future research emerge from our findings. 

In Ref.~\cite{KW01}, protocols with multiple outputs were considered, from which a ``rate'' similar to the Shannon theory notion can be defined. 
We conjecture that it may be possible to generalize the optimal QPA protocol to these cases as well. Similar to~\cref{thm:choi}, the Choi matrix projects onto the unique irrep subspace $\yd{\mu_{i^\ast}}$, defined by removing $m$ boxes in a specific order, as detailed in Supplemental~\cref{sec:conjecture}. However, proving optimality might require more advanced tools of representation theory~\cite{Koc07a}, so we leave this as an open question. 

On the experimental front, we aim to demonstrate the feasibility of our protocol using Rydberg quantum simulators, potentially with hybrid analog-digital methods that combine adiabatic evolution for simulation with digital control for implementing QPA~\cite{scholl2023erasure}. Given the platform-agnostic nature of our scheme, the optimal QPA protocol could also be adapted for other quantum computing platforms, such as ions traps, which can potentially implement qudits natively~\cite{ringbauer2022universal}.

Our protocol holds promise in various applications, including quantum sensing~\cite{YEHM+22}, quantum state estimation, and quantum state preparation. Furthermore, QPA also has intimate connection with quantum cloning~\cite{SIGA05,BDEF98}, tomography~\cite{B06,haah2016sample}, quantum spectrum testing~\cite{o2015quantum}. These connections stem from shared underlying mathematical structures, which might be illuminated using similar symmetry properties. As a result, our approach may inspire further developments in related areas. 

\section{Acknowledgement}
The authors thank Adam Wills, Andy Liu, Angus Lowe, Aram Harrow, Debbie Leung, Dmitry Grinko, Elias Theil, Frank Zhang, Guanghao Ye, Hongyue Li, Jiani Fei, Jinzhao Wang, John Martyn, Kaifeng Bu, Lambert Lin, M\={a}ris Ozols, Nazlı Uğur Köylüoğlu, Norah Tan, Patrick Hayden, Quynh Nguyen, Richard Allen, Sa\'ul Pilatowsky-Cameo, Scott Xu, Soonwon Choi, Suvrit Sra, Wenjie Gong, Xiaoyang Shi, Xin Wang, Yuichiro Matsuzaki, Yu-jie Liu, Zi-yin Liu, and Zijing Di for helpful discussions. The authors especially thank Debbie Leung and Soonwon Choi for sharing their manuscripts. This project was supported by the U.S. Department of Energy, Office of Science, National Quantum Information Science Research Centers, Co-design Center for Quantum Advantage (C$^2$QA) under contract number DE-SC0012704, and the US National Science Foundation QLCI program (grant OMA-2016245). 

Source code for the numerical simulations conducted in this paper is available at \url{https://github.com/qoan-projects/OQPA}.

\clearpage
\newpage
\onecolumngrid
\vspace{2cm} %

\begin{center}
    {\Large \textbf{Supplementary Information}}\\[0.5cm]
    {\large for}\\[0.5cm]
    {\large \textbf{Optimal Quantum Purity Amplification}}\\[0.5cm]
    
    \normalsize{Zhaoyi Li,$^{1}$ Honghao Fu,$^{2}$ Takuya Isogawa,$^{3}$ Caio Silva,$^{1}$ and Isaac Chuang$^{1}$}\\[0.5cm]
    
    \footnotesize{\it
    $^{1}$Department of Physics, Massachusetts Institute of Technology, Cambridge, MA 02139, USA \\
    $^{2}$Computer Science and Artificial Intelligence Lab, Massachusetts Institute of Technology, Cambridge, MA 02139, USA \\
    Concordia Institute for Information Systems Engineering, Concordia University, Montreal, QC H3G 1S6, Canada \\
    $^{3}$Department of Nuclear Science and Engineering, Massachusetts Institute of Technology, Cambridge, MA 02139, USA
    }\\
    (Dated: \today)
\end{center}
\beginsupplement

\section{Math Prerequisite}
\subsection{Group Theory}\label{sec:GT}
\subsubsection{Young Diagrams, Young Tableaux, and Weyl Tableaux }
We denote integer compositions by sequences in square brackets. Specifically, a composition $\yd{a_1,a_2,\cdots,a_k}\vdash n$ satisfies $\sum_{i=1}^k a_i =n$, where $k, a_1,\cdots,a_k,n\in\mathbb{N}$. For example, $\yd{1,3,2}\vdash{6}$ represents a composition of 6. 

Young diagrams (YDs) are graphical representations of integer partitions, i.e., compositions that are in non-increasing order. For the composition $\yd{\varsigma_1, \varsigma_2, \cdots, \varsigma_k}$ with $\varsigma_1\geq\cdots\geq\varsigma_k$, the corresponding YD has $k$ rows with $\varsigma_i$ boxes on the $i$-th row. For example, $\yd{3,2}=\Yvcentermath1\scriptsize\yng(3,2)$. We index YDs using square brackets with Greek letter labels such as $\yd{\varsigma}$. 
Moreover, we write $\yd{\varsigma}\vdash n$ to indicate that $\yd{\varsigma}$ is a valid YD with $n$ boxes, as shown in~\cref{fig:YD}. 

\begin{figure}[h]
    \centering
    \begin{tikzpicture}[scale=0.8, every node/.style={minimum size=0.8cm,font=\normalsize}, on grid]
    \foreach \x in {1,...,3}
        \node [draw] at (\x, 3) {};
    \node at (4.05,3) {$\cdots$};
    \node [draw] at (5, 3) {};
    \node [draw] at (6, 3) {};

    \foreach \x in {1,...,3}
    \node [draw] at (\x, 2) {};
    \node at (4.05,2) {$\cdots$};
    \node [draw] at (5, 2) {};
    
    \node at (1,1.1) {$\vdots$};
    \node at (2,1.1) {$\vdots$};

    \foreach \x in {1,...,2}
        \node [draw] at (\x, 0) {};

    \draw [decorate, decoration={calligraphic brace, amplitude=5pt, raise=2pt}, yshift=0pt,line width=0.75pt]
    (0.5,3.5) -- (6.5,3.5) node [black, midway, yshift=0.6cm] { $\varsigma_1$ };

    \draw [decorate, decoration={calligraphic brace, amplitude=5pt, mirror, raise=2pt}, yshift=0pt,line width=0.75pt]
    (0.5,-0.5) -- (2.5,-0.5) node [black, midway, yshift=-0.6cm] { $\varsigma_k$ };

    \draw [decorate, decoration={calligraphic brace, amplitude=5pt, raise=2pt}, xshift=0pt,line width=0.75pt]
    (0.5,-0.5) -- (0.5,3.5) node [black, midway, xshift=-0.6cm, rotate=90] { $k$ };
    \end{tikzpicture}
    \caption{Layout of a YD $\yd{\varsigma}$ with $k$ rows. Here the first row has $\varsigma_1$ boxes and the last row has $\varsigma_k$ boxes.}
    \label{fig:YD}
\end{figure}

(Standard) Young tableaux (YTs) are YDs filled with integers $1$ to $n$ according to specific rules: the numbers must increase strictly down each column and across each row, for instance, $\Yvcentermath1\scriptsize\young(125,34)$. We use parenthesized Latin letters, such as $\yt{s}$, to label YTs. We use \( \yt{s} \vdash \yd{\varsigma} \) to indicate that \( \yt{s} \) is a valid YT of shape \( \yd{\varsigma} \). Moreover, \( \yt{s} \vdash n \) means that there exists a YD \( \yd{\varsigma} \vdash n \) such that \( \yt{s} \vdash \yd{\varsigma} \).

Finally, we introduce Weyl Tableaux (WTs), also known as semi-standard Young Tableaux (SSYTs). Unlike YTs, WTs are filled with integers from $1$ to $d$, allowing repetitions; their entries increase weakly across each row, such as $\Yvcentermath1\scriptsize\young(122,23)$ for $d=3$. We denoted WTs using Latin letters with corner brackets, such as $\wt{m}$. By $\wt{m}\vdash \yd{\varsigma}$, we mean that $\wt{m}$ is a valid WT of the YD $\yd{\varsigma}$.

\subsubsection{Generalized Quantum Fourier Transform}
It turns out that YTs, YDs, and WTs are deeply connected with representation theory, through concepts such as (generalized) Fourier transform (FT) and Schur-Weyl duality.

To introduce the FT, we first look at the group algebra of the symmetric group $S_n$.  Let $S_n=\{g_i\mid i=0,\dots,n!-1\}$ with $g_0=()$ the identity.  The group algebra $A_{S_n}$ is the complex vector space having the $g_i$ as basis vectors, and multiplication is defined by extending the group product linearly:
\begin{equation}
    \left(\sum_i \alpha_i g_i\right) \left(\sum_j \beta_j g_j\right) = \sum_{i,j} \alpha_i \beta_j g_i g_j.
\end{equation}
On this algebra, we define left and right multiplication operators \( L(g) \) and \( R(g) \) for each \( g \in S_n \). The left multiplication is given by the homomorphism \( L:S_n\to \mathcal{L}(A_{S_n}) \), where \( \mathcal{L}(A_{S_n}) \) denotes the space of linear operators on \( A_{S_n} \), while the right multiplication is denoted as \( R(\cdot) \). For $g\in S_n$, $v\in A_{S_n}$, the action of these operators are:
\begin{equation}
L(g)(v) = g v, \quad R(g)(v) = v g^{-1}.
\end{equation}

The domain of these operators extends linearly to the entire algebra, defining representations of \( A_{S_n} \), termed the left(right) natural representations. By associativity, the operators \( L(g) \) and \( R(h) \) commute naturally for all \( g, h \in A_{S_n} \). Schur's lemma allows us to block-diagonalize these operators via FT, which is given by the unitary
\begin{equation}
    F:A_{S_n}\to \bigoplus_{\yd{\varsigma}}W^{\yd{\varsigma}}_L\otimes W^{\yd{\varsigma}}_R, 
\end{equation}
which decomposes the space into symmetry sectors, each labeled by $\yd{\varsigma}$, where $\yd{\varsigma}$ runs over all possible YDs satisfying $\yd{\varsigma}\vdash n$. From the representation theory of the symmetric group, YDs correspond one-to-one with irreducible representations (irreps). The symmetric sector takes on a tensor product structure with two registers, also known as Specht modules, where the left and right natural representations act independently on the \(L\) and \(R\) registers, each forming a copy of the irrep of \( \yd{\varsigma} \) The dimension of each register is given by \( g^{\yd{\varsigma}} \), which also represents the degeneracy of the irrep of the left (right) natural representation, as the basis in \(L\) labels the \( g^{\yd{\varsigma}} \) degenerate irreps in \(R\), and vice versa.

To study this transformation more explicitly, let us first work in the passive transformation picture to find a basis of $A_{S_n}$ on which the operators $R(g)$ are block diagonal. It turns out that this basis is given by the normal idempotents and transition operators. 

To find the normal idempotents, we build the YTs progressively by appending one number at a time, starting from the empty tableau. This process is represented by the Bratteli tree, as shown in~\cref{fig:btree}, which is a hierarchical structure that reflects how numbers are sequentially filled in a YT. In this tree, if $(t)$ is a child node of $(s)$, i.e. $(s)$ is the parent YT of $(t)$, we write $(s)\rightarrow(t)$. 
\begin{figure}[H]
\begin{center}
\includegraphics[scale=0.17]{"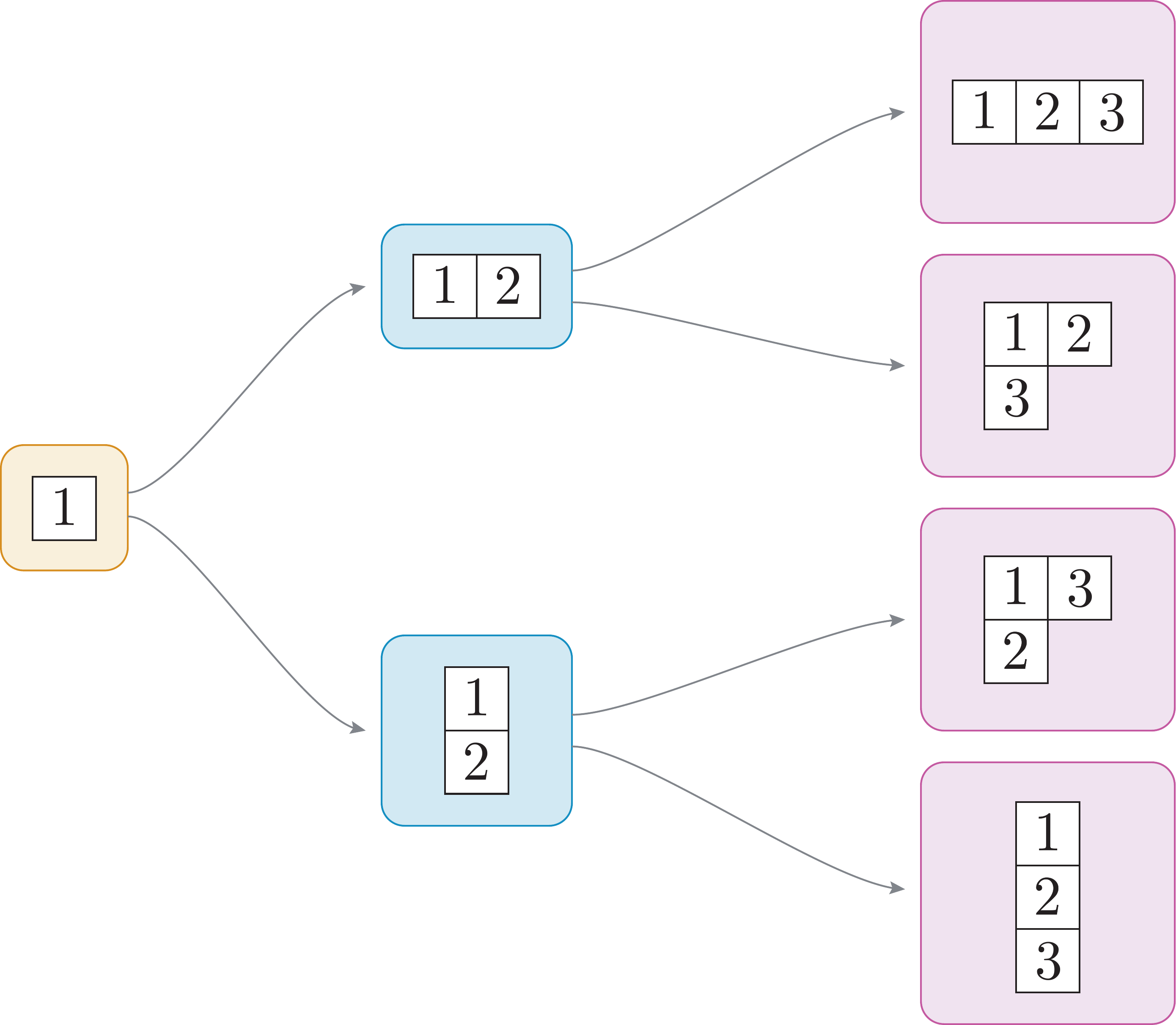"}
\end{center}
\caption{Bratteli tree for $S_3$: starting from the left, we begin with the YTs for $n = 1$, then construct the $n = 2$ YTs by adding one box, and finally reach the $n = 3$ YTs by repeating the process.}
\label{fig:btree}
\end{figure}
The normal idempotents $O_{\yt{t}}$ are defined recursively via Thall's algorithm by following paths in the Bratteli tree~\cite{Thr41}. Specifically, for $\yt{t_i} \to \yt{t_{i+1}}$, we have $O_{\yt{t_{i+1}}} \propto O_{\yt{t_i}} Y_{\yt{t_{i+1}}} O_{\yt{t_i}}$, where $Y_{\yt{t_{i+1}}}$ is the Young symmetrizer~\cite{KS13}. Alternatively, they can be defined using the Measure of Lexical Disorder (MOLD) algorithms~\cite{AW17}. We can extend the idempotents \( O_{\yt{t}} \) to a complete basis of $A_{S_n}$ by introducing the matrix units \( O^{\;\;\yd{\varsigma}}_{\yt{s}\yt{t}} \), where $\yt{s},\yt{t}\vdash\yd{\varsigma}$. When $\yt{s}=\yt{t}$, the matrix units \( O^{\;\;\yd{\varsigma}}_{\yt{t}\yt{t}} \) coincide with \( O_{\yt{t}} \). For each YD $\yd{\varsigma}$, there are $g^{\yd{\varsigma}}$ YTs, so the number of basis vectors $O^{\yd{\varsigma}}_{\yt{s}\yt{t}}$ with $\yt{s}, \yt{t} \vdash \yd{\varsigma}$ is $\left(g^{\yd{\varsigma}}\right)^2$. Summing over all diagrams, the total number of basis vectors is $\sum_{\yd{\varsigma}} \left(g^{\yd{\varsigma}}\right)^2$. For \(\yt{t} \neq \yt{s}\), the matrix unit \( O^{\;\;\yd{\varsigma}}_{\yt{s}\yt{t}} \) act as transition operators mapping \(\yt{t}\) to \(\yt{s}\).

\begin{enumerate}
    \item \textit{Transversality}:  
    \begin{equation} 
    O^{\;\;\yd{\lambda}}_{\yt{s}\yt{t}} O^{\;\;\yd{\mu}}_{\yt{u}\yt{v}} = O^{\;\;\yd{\lambda}}_{\yt{s}\yt{v}} \delta^{\yd{\lambda} \yd{\mu}} \delta_{\yt{t}\yt{u}}.
     \end{equation}
    This ensures that the transition operators function as pipelines, mapping irreps to irreps within the same symmetry sector labeled by $\yd{\lambda}$, while vanishing across non-equivalent irreps.

    \item \textit{Conjugation under the standard inner product}:  
    Given the standard inner product \( \langle \cdot, \cdot \rangle \) on \( A_G \) (where group elements are treated as an orthonormal basis), the transition operators satisfy the adjoint relation:
    \begin{equation}
    O^{\;\;\yd{\lambda} \dagger}_{\yt{s}\yt{t}} = O^{\;\;\yd{\lambda}}_{\yt{t}\yt{s}}.
    \end{equation}
    This means that the transition operators reverses under Hermitian conjugates.

    \item \textit{Orthonormality}:  
    The transition operators are orthonormal in the sense that:
    \begin{equation}
    \left\langle O^{\;\;\yd{\lambda}}_{\yt{s}\yt{t}}, O^{\;\;\yd{\mu}}_{\yt{u}\yt{v}} \right\rangle = \frac{n!}{g^{\yd{\lambda}}} \delta^{\yd{\lambda} \yd{\mu}}\delta_{\yt{s} \yt{u}} \delta_{\yt{t} \yt{v}}.
    \end{equation}
    \end{enumerate}
Two additional properties emerge when considering the normal idempotents of the \( S_{n-1} \) subgroup of \( S_n \):
    \begin{enumerate}
    \item \textit{Compatibility: \(O_{\yt{t}} O_{\yt{s}} = O_{\yt{t}}\) if \( \yt{s} \to \yt{t} \).}
    \item {\it Resolution of identity}: The parent projector can be written as a sum of child projectors:  
    \[
    O_{\yt{s}} = \sum_{\yt{t} \colon \yt{s} \to \yt{t}} O_{\yt{t}}.
    \]  
\end{enumerate}
The third property, along with the fact that there are $\sum_{\yd{\varsigma}} g^{\yd{\varsigma}\ 2} = n!$ such operators in total, allows us to interpret the transition operators as a basis in which the image of $R$ and $L$ are mutually block diagonalized. More explicitly, for any $g \in S_n$, we decompose it as 

\begin{equation}
g = \sum_{\yt{s}, \yt{t}} \sum_{\yd{\varsigma}} A^{\yt{s} \yt{t}}_{\;\;\yd{\varsigma}}(g) O^{\;\;\yd{\varsigma}}_{\yt{s} \yt{t}}.
\end{equation} 

The components $A^{\yt{s} \yt{t}}_{\;\;\yd{\varsigma}}(g)$ form a unitary representation of \( g \), following from transversality of the basis. More explicitly, we use \(\yt{s}\) and \(\yt{t}\) to index registers \(L\) and \(R\), respectively, and flatten them to define the first and second levels of the representation matrix. This way, each degenerate subspace \( W^{\yd{\varsigma}}_L \) is spanned by fixing \( \yt{t} \) and varying \( \yt{s} \) in \( O^{\;\;\yd{\lambda}}_{\yt{s} \yt{t}} \), forming the Young-Yamanouchi (YY) basis~\cite{Yamanouchi1937,de2012subgroup}. Typically, the YY basis refers to a single Specht module's basis, which is identified abstractly with YTs rather than as an element of \( A_{S_n} \), due to the shared block structure of all modules. The right multiplication operator \( R(g) \) then takes a block-diagonal form:
\begin{equation}
\begin{bmatrix}
\underbrace{
    \begin{bmatrix}
        A^{\yt{s} \yt{t}}_{\;\;\yd{\lambda}}(g) & & \\
        & \ddots & \\
        & & A^{\yt{s} \yt{t}}_{\;\;\yd{\lambda}}(g)
    \end{bmatrix}
}_{g^{\yd{\lambda}} \text{ repetitions}} & & \\
& \ddots & \\
& & \underbrace{
    \begin{bmatrix}
        A^{\yt{s} \yt{t}}_{\;\;\yd{\mu}}(g) & & \\
        & \ddots & \\
        & & A^{\yt{s} \yt{t}}_{\;\;\yd{\mu}}(g)
    \end{bmatrix}
}_{g^{\yd{\mu}} \text{ repetitions}}
\end{bmatrix},
\label{equ:block_diagonal}
\end{equation}
where each block $A^{\yt{s} \yt{t}}_{\;\;\yd{\varsigma}}(g)$ appears $g^{\yd{\varsigma}}$ times, corresponding to different irreducible representations of $\yd{\varsigma}$. Vice versa, the left multiplication operator \( L(g) \) acts only on the \(L\) register, leaving the \(R\) register unchanged, thereby leading to the decomposition summarized in Table~\ref{tab:fourier_correspondence}.

\begin{table}[H]
    \centering
    \caption{Roles of left and right natural representations in the Fourier transform.}
    \label{tab:fourier_correspondence}
    \renewcommand{\arraystretch}{1.8} %
    \begin{tabular}{|c|c|c|c|}
        \hline
        Representation  & Register & Basis & Action on $A_G$   \\
        \hline
        Left natural   & $W^{\yd{\varsigma}}_L$ & $O^{\;\;\yd{\varsigma}}_{\yt{s} \yt{\cdot}}$ & $L(g), g \in G$  \\
        \hline
        Right natural  & $W^{\yd{\varsigma}}_R$ & $O^{\;\;\yd{\varsigma}}_{\yt{\cdot} \yt{t}}$ & $R(g), g \in G$  \\
        \hline
    \end{tabular}
\end{table}

To implement this transformation on a standard qubit-based system, we introduce an additional register \(\Lambda\) to encode \(\yd{\varsigma}\), alongside the two registers \( L \) and \( R \), which are padded sufficiently to accommodate the dimensions of \( g^{\yd{\varsigma}} \). This allows us to encode the normalized basis states  $\sqrt{\frac{g^{\yd{\varsigma}}}{n!}} O^{\;\;\yd{\varsigma}}_{\yt{s} \yt{t}}
$
as  
\begin{equation}\ket{\substack{\yd{\varsigma} \\ \yt{s} \yt{t}}}  = \ket{\yd{\varsigma}}_{\Lambda} \ket{\yt{s}}_L \ket{\yt{t}}_R.
\end{equation}
Moreover, an active transformation, realized as a unitary operation, is necessary to physically implement the mapping. Suppose we have an encoding of \( g \in G \), representing the group elements as \( \ket{g} \) across the entire ctrl register, we define the generalized quantum Fourier Transform (GQFT) operator (where we slightly overload the notation of \( F \) to represent an operator acting on quantum states) to act as  
\begin{equation}
F:\ket{g} \mapsto \sum_{ \yd{\varsigma}}\sum_{\yt{s}, \yt{t}\vdash\yd{\varsigma}} F^{\yt{s} \yt{t}}_{\;\;\yd{\varsigma}}(g) \ket{\substack{\yd{\varsigma} \\ \yt{s} \yt{t}}}, 
\quad
F^{-1}: \ket{\substack{\yd{\varsigma} \\ \yt{s} \yt{t}}} \mapsto \sum_{g \in G} F^{\yt{s} \yt{t}}_{\;\;\yd{\varsigma}}(g) \ket{g},
\end{equation}
whose coefficients are defined as $F^{\yt{s} \yt{t}}_{\;\;\yd{\varsigma}}(g)=\sqrt{\frac{n!}{g^{\yd{\varsigma}}}}A^{\yt{s} \yt{t}}_{\;\;\yd{\varsigma}}(g)$. Under conjugation by GQFT, the operator $ FR(g)F^{-1}$ assumes the block diagonal form given in equation~\cref{equ:block_diagonal}. Finally, since ancillae are used to pad the Hilbert space to fit within a qubit-based system, the operator \( F \) can be freely defined on the unencoded states.

\subsubsection{Schur-Weyl Duality, Schur Transform, and Brattelli Tree}
\label{subsubsec:Schur}
Besides the duality between the natural representations of $S_n$, an analogous relation between permutation and global operations arises in the data register, known as Schur-Weyl duality. This duality links the actions of the symmetric group $S_n$ and the general linear group $GL(d)$ for $n$ qudits, and it has significant applications in quantum information theory.
Let $W=\mathbb{C}^d$ represent a $d$-dimensional state (qudit) space, and $W^{\otimes n}$ is the state space for a system of $n$ qudits. Consider the action of the symmetric group $S_n$ on $W^{\otimes n}$. Let $\pi \in S_n$ represent a permutation, and let $P_\pi$ denote the corresponding permutation operator that reorders the registers according to $\pi$. That is, for the state $\ket{\Psi} \in W^{\otimes n}$, its components in the computational basis transform as $P_\pi \Psi_{i_1, \dots, i_n} = \Psi_{i_{\pi(1)}, \dots, i_{\pi(n)}}$, with $i_k \in \{0, \dots, d-1\}$.
Consider a qudit operator $A$, or an element of $GL(d)$ in the fundamental representation. The permutation commutes with global transformations of the form $A^{\otimes n}$, which apply the same operator $A$ to each register, i.e. $[P_\pi, A^{\otimes n}] = 0$. Also note that due to linearity, this commutation relation can be extended to the group algebra $A_{S_n}$ of $S_n$, which consists of linear combinations of the group elements. Similarly, the global transformations can be extended to arbitrary exchange invariant operators.

Analogous to GQFT, Schur-Weyl duality states that the space $W^{\otimes n}$ can be decomposed using a unitary transformation known as the Schur transform. This transformation maps the tensor product space into a direct sum of tensor products of irreps spaces of $GL(d)$ and $S_n$ labeled by all possible YDs that satisfy $\yd{\varsigma}\vdash n$.
\begin{equation}
\label{equ:schur}
U_{\mathrm{Sch}}: W^{\otimes n} \rightarrow \bigoplus_{\yd{\varsigma}} W^{\yd{\varsigma}}_W \otimes V^{\yd{\varsigma}}_S,
\end{equation}
where $U_{\mathrm{Sch}}$ is the active Schur transform unitary, $W^{\yd{\varsigma}}$ denotes the Weyl module corresponding to the YD ${\yd{\varsigma}}$ (the space the $GL(d)$ representation acts on), and $V^{\yd{\varsigma}}$ denotes the Specht module (the space the $S_n$ representation acts on). 
We denote the dimension of $W^{\yd{\varsigma}}$ by $d^{\yd{\varsigma}}$ and the dimension of $V^{\yd{\varsigma}}$ by $g^{\yd{\varsigma}}$. This duality extends linearly beyond \( GL(d) \) to include all qudit operators. For a certain qudit operator \( A \), we denote its representation matrix by \( A^{\yd{\varsigma}} \). For instance, when dealing with density matrices, we will refer to $\rho^{\yd{\varsigma}}$ to indicate the density matrix associated with the irrep labeled by $\yd{\varsigma}$. 

If we instead interpret the Schur transform matrix as a passive transformation, it represents a change of basis from the computational basis to the Schur basis. 
The Schur basis is constructed following Ref.~\cite{BCH06}, though we do not explicitly derive it as a basis of \( W^{\otimes n} \) in this paper. However, just as the YY basis is defined on a single Specht module, the basis of a single Weyl module \( W^{\yd{\varsigma}} \) can be identified analogously with WTs. For every YD $\yd{\varsigma}$, such as the one in~\cref{fig:YD}, each valid WT $\wt{m}$ defines a basis vector $\ket{\wt{m}}$. This defines an orthonormal basis of $V^{\yd{\varsigma}}$, on which the representation of an operator $\rho$ takes the standard form $\rho^{\yd{\varsigma}}$. An equivalent way to represent WTs is through Gel'fand-Tsetlin (GT) patterns, which are triangular arrays of numbers that satisfy the betweenness property, where each entry lies between two adjacent entries in the row above, 
such as in~\cref{fig:GT}. Since $\wt{m}\vdash\yd{\varsigma}$, the first row of the GT pattern is a partition corresponding to $\yd{\varsigma}$. By using YTs and WTs together, we can label the Schur basis uniquely.

\begin{figure}

  \begin{tikzpicture}[scale=0.8, every node/.style={font=\normalsize}]

\draw[rounded corners=10pt] (-1.7,0) -- (5.7,0) -- (2,-5.8) -- cycle;

\node at (-0.8,-0.5) {$\varsigma_1$};
\node at (0.5,-0.5) {$\cdots$};
\node at (2.5,-0.5) {$\varsigma_{d-1}$};
\node at (4.8,-0.5) {$\varsigma_d$};

\node at (0.4,-1.5) {$m_{1,d-1}$};
\node at (1.7,-1.5) {$\cdots$};
\node at (3.4,-1.5) {$m_{d-1,d-1}$};

\node at (2,-2.7) {$\vdots$};
\node at (2,-4.5) {$m_{1,1}$};
\end{tikzpicture}
\caption{Layout of a possible GT pattern $\wt{m}$ of the YD $\yd{\varsigma}$.}
\label{fig:GT}
\end{figure}

As an example, consider the case with three qudits. By Schur-Weyl duality, we can break down $W^{\otimes 3}$ into a direct sum as
\begin{equation}
\label{equ:example_yd_decomp}
    W^{\otimes 3}\cong W^{\Yvcentermath1\tiny\yng(3)}\oplus W^{\Yvcentermath1\tiny\yng(2,1)}\otimes V^{\Yvcentermath1\tiny\yng(2,1)}\oplus W^{\Yvcentermath1\tiny\yng(1,1,1)}.
\end{equation} 
However, there are times when we want to study the irrep spaces of $GL(n)$ as subspaces of $W^{\otimes 3}$, including their degeneracies. In such cases, it is more convenient to use YTs rather than YDs. By expanding the Specht module in the YY basis, we obtain four invariant subspaces of exchange-invariant operators:
\begin{equation}
\label{equ:example_yt_decomp}
W^{\otimes 3}=W_{\Yvcentermath1\tiny\young(123)}+W_{\Yvcentermath1\tiny\young(12,3)}+ W_{\Yvcentermath1\tiny\young(13,2)}+ W_{\Yvcentermath1\tiny\young(1,2,3)} = \bigoplus_{\yt{s}:\yt{s}\vdash3}W_{\yt{s}}.
\end{equation}
Here, each term $W_{\yt{s}}$ represents a subspace of $W^{\otimes 3}$ associated with a specific YT $\yt{s}$ with three boxes. More explicitly, these subspaces are defined as the images of the projectors \( \Pi_{(s)} \) acting on the tensor product space \( W^{\otimes 3} \), which are representations of \( O_{\yt{s}} \) on $W^{\otimes n}$. These projectors possess another important property:  

\begin{enumerate}  
    \item {\it Tracing Property}: If \( \yt{t} \) consists of \( n \) registers, then tracing out the \( n \)-th register from the projector \( \Pi_{\yt{s}} \) yields the parent projector, scaled by the ratio of their dimensions:  
    \begin{equation}  
        \Tr_{n}(\Pi_{\yt{s}}) = \frac{d^{\yt{s}}}{d^{\yt{t}}} \Pi_{\yt{t}}.  
    \end{equation}  
\end{enumerate}  
Here, \( d^{\yt{\cdot}} \) denote the dimension of the irrep associated with the YD $\yt{\cdot}$. For simplicity of notation, we also denote the restrictions of tensor product operators on those invariant subspaces with a subscript. For instance, $\rho^{\otimes n}=\sum_{(s):(s)\Vdash n}\rho_{(s)}$ and $\rho_{\yt{s}}=\Pi_{\yt{s}}\rho^{\otimes n}\Pi_{\yt{s}}$. Moreover, since $\rho^{\otimes n}$ is already block-diagonal in the Schur basis, we can obtain the restriction to each block using only single-sided projections, $\Pi_{\yt{s}}\rho^{\otimes n}=\rho^{\otimes n}\Pi_{\yt{s}}$. The reader should be attentive to the difference between $\rho^{\yd{\varsigma}}$ and $\rho_{\yt{s}}$ here: the former is a $d^{\yd{\varsigma}}$-dimensional representation in its standard form, whereas the latter is an operator living in the $d^n$-dimensional tensor product space.
The Schur transform further clarifies the relation between these two notations~\footnote{We will omit the condition $\yt{s}\vdash n$ in the following text whenever $\yt{s}$ is summed over.}:
\begin{equation}
\rho^{\otimes n}=U_{\mathrm{Sch}}^\dag\bigoplus_{\yd{\varsigma}}\rho^{\yd{\varsigma}}\otimes\id_{g^{\yd{\varsigma}}} U_{\mathrm{Sch}}=\sum_{\yt{s}}\rho_{\yt{s}}.\end{equation}
Finally, we summarize the properties of YTs and WTs in the table below:

\begin{table}[H]
    \centering
    \caption{Summary of the properties of YTs and WTs and their roles in the duality.}
    \label{tab:correspondence}
    \begin{tabular}{|c|c|c|c|c|c|c|}
        \hline
        Tableaux & Group & correspondence  & Irrep dim. &  Basis name & Basis index & Action on \( W^{\otimes n} \) \\
        \hline
         YT  & \( S_n \)  & Specht Module \( W^{\yd{\varsigma}} \) & $g^{\yd{\varsigma}}$ & YY & \( \yt{t} \) & Permutation \( P_{\pi} \), \( \pi \in S_n \) \\
        \hline
         WT & \( GL(d) \) & Weyl Module \( V^{\yd{\varsigma}} \) & $d^{\yd{\varsigma}}$ & GT & \( \wt{m} \) & Global transformation \( A^{\otimes n} \), \( A \in GL(d) \) \\
        \hline
    \end{tabular}
\end{table}

\subsubsection{Example With Qutrits}
Now let us explicitly work out the example in~\cref{equ:example_yd_decomp,equ:example_yt_decomp} for qutrits, i.e., $d=3$. In high-energy literature, the decomposition in~\cref{equ:example_yd_decomp} is typically written as $\mathbf{3}^{\otimes 3}=\mathbf{10}\oplus\mathbf{8}^{\oplus 2}\oplus\mathbf{1}$, with the Weyl modules (blocks) labeled by their corresponding dimensions. The computational basis states can be written as $\ket{ijk}$, where $i, j, k \in \{0, 1, 2\}$. The Schur basis corresponding to the YT $\Yvcentermath1\scriptsize\young(12,3)$ is listed in~\cref{tab:basis}. The full Schur basis is displayed as the columns of the inverse Schur transform matrix $U^\dagger_{\mathrm{Sch}}$ in~\cref{subfig:Bplot}. A randomly generated qutrit mixed state $\rho$ is shown in the Schur basis in~\cref{subfig:matplot}, revealing four distinct blocks, each associated with a specific YT, such as $\Yvcentermath1\scriptsize\young(123)$. Notably, the two degenerate blocks, $\Yvcentermath1\scriptsize\young(12,3)$ and $\Yvcentermath1\scriptsize\young(13,2)$, both corresponding to the YD $\Yvcentermath1\scriptsize\yng(2,1)$, share the same form $\rho^{\,\Yvcentermath1\tiny\yng(2,1)}$.

\begin{table}[ht]
\caption{Table showing the correspondence between WTs and the Schur basis states, expressed in terms of computational basis states.}
\label{tab:basis}
\begin{center}
\begin{tabular}{|c|m{10cm}|}
\hline
\textbf{SSYT} & \centering\textbf{Basis State} \\ \hline

$ {\Yvcentermath1\scriptsize\young(11,2)} $ & 
\raisebox{-.75em}{\rule{0pt}{1.5em}} $ \sqrt{\frac{2}{3}} \ket{001} - \frac{1}{\sqrt{6}} \ket{010} - \frac{1}{\sqrt{6}} \ket{100} $ \\ \hline

$ {\Yvcentermath1\scriptsize\young(12,2)} $ & 
\raisebox{-.75em}{\rule{0pt}{1.5em}} $ \frac{1}{\sqrt{6}} \ket{011} + \frac{1}{\sqrt{6}} \ket{101}-\sqrt{\frac{2}{3}} \ket{110} $ \\ \hline

$ {\Yvcentermath1\scriptsize\young(13,2)} $ & 
\raisebox{-.75em}{\rule{0pt}{1.5em}} $ \sqrt{\frac{1}{3}} \ket{012} - \frac{1}{2\sqrt{3}} \ket{021} + \frac{1}{\sqrt{3}} \ket{102}  - \frac{1}{2\sqrt{3}} \ket{120}  - \frac{1}{2\sqrt{3}} \ket{201} - \frac{1}{2\sqrt{3}} \ket{210}  $ \\ \hline

$ {\Yvcentermath1\scriptsize\young(11,3)} $ & 
\raisebox{-.75em}{\rule{0pt}{1.5em}} $ \sqrt{\frac{2}{3}} \ket{002} - \frac{1}{\sqrt{6}} \ket{020}- \frac{1}{\sqrt{6}} \ket{200} $ \\ \hline

$ {\Yvcentermath1\scriptsize\young(12,3)} $ & 
\raisebox{-.75em}{\rule{0pt}{1.5em}} $ \frac{1}{2} \ket{021} - \frac{1}{2} \ket{120}+\frac{1}{2} \ket{201} - \frac{1}{2} \ket{210} $ \\ \hline

$ {\Yvcentermath1\scriptsize\young(22,3)} $ & 
\raisebox{-.75em}{\rule{0pt}{1.5em}} $ \sqrt{\frac{2}{3}} \ket{112} - \frac{1}{\sqrt{6}} \ket{121} - \frac{1}{\sqrt{6}} \ket{211} $ \\ \hline

$ {\Yvcentermath1\scriptsize\young(13,3)} $ & 
\raisebox{-.75em}{\rule{0pt}{1.5em}} $ \frac{1}{\sqrt{6}} \ket{022} + \frac{1}{\sqrt{6}} \ket{202} - \sqrt{\frac{2}{3}} \ket{220} $ \\ \hline

$ {\Yvcentermath1\scriptsize\young(23,3)} $ & 
\raisebox{-.75em}{\rule{0pt}{1.5em}} $  \frac{1}{\sqrt{6}} \ket{122} + \frac{1}{\sqrt{6}} \ket{212}  -\sqrt{\frac{2}{3}}\ket{221} $ \\ \hline
\end{tabular}
\end{center}
\end{table}

\begin{figure}[h]
\begin{subfigure}[t]{0.47\columnwidth}
\begin{center}
\includegraphics[scale=0.6]{"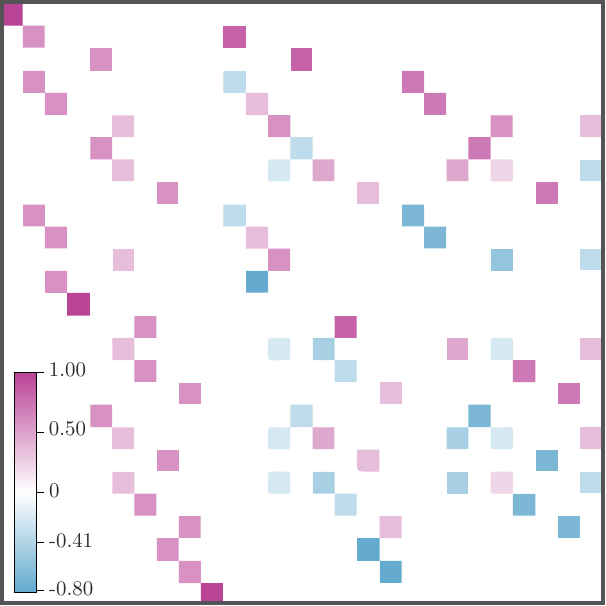"}
\end{center}
\caption{Matrix representation for $ U_{\mathrm{Sch}}^\dagger$. Columns correspond to Schur basis vectors expressed in the computational basis.} 
\label{subfig:Bplot}
\end{subfigure}%
\hfill
\begin{subfigure}[t]{0.47\columnwidth}
\begin{center}
\includegraphics[scale=0.6]{"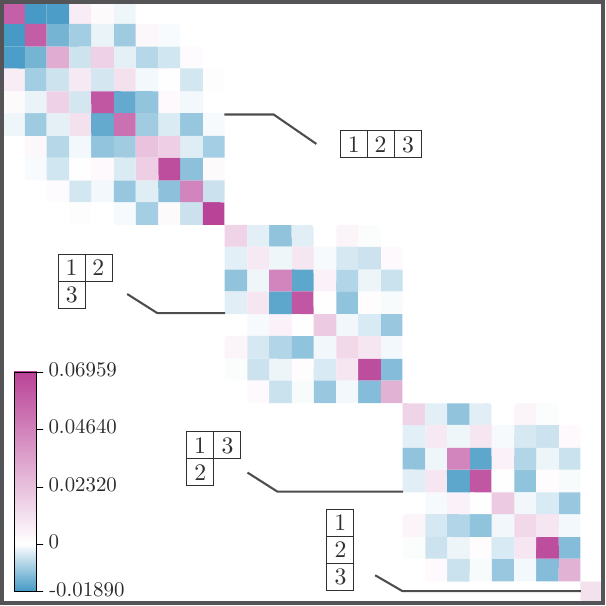"}
\end{center}
\caption{Visualization of $\rho^{\otimes 3}$ under the Schur basis for a randomly generated qutrit mixed state $\rho$. In this case there are four blocks, each corresponding to a different YT.}
\label{subfig:matplot}%
\end{subfigure}
\end{figure}

\subsubsection{Mixed Schur Transform}
Note that all the above concept can be generalized to the case when the representation of $GL(d)$ is a tensor product of both fundamental and anti-fundamental representations, namely $\rho^{\otimes n}\otimes \left(\rho^{\top -1}\right)^{\otimes m}$. In this case, we can define generalized (staircase) YDs, YTs, and WTs, which are characterized by the two parameters $n$ and $m$. Similarly, we write $\yd{\varsigma}\vdash(n,m)$ to indicate that $\yd{\varsigma}$ is a valid YT for this configuration. 

Here, we outline only the key information needed to complete our proof; for a more detailed discussion, see Ref.~\cite{Ngu23}.  Recall that in the context of the Bratteli tree, tensoring with a fundamental representation corresponds to adding a box to the original YD, ensuring that the result remains a valid YD. Similarly, for anti-fundamental representations, 
this process corresponds to removing a box. The reason is as follows: The anti-fundamental representation $\overline{\Box}$ differs from the totally antisymmetric representation with $d - 1$ boxes in a single column by a determinant factor $\det \rho$. This one-dimensional determinant representation, in turn, corresponds to the totally antisymmetric representation with $d$ boxes in a column. Removing this factor corresponds to deleting a single column of length $d$ from the leftmost side of the YD, resulting in the desired irrep. Together, these steps amount to removing a single box from a specific row while preserving the validity of the YD, as illustrated in \cref{fig:removal}.

\begin{figure}[H]
\begin{center}
\vspace{0.5cm}
\includegraphics[scale=0.23]{"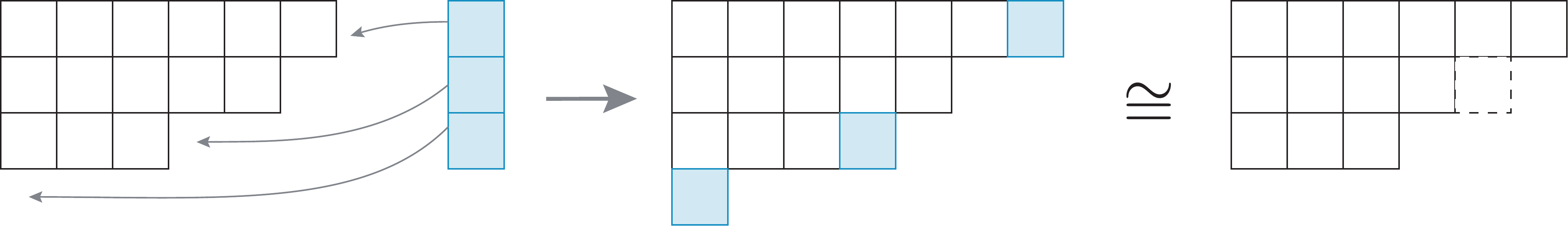"}
\vspace{0.5cm}
\caption{Tensoring with the anti-fundamental representation is equivalent, up to a one-dimensional determinant representation, to first tensoring with the totally antisymmetric representation of \( d-1 \) boxes (blue). Among the resulting diagrams, we select the one where new boxes are added to the first, third, and fourth rows. This diagram is equivalent to the irrep obtained by removing a single box (dashed line) from the second row of the original YD.}
\label{fig:removal}
\end{center}
\end{figure}

We will illustrate this process through a concrete example. Suppose $(s)\Vdash (2,1)$. Consider the qudit representation with $d \geq 3$, i.e., the YT has a minimum of three rows. Suppose we start with the representation $\Yvcentermath1\scriptsize\young(12)$. When we tensor an anti-fundamental representation, to ensure the outcome is a valid YT, one possibility is to remove the box $\Yvcentermath1\scriptsize\young(2)$. In this case, we denote the step by crossing off the box and marking the index, $3$, of the register, as shown in the first pink box from the top in~\cref{fig:mixed_btree}. Alternatively, we can remove one box from the last row. The $d-2$ rows that we have skipped over are represented as a zig-zag line, as shown in the second pink box from the top in~\cref{fig:mixed_btree}.

There is also a generalization of Schur-Weyl duality for systems with mixed symmetry.  In this case, the walled-Brauer algebra plays a role analogous to that of the symmetric group algebra $A_{S_n}$ in the standard Schur-Weyl duality, exhibiting a similar duality with the representation of $GL(d)$.

\begin{figure}[h]
\begin{center}
\includegraphics[scale=0.17]{"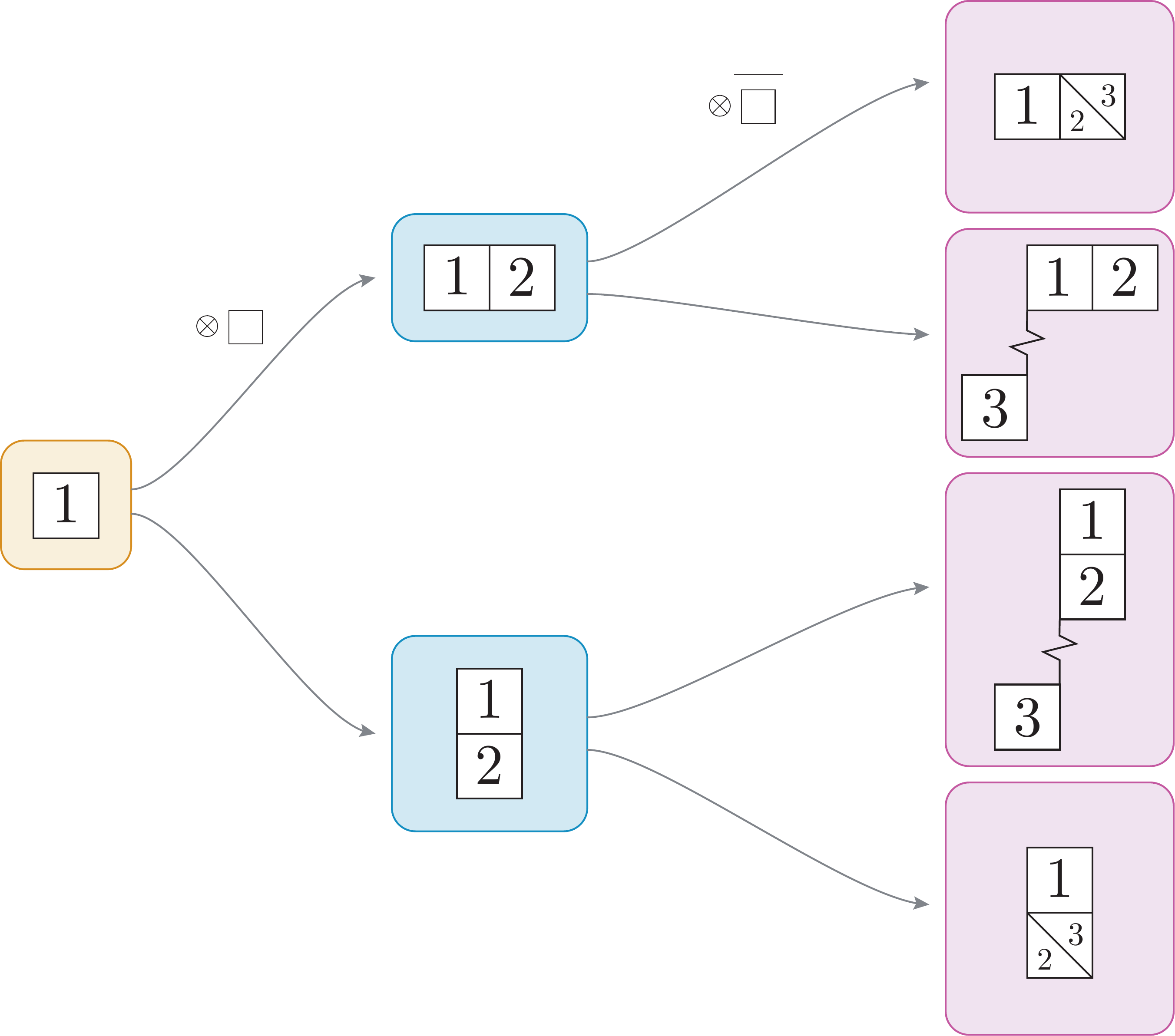"}
\end{center}

\caption{\label{fig:mixed_btree} Bratteli tree for mixed tensor product spaces in the $(2,1)$ configuration. Starting from the left, we first take the tensor product of two copies of the fundamental representation and decompose it into irrep subspaces. For each subspace, we then tensor in another anti-fundamental representation and further decompose into irreps labeled by YTs.}
\end{figure}

\section{Construction of the optimal QPA protocol}
\label{sec:construction}

This section is in correspondence with~\cref{subsec:symmetry,subsec:optimality,subsec:operational} of the main text. We first formalize the symmetry argument and reduction to the LP problem in~\cref{subsec:symmetry_supp} followed by the completion of the optimality proof, which characterizes the optimal QPA protocol~\footnote{We adopt this terminology to distinguish between similar concepts, such as purifying density states to a larger Hilbert space, virtual purification~\cite{HMOL20}, preparing a known target state~\cite{NUY04}, or performing state projection under continuous measurements~\cite{RCMW12}.Moreover, QPA is analogous to amplitude amplification protocols, as both require multiple iterations of a black-box procedure to enhance a specific performance metric—fidelity in the case of QPA and success probability in amplitude amplification} through its Choi representations in~\cref{subsec:optimality_supp}, leading to the main result,~\cref{thm:choi}. We will then show that this definition is equivalent to the three-stage operational protocol in~\cref{subsec:operational_supp}. Finally, we analyze the fidelity in~\cref{subsec:fidelity_supp}. 

\subsection{Symmetry Analysis and Reduction to LP}
\label{subsec:symmetry_supp}
As introduced in~\cref{subsec:symmetry}, we need to find $T^{\yd{\varsigma}}$ by solving the following SDP for each sampling outcome $\yd{\varsigma}$:
\begin{equation}
  \label{equ:subSDP}
    \begin{aligned}
\text{find a matrix} \quad & T^{\yd{\varsigma}}, \\ 
\text{maximize} \quad & \mathrm{tr}((C^{\yd{\varsigma}})^\top T^{\yd{\varsigma}}),\\
\text{subject to} \quad & \mathrm{tr_{out}}T^{\yd{\varsigma}}=\mathbb{I}^{\yd{\varsigma}},\\
& T^{\yd{\varsigma}} \succeq 0.
\end{aligned}
\end{equation}
Here, the cost matrix is $C^{\yd{\varsigma}}=\int_{\mu_\mathrm{Haar}}\ketbra{\psi}{\psi}^\top\otimes\rho^{\yd{\varsigma}}$, and the constraints ensures that $T^{\yd{\varsigma}}$ is the Choi state of a valid quantum channel.
By incorporating covariance symmetry, we can further reduce this SDP into an LP by using the following properties of invariant operators.
\begin{proposition}[Invariant Operator]
    An invariant operator $M^{\yd{\varsigma}}$, i.e. an operator that satisfies the condition $\left[U^\ast_{\mathrm{out}}\otimes U^{\yd{\varsigma}}_{\mathrm{in}},M^{\yd{\varsigma}} \right]=0$ for all $U$ of $SU(d)$, can be decomposed as
    \begin{align}
       U_{\mathrm{CG}}^{\yd{\varsigma}} M^{\yd{\varsigma}} U_{\mathrm{CG}}^{\yd{\varsigma}\dagger} = \bigoplus_{\text{feasible}\,i} m_i \id^{\yd{\mu_i}}.
    \end{align}
    \label{prop:cov-choi}
\end{proposition}
\begin{proof}
In the Schur basis, we divide the matrix of $M^{\yd{\varsigma}}$ into blocks corresponding to intertwining maps between irreps:
\begin{equation}
\begin{bmatrix}
M^{\yd{\mu_{i_1}}} & M^{\yd{\mu_{i_1}}, \yd{\mu_{i_2}}} & \dots & M^{\yd{\mu_{i_1}}, \yd{\mu_{i_l}}} \\
M^{\yd{\mu_{i_2}}, \yd{\mu_{i_1}}} & M^{\yd{\mu_{i_2}}} & \dots & M^{\yd{\mu_{i_2}}, \yd{\mu_{i_l}}} \\
\vdots & \vdots & \ddots & \vdots \\
M^{\yd{\mu_{i_l}}, \yd{\mu_{i_1}}} & M^{\yd{\mu_{i_l}}, \yd{\mu_{i_2}}} & \dots & M^{\yd{\mu_{i_l}}}
\end{bmatrix},
\end{equation}
where the index runs over all feasible $i$'s, i.e. $i_1,\cdots,i_l$, for some $l\leq d$~\footnote{Mathematically, a row $i$ is considered feasible if either: $i = d$ and $\varsigma_d > 0$, or $1 \leq i \leq d-1$ and $\varsigma_i > \varsigma_{i+1}$. In such cases, we define $\yd{\mu_i}=\yd{\varsigma}-\yd{\delta_i}$, with a $1$ in the $i$-th position and $0$ elsewhere, i.e., $\yd{\delta_i}=\yd{0,\cdots,1 \tiny{\text{(at position i)}},\cdots,0}$}. By Schur's Lemma, non-trivial intertwining maps between two irreps exist only when the irreps are equivalent, and in such cases, the map is proportional to the identity. Since the irreps labeled by $\yd{\mu_i}$ are all inequivalent, we conclude that $M^{\yd{\mu_i}} = m_i \id^{\yd{\mu_i}}$ for each feasible $i$, while the off-diagonal blocks are $\mathbf{0}$. Therefore, $M^{\yd{\varsigma}}$ has a block diagonal form.

\end{proof}
According to this proposition, $ C^{\yd{\varsigma}}$ decomposes into blocks. Because the Choi matrix $T^{\yd{\varsigma}}$ has the same symmetry as $C^{\yd{\varsigma}}$, due to the average argument in~\cite{BLM+22}, this decomposition holds as well:
\begin{equation}
\begin{aligned}
    U^{\yd{\varsigma}}_{\text{CG}} C^{\yd{\varsigma}} U^{\yd{\varsigma}\dagger}_{\text{CG}} = \bigoplus_{\text{feasible } i} c_i \id^{\yd{\mu_i}},\\
     U^{\yd{\varsigma}}_{\text{CG}} T^{\yd{\varsigma}} U^{\yd{\varsigma}\dagger}_{\text{CG}} = \bigoplus_{\text{feasible } i} t_i \id^{\yd{\mu_i}}.
\end{aligned}
\end{equation}
With $T^{\yd{\varsigma}}$ now parameterized by the coefficients $t_i$, our goal is to translate the constraints on $T^{\yd{\varsigma}}$ into conditions on these coefficients. The constraints from~\cref{equ:subSDP}, which ensure that $T^{\yd{\varsigma}}$ is a valid Choi state, imply that the coefficients must satisfy $t_i \geq 0$ for all feasible $i$. Additionally, the normalization condition $\sum_{\text{feasible}\,i} t_i d^{\yd{\mu_i}} = d^{\yd{\varsigma}}$ must hold. Finally, we will derive the new objective functions, which are given by:
\begin{equation}
    \Tr(C^{\yd{\varsigma}}T^{\yd{\varsigma}}) 
    = \sum_{\text{feasible } i} c_it_i \Tr(\id^{\yd{\mu_i}})=\sum_{\text{feasible } i}d^{\yd{\mu_i}} c_it_i. 
\end{equation}
For clarity, we use the coefficients $f^{\yd{\mu_i}}=d^{\yd{\varsigma}}c_i$ and the normalized variables $\overline{t}_i = t_i d^{\yd{\mu_i}}/d^{\yd{\varsigma}}$, the LP formulation is given by:
\begin{equation}
\label{equ:lp}
\begin{aligned}
    \text{find a vector}\quad&\overline{t}_i \quad( i\, \text{ranges through all feasible indices}),\\
    \text{that maximizes} \quad &  \sum_{\text{ feasible }i}f^{\yd{\mu_i}}\overline{t}_i,\\
    \text{subjected to} \quad & \sum_{\text{ feasible }i}\overline{t}_i=1,\quad\overline{t}_i\geq 0.
\end{aligned} 
\end{equation}

\subsection{Proof of the Optimal Choi Matrix}
\label{subsec:optimality_supp}
The main result of the optimality proof,~\cref{thm:choi}, is restated below:
\begin{theorem}
    The Choi matrix of the optimal QPA protocol is given by $T=U_{\mathrm{mSch}}^\dag\bigoplus_{\yd{\varsigma}}\frac{d^{\yd{\varsigma}}}{d^{\yd{\mu_{i^\ast}}}}\Pi^{\yd{\mu_{i^\ast}}}U_{\mathrm{mSch}}$, where $i^\ast$ is the smallest feasible row index, and $
\Pi^{\yd{\mu_{i}}} = \bigoplus_{j < i} \mathbf{0}^{\yd{\mu_j}} \oplus \id^{\yd{\mu_{i}}} \oplus \bigoplus_{j > i} \mathbf{0}^{\yd{\mu_j}}
$ is the projector onto the block corresponding to the irrep $\yd{\mu_{i}}$.
\end{theorem}
Here, the unitary $U_{\mathrm{mSch}}$ performs the mixed Schur transform as defined in~\cref{equ:Msch}~\cite{AW18,SHM13,FTH23,GBO23}. In other words, after getting $\rho^{\yd{\varsigma}}$ with $\yd{\varsigma}\vdash n$ from Schur sampling, the optimal channel acting on it has a Choi represenation
$T^{\yd{\varsigma}}=U_{\text{CG}}^\dag \frac{d^{\yd{\varsigma}}}{d^{\yd{\mu_{i^\ast}}}}\Pi^{\yd{\mu_{i^\ast}}}U_{\mathrm{CG}}$.
The rest of the section is devoted to proving this theorem.

To satisfy the constraints and maximize the objective function, the solution to the LP is achieved by setting $\overline{t}_{i^\ast} = 1$ for the largest coefficient $c_{i^\ast}$, with all other entries of $\overline{t}_i$ set to zero. However, directly evaluating $c_i$ is challenging. Instead, we reformulate the problem using a different cost matrix 
\begin{align}
    C^{\prime}=&-\int_{\mu_\mathrm{Haar}}\left({\rho^\top}^{-1}\otimes\rho^{\otimes n}\right),
\end{align} which exhibits Walled-Brauer symmetry and is easier to evaluate.

The reason we can make such replacement lies in the relationship $\ketbra{\psi}{\psi}=A(-{\rho}^{-1})+B\frac{\mathbb{I}_d}{d}$, where $A=\frac{\lambda(d(1-\lambda)+\lambda)}{d^2(1-\lambda)}$ and $B=d+\frac{\lambda}{1-\lambda}$. This shows that $C^{\prime\yd{\varsigma}}$ is related to $C^{\yd{\varsigma}}$ by an affine transformation:
\begin{equation}
    \label{equ:C_and_Cprime}
    \begin{aligned}
     C^{\yd{\varsigma}} &= \int_{\mu_{\mathrm{Haar}}} \left(A(-{\rho^\top}^{-1})+B\frac{\mathbb{I}_d}{d}\right) \otimes \rho^{\yd{\varsigma}} \\
      &= A C^{\prime\yd{\varsigma}} + B \frac{\mathbb{I}_d}{d} \otimes \int_{\mu_{\mathrm{Haar}}} \rho^{\yd{\varsigma}} \\
      &=  A C^{\prime\yd{\varsigma}} + B \Tr(\rho^{\yd{\varsigma}}) \frac{\mathbb{I}_d}{d} \otimes \frac{\mathbb{I}^{\yd{\varsigma}}}{d^{\yd{\varsigma}}}.
\end{aligned}
\end{equation}

It turns out this affine transformation does not affect the optimization, as shown by a projection argument. For an SDP, the linear constraints defines an affine subspace $\mathbb{U}=\{x|Lx=b\}\subset\mathbb{V}$ for some surjective linear operator $L:\mathbb{V}\to\mathbb{W}$ and $b\in\mathbb{W}$, illustrated as the upper blue plane in~\cref{fig:projection}. This hyperplane intersects with the positive semi-definite (PSD) cone, shown as the pink cone, defining the set of feasible solutions. If the cost matrix $C^{\yd{\varsigma}}$ (dark green arrow) has a component perpendicular to the hyperplane (dashed lines), it is not going to affect the optimization. Motivated by this, we project onto the parallel, homogeneous hyperplane $\mathbb{U}^\prime=\{x|Lx=0\}$ using the projector 
$P=1-L^+L$, where $L^+=L^\star(LL^\star)^{-1}$ is the pseudoinverse of $L$, and $L^\star$ denotes the adjoint map $L^\star:\mathbb{W}\to\mathbb{V}$ under the usual inner product $\Tr(\cdot^\dag\cdot)$.

In our case, $\mathbb{V}$ represents the space of linear operators acting on the out and in registers, and $\mathbb{U}$ is the hyperplane of operators that satisfy (up to a constant) the Choi matrix trace-preserving condition: $\mathrm{tr}_\mathrm{out}(\cdot)=\frac{\mathbb{I}^{\yd{\varsigma}}}{d^{\yd{\varsigma}}}$. Therefore, the feasible set defines all valid Choi states. Since the trace map on the out register satisfies $\Tr_\mathrm{out}^+(\cdot)=\Tr_\mathrm{out}^\star(\cdot)=\frac{\mathbb{I}}{d}\otimes\cdot$, the projector takes the following form:
\begin{lemma}
    $P=\mathrm{Id}-\frac{\mathbb{I}}{d}\otimes\mathrm{tr_{out}}$ is a projector acting on $\mathbb{V}_{\mathrm{out}}\otimes\mathbb{V}^{\yd{\varsigma}}_{\mathrm{in}}$. 
\end{lemma}
\begin{figure}[H]
\begin{center}
\includegraphics[scale=0.5]{"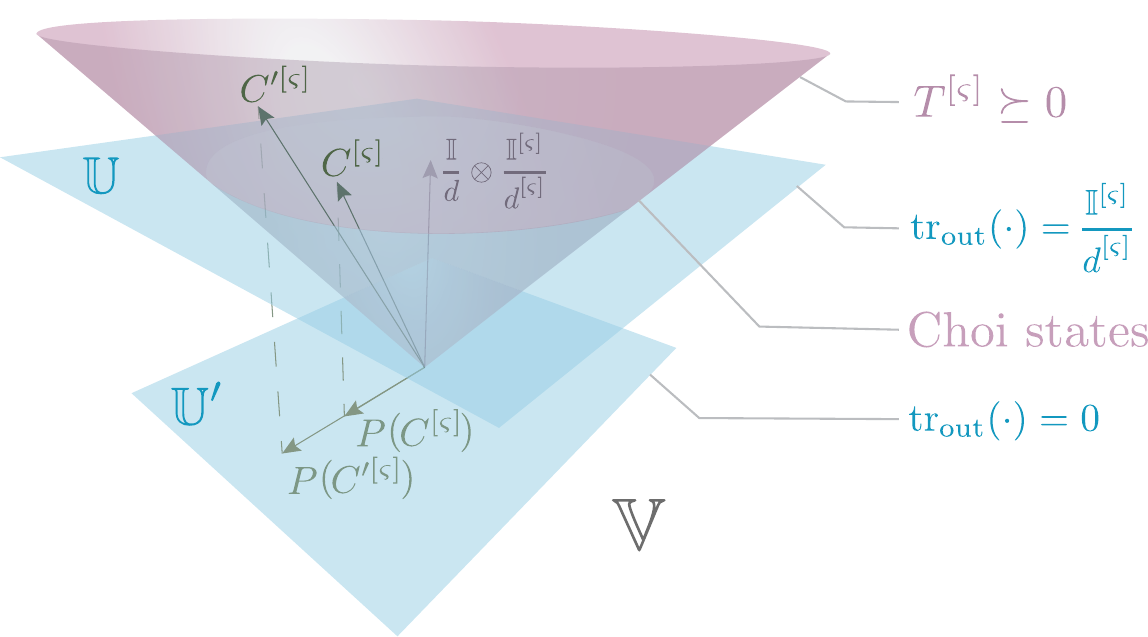"}
\caption{The ambient space $\mathbb{V}$ is the linear subspace of all operators acting on the out and in registers. The pink cone represents the PSD cone, which intersects with the hyperplane $\mathbb{U}$ corresponding to the affine subspace where the partial trace is proportional to the identity. This intersection defines a convex set containing all valid Choi states. While $C^{\yd{\varsigma}}$ is a valid Choi state, $C^{\prime\yd{\varsigma}}$ may not be, as it does not necessarily satisfy the trace condition. However, their projections are collinear and define the same optimization problem.}
\label{fig:projection}
\end{center}
\end{figure}

With the definition of the projection in place, we now formalize the equivalence of cost matrices in~\cref{lm:eq_C}:
\begin{lemma}[Equivalence of cost matrixs]
\label{lm:eq_C}
When solving for the optimal Choi matrices with SDP, if two cost matrices $C_1$ and $C_2$ satisfy $P(C_1) = A P(C_2)$, where $A$ is a positive scalar, then the SDP problems with either cost matrices $C_1$ or $C_2$ yield the same optimal solution $T^\ast$.
\end{lemma}
\begin{proof}
    For any valid Choi state $T$,
    \begin{align}
        \nonumber\Tr( C_1 T) &= \Tr[ P(C_1) T] + \Tr\left[ \frac{\id_{\mathrm{out}}}{d_\mathrm{out}} \otimes \Tr_{\mathrm{out}}(C_1)T\right] \\\nonumber
        &=\Tr[ P(C_1) T] + \frac{\Tr\left[\Tr_{\mathrm{out}}(C_1) \Tr_{\mathrm{out}}(T) \right]}{d_\mathrm{out}} \\
        &=\Tr[ P(C_1) T] + \frac{\Tr(C_1)}{d_\mathrm{out}}.
    \end{align} 
    Similarly, $\Tr( C_2 T) = \Tr[ P(C_2) T] + \Tr(C_2)/d_{\mathrm{out}}$.
    The lemma follows from the condition that for all $T$, $\Tr[ P(C_1) T] = A\Tr[ P(C_2) T]$, and that $\Tr(C_1)$ and $\Tr(C_2)$ are constants independent of $T$.
\end{proof}
Returning to the case of the cost matrices $C^{\yd{\varsigma}}$ and $C^{\prime\yd{\varsigma}}$, by~\cref{equ:C_and_Cprime}, we have:
\begin{equation}
  P(C^{\yd{\varsigma}}) = A P(C^{\prime\yd{\varsigma}}) + B \Tr(\rho^{\yd{\varsigma}}) P\left(\frac{\mathbb{I}_d}{d} \otimes \frac{\mathbb{I}^{\yd{\varsigma}}}{d^{\yd{\varsigma}}}\right) = A P(C^{\prime\yd{\varsigma}}).
\end{equation}
By~\cref{lm:eq_C}, we can replace $C^{\yd{\varsigma}}$ by $C^{\prime\yd{\varsigma}}$ and reformulate our LP in~\cref{equ:lp} with $c^\prime_i$ instead of $c_i$. 

This time, $c^\prime_i$ can be easily evaluated by recognizing that ${\rho^\top}^{-1}\otimes\rho^{\yd{\varsigma}}$ is a rational representation of $\rho$. As a result, we have the decomposition ${\rho^\top}^{-1}\otimes\rho^{\yd{\varsigma}}=\bigoplus_{\text{feasible  } i} \rho^{\yd{\mu_i}}$. 
Using the Weyl character formula to express the trace of $\rho^{\yd{\mu_i}}$ as a Schur polynomial $s^{\yd{\mu_i}}$, and defining the normalized Schur polynomial  $S^{\yd{\mu_i}}$ by dividing it by the dimension of the irrep, we obtain:
\begin{equation}
c^\prime_i=-\frac{\Tr(\rho^{\yd{\mu_i}})}{d^{\yd{\mu_i}}}=-\frac{s^{\yd{\mu_i}}(a,\cdots,a,b)}{d^{\yd{\mu_i}}}=-S^{\yd{\mu_i}}(a,\cdots,a,b).
\end{equation}
Here, the arguments $a$ and $b$ in the Schur polynomials represent the two eigenvalues of the depolarized state, given by $\frac{\lambda}{d}$ and $1 - \frac{d-1}{d} \lambda$, respectively.

Now we are ready to provide an explicit solution to this LP problem by considering the ordering of normalized Schur polynomials. Based on Ref.~\cite{S16}, we know that for the following partial order defined on YDs:
\begin{equation}
\yd{\mu}\preceq\yd{\lambda}\Leftrightarrow \forall k\in\{1,\cdots,n\}:\sum_{a=1}^k\mu_a\leq\sum_{a=1}^k\lambda_a,
\end{equation}
the size of the normalized Schur polynomials for positive arguments satisfies 
\begin{equation}
S^{\yd{\mu}}\leq S^{\yd{\lambda}}\Leftrightarrow \yd{\mu}\preceq\yd{\lambda}.
\end{equation}
Here, the Schur polynomials $s^{\yd{\mu}}$, which correspond to the traces of density matrices, are always positive, and the same holds for $S^{\yd{\mu}}$.
Moreover, observe that $\yd{\mu_i}\preceq\yd{\mu_j}$ whenever $i\leq j$.
For instance, consider the sampling outcome $\yd{\varsigma}=\Yvcentermath1\scriptsize\yng(4,3,1)$, the possible $\yd{\mu_i}$ are:
$\yd{\mu_1}=\Yvcentermath1\scriptsize\yng(3,3,1)$,
$\yd{\mu_2}=\Yvcentermath1\scriptsize\yng(4,2,1)$,
$\yd{\mu_3}=\Yvcentermath1\scriptsize\yng(4,3)$. Clearly, $\yd{\mu_1}\preceq\yd{\mu_2}\preceq\yd{\mu_3}$.
Therefore, to maximize the objective function, we select the smallest feasible index \(i^\ast\) since \(\yd{\mu_i} \preceq \yd{\mu_j}\) whenever \(i < j\).
Hence, for each $\yd{\varsigma}$, the Choi representation of the optimal channel $T^{\yd{\varsigma}}$ is $U_{\text{CG}}^\dag\frac{d^{\yd{\varsigma}}}{d^{\yd{\mu_{i^\ast}}}}\Pi^{\yd{\mu_{i^\ast}}}U_{\mathrm{CG}}$.

Finally, by combining the branches $T^{\yd{\varsigma}}$ for each sector $\yd{\varsigma}$ and transforming back to the computational basis using $U_{\text{Sch}}$, we obtain the explicit construction of the optimal QPA Choi matrix $T$.

\subsection{Operational Definition of optimal QPA protocol}
\label{subsec:operational_supp} 

For this proof, we work with Young tableaux, as they facilitate the study of physical operations by embedding abstract operations into transformations within the physical space. As explained in the main text, each outcome $\yt{s}$ of the Schur sampling is processed by applying a correction followed by a partial trace over the additional registers. We now provide the proof based on the Choi matrix $T^{\yd{\varsigma}}$ for each outcome. Recall that for each YT $\yt{s}$, we can apply a correction to map it into the column-ordered YT $\yt{\varsigma^\diamond}$, where the box labeled $n$ is the one we are going to remove. For instance, both $\Yvcentermath1\scriptsize\young(1357,246)$ for $n=7$ and $\Yvcentermath1\scriptsize\young(14,25,3)$ for $n=5$ are column-ordered.

 We assume that \( \yt{s} \vdash \yd{\varsigma} \) and \( \yt{m} \vdash \yd{\mu_{i^\ast}} \). Moreover, after applying the correction, we can assume that \( \yt{s} \) is column-ordered. From this point, we need to explicitly construct the projector $\Pi_{\yt{m}}$. It turns out that in tensor network notation, the projector is obtained by bending the last leg of the projector $\Pi_{\yt{s}}$. Here, in the equations, the registers are arranged from top to bottom as $\mathrm{in}_1, \cdots, \mathrm{in}_n, \mathrm{out}$, and the thick gray lines represent multiple registers. Transcriptions of the tensor network diagrams are provided beneath each one for those who prefer non-graphic notation.

\begin{center}
\hfill
\begin{minipage}[c]{0.1\textwidth}
\begin{figure}[H]
\begin{center}
\includegraphics[scale=0.3]{"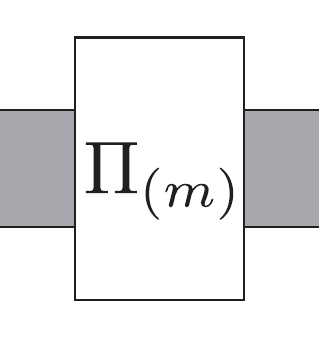"}
\end{center}
\end{figure}
\end{minipage}%
\begin{minipage}[c]{0.1\textwidth}
\begin{center}
\raisebox{0.1cm}{
\(\xlongequal{\quad}\quad \dfrac{d^{\yt{m}}}{d^{\yt{s}}}\)}
\end{center}
\end{minipage}%
\begin{minipage}[c]{0.3\textwidth}
\begin{figure}[H]
\begin{center}
\raisebox{0.15cm}{\includegraphics[scale=0.3]{"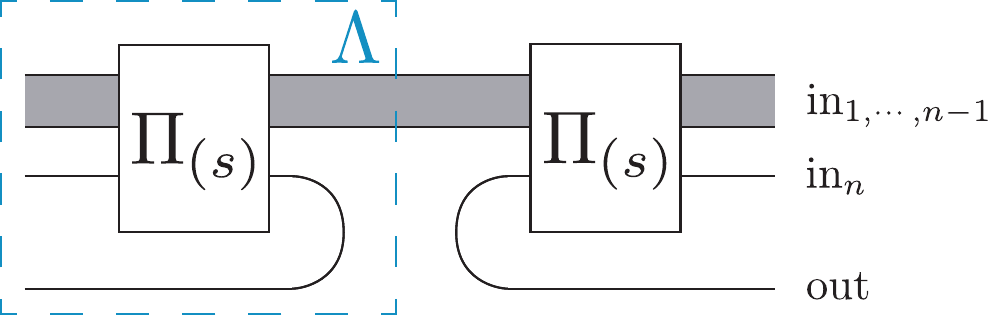"}}
\end{center}
\end{figure}
\end{minipage}%
\hfill
\raisebox{0.38cm}{
\begin{minipage}[c]{0.05\textwidth}
\begin{equation}
    \quad
    \label{equ:proj_def}
\end{equation}
\end{minipage}%
}
\end{center}
Transcription:
$$\Pi_{\yt{m}}=\frac{d^{\yt{m}}}{d^{\yt{s}}}\Pi_{\yt{s}\,\mathrm{in}}d\ketbra{\mathrm{EPR}}{\mathrm{EPR}}_{\mathrm{in}_n,\mathrm{out}}\Pi_{\yt{s}\,\mathrm{in}}.$$ 

To verify that the operator on the RHS is indeed proportional to the projector on the irrep subspace of $\yt{m}$,
note that the operator can be split as $\frac{d^{\yt{m}}}{d^{\yt{s}}}\Lambda^\dag\Lambda$, where the isometry $\Lambda$ is given by $\Lambda=\sqrt{d}\bra{\mathrm{EPR}}_{\mathrm{in}_n,\mathrm{out}}\Pi_{\yt{s}\,\mathrm{in}}$ as shown in the dashed blue box in~\cref{equ:proj_def}. If we define the parent YT of $\yt{s}$ as $\yt{s^\prime}$, i.e., $\yt{s^\prime} \to \yt{s}$,
this establishes the following relationship within the Bratteli tree:
\begin{equation}
    \yt{s^\prime} \to \yt{s} \to \yt{m}.
\end{equation}
Since $\yt{s}$ is column-ordered, the $n$-th box is always located in the lower right corner of $\yt{s}$ and will be removed when deriving $\yt{s^\prime}$; therefore, the irrep $\yt{m}$ is equivalent to $\yt{s^\prime}$. 
Upon closer inspection of $\Lambda$, we see that it exhibits the following covariance:
\begin{equation}
\Lambda \left( U_{\yt{s^\prime}\,\mathrm{in}_{1,\cdots,n-1}} \otimes U_{\mathrm{in}_n} \otimes U^\ast_{\mathrm{out}} \right) = U_{\yt{s^\prime}\,\mathrm{in}_{1,\cdots,n-1}}\Lambda.
\end{equation}
Since $\yt{m}$ is the only irrep within $\yt{s}\otimes\overline{\Box}$ that is equivalent to $\yt{s^\prime}$, by Schur's lemma, $\Lambda$ is the unique (up to scalar multiplication) intertwining map from $V_{\yt{m}}$ to $V_{\yt{s^\prime}}$. Moreover, \(\Lambda\) maps all other non-equivalent irreps to 0. After the application of \( \Lambda \), \( \Lambda^\dag \) acts as an intertwining map that isometrically lifts it back to the original space. Therefore, $\Lambda^\dag \Lambda$ is proportional to the projector onto $V_{\yt{m}}$.

To fix the normalization, we check its idempotency: 
\begin{center}
\begin{minipage}[c]{0.175\textwidth}
\begin{center}
\raisebox{0.1cm}{
$\Pi_{\yt{m}}^2= \left(\dfrac{d^{\yt{m}}}{d^{\yt{s}}}\right)^2$}
\end{center}
\end{minipage}%
\begin{minipage}[c]{0.44\textwidth}
\begin{figure}[H]
\begin{center}
\raisebox{0.15cm}{\includegraphics[scale=0.3]{"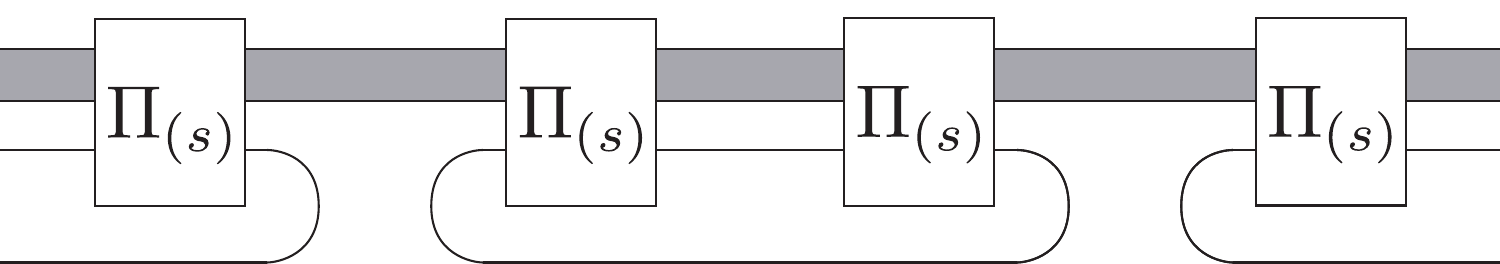"}}
\end{center}
\end{figure}
\end{minipage}%
\begin{minipage}[c]{0.05\textwidth}
\begin{center}
\raisebox{-0.1cm}{
\(\xlongequal{\quad}\)}
\end{center}
\end{minipage}%
\begin{minipage}[c]{0.34\textwidth}
\begin{figure}[H]
\begin{center}
\raisebox{0.15cm}{\includegraphics[scale=0.3]{"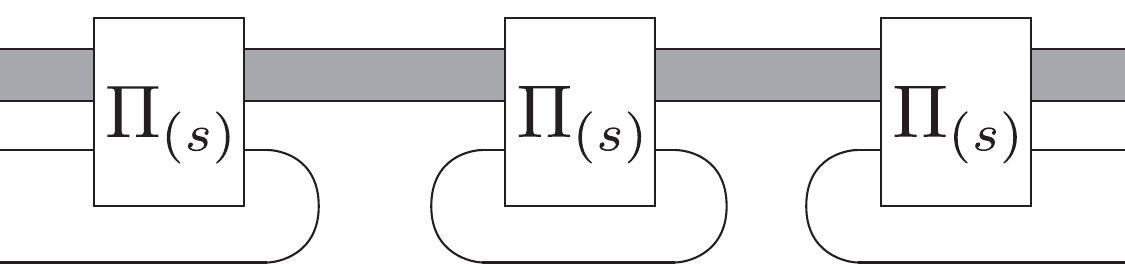"}}
\end{center}
\end{figure}
\end{minipage}%
\end{center}
\begin{center}
\hfill
\begin{minipage}[c]{0.1\textwidth}
\begin{center}
\raisebox{0.1cm}{
\(\xlongequal{\quad}\)}
\end{center}
\end{minipage}%
\begin{minipage}[c]{0.1\textwidth}
\begin{center}
\raisebox{0.1cm}{
\(\dfrac{d^{\yt{m}}}{d^{\yt{s}}}\)}
\end{center}
\end{minipage}%
\begin{minipage}[c]{0.22\textwidth}
\begin{figure}[H]
\begin{center}
\raisebox{0.15cm}{\includegraphics[scale=0.3]{"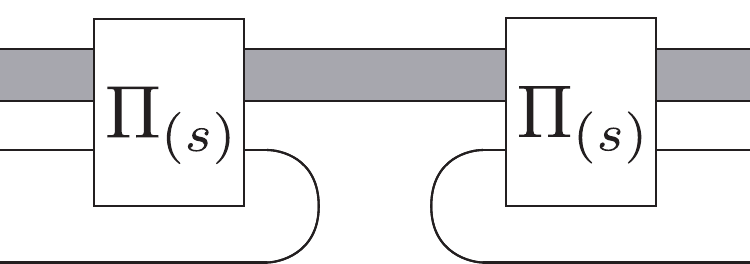"}}
\end{center}
\end{figure}
\end{minipage}%
\begin{minipage}[c]{0.2\textwidth}
\begin{center}
\raisebox{0.1cm}{
\(\xlongequal{\quad}\qquad\Pi_{\yt{m}}\)}
\end{center}
\end{minipage}%
\hfill
\raisebox{0.38cm}{
\begin{minipage}[c]{0.05\textwidth}
\begin{equation}
    \quad
\end{equation}
\end{minipage}%
}
\end{center}
Transcription:
\begin{align*}
\Pi_{\yt{m}}^2=&\left(\frac{d^{\yt{m}}}{d^{\yt{s}}}\Pi_{\yt{s}\,\mathrm{in}}d\ketbra{\mathrm{EPR}}{\mathrm{EPR}}_{\mathrm{in}_n,\mathrm{out}}\Pi_{\yt{s}\,\mathrm{in}}\right)^2\\ &=\left(\frac{d^{\yt{m}}}{d^{\yt{s}}}\right)^2\Pi_{\yt{s}\,\mathrm{in}}d\ketbra{\mathrm{EPR}}{\mathrm{EPR}}_{\mathrm{in}_n,\mathrm{out}}\Pi_{\yt{s}\,\mathrm{in}}d\ketbra{\mathrm{EPR}}{\mathrm{EPR}}_{\mathrm{in}_n,\mathrm{out}}\Pi_{\yt{s}\,\mathrm{in}}\\
    &=\left(\frac{d^{\yt{m}}}{d^{\yt{s}}}\right)^2\Pi_{\yt{s}\,\mathrm{in}}\sqrt{d}\ket{\mathrm{EPR}}_{\mathrm{in}_n,\mathrm{out}}\Tr_{\mathrm{in}_n}\left(\Pi_{\yt{s}\,\mathrm{in}}\right)\sqrt{d}\bra{\mathrm{EPR}}_{\mathrm{in}_n,\mathrm{out}}\Pi_{\yt{s}\,\mathrm{in}}\\
    &=\frac{d^{\yt{m}}}{d^{\yt{s}}}\Pi_{\yt{s}\,\mathrm{in}}d\ketbra{\mathrm{EPR}}{\mathrm{EPR}}_{\mathrm{in}_n,\mathrm{out}}\Pi_{\yt{s}\,\mathrm{in}}\\
    &=\Pi_{\yt{m}}.
\end{align*}
In the third equality (fifth in the transcription), we applied the compatibility of Hermitian projectors and note that $d^{\yt{s^\prime}}=d^{\yt{\mu^\ast}}$:\\
\begin{center}
\hfill\begin{minipage}[c]{0.1\textwidth}
\begin{figure}[H]
\begin{center}
\raisebox{0.15cm}{\includegraphics[scale=0.3]{"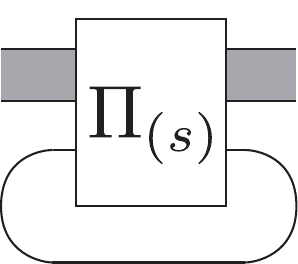"}}
\end{center}
\end{figure}
\end{minipage}%
\begin{minipage}[c]{0.1\textwidth}
\begin{center}
\raisebox{0.1cm}{
\(\xlongequal{\quad}\frac{d^{\yt{s}}}{d^{\yt{s^\prime}}}\)}
\end{center}
\end{minipage}%
\begin{minipage}[c]{0.1\textwidth}
\begin{figure}[H]
\begin{center}
\raisebox{0.15cm}{\includegraphics[scale=0.3]{"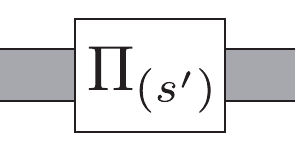"}}
\end{center}
\end{figure}
\end{minipage}%
\hfill
\raisebox{0.38cm}{
\begin{minipage}[c]{0.05\textwidth}
\begin{equation}
    \quad
\end{equation}
\end{minipage}%
}
\end{center}
Transcription:
$$
    \Tr_{\mathrm{in}_n}(\Pi_{\yt{s}\, \mathrm{in}})=\frac{d^{\yt{s}}}{d^{\yt{s^\prime}}}\Pi_{\yt{s^\prime}\, \mathrm{in}_{1,\cdots,n}}
.$$ For more insight into these projectors, see Ref.~\cite{KS18}. Finally, we apply the projector and use the necklace rule to reduce operation into a partial trace, as shown below:
\begin{center}
\hfill
\begin{minipage}[c]{0.15\textwidth}
\begin{figure}[H]
\begin{center}
\includegraphics[scale=0.3]{"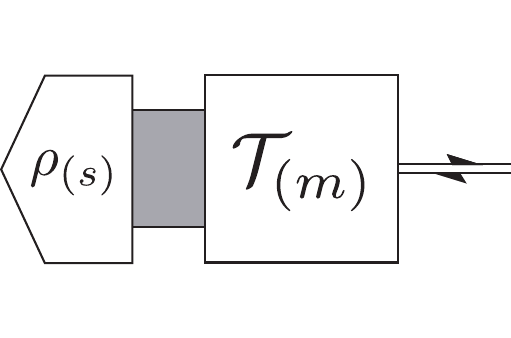"}
\end{center}
\end{figure}
\end{minipage}%
\begin{minipage}[c]{0.05\textwidth}
\begin{center}
\raisebox{0.1cm}{
\(\xlongequal{\quad}\)}
\end{center}
\end{minipage}%
\begin{minipage}[c]{0.15\textwidth}
\begin{figure}[H]
\begin{center}
\includegraphics[scale=0.3]{"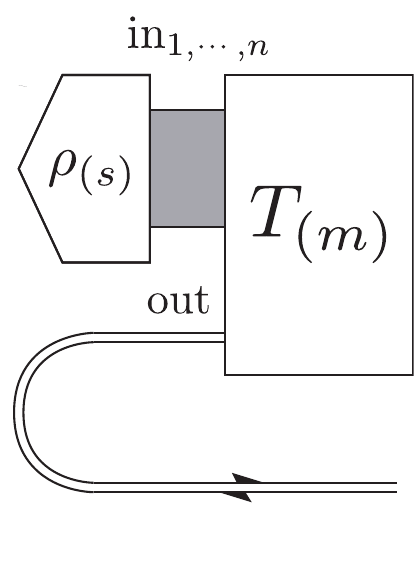"}
\end{center}
\end{figure}
\end{minipage}%
\begin{minipage}[c]{0.05\textwidth}
\begin{center}
\raisebox{0.1cm}{
\(\xlongequal{\quad}\)}
\end{center}
\end{minipage}%
\begin{minipage}[c]{0.2\textwidth}
\begin{figure}[H]
\begin{center}
\includegraphics[scale=0.3]{"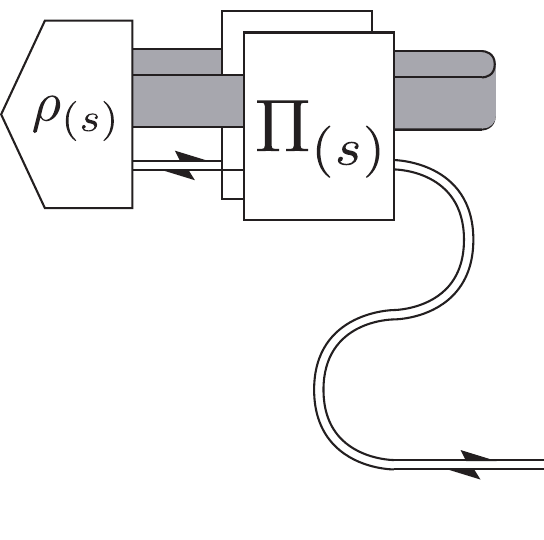"}
\end{center}
\end{figure}
\end{minipage}%
\begin{minipage}[c]{0.05\textwidth}
\begin{center}
\raisebox{0.1cm}{
\(\xlongequal{\quad}\)}
\end{center}
\end{minipage}%
\begin{minipage}[c]{0.18\textwidth}
\begin{figure}[H]
\begin{center}
\includegraphics[scale=0.3]{"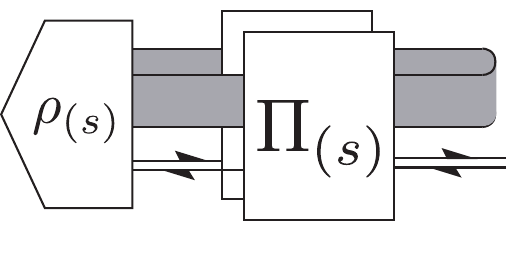"}
\end{center}
\end{figure}
\end{minipage}%
\hfill
\raisebox{0.38cm}{
\begin{minipage}[c]{0.05\textwidth}
\begin{equation}
    \quad
\end{equation}
\end{minipage}%
}
\end{center}
Transcription:
\begin{align*}
\mathcal{T}_{\yt{m}}\left(\rho_{\yt{s}}\right)=&d\Tr_{\mathrm{in, out}_1}\left(\rho_{\yt{s}\,\mathrm in} \ketbra{\mathrm{EPR}}{\mathrm{EPR}}_{\mathrm{out}_{1,2}}T_{\yt{m}\,\mathrm{in},\mathrm{out}}\right)\\
    =&d\Tr_{\mathrm{in, out}_1}\left(\rho_{\yt{s}\,\mathrm in} \ketbra{\mathrm{EPR}}{\mathrm{EPR}}_{\mathrm{out}_{1,2}}\left(\Pi_{\yt{s}\,\mathrm{in}}d\ketbra{\mathrm{EPR}}{\mathrm{EPR}}_{\mathrm{in}_n,\mathrm{out}_1}\Pi_{\yt{s}\,\mathrm{in}} \right)\right)\\
    =&d\Tr_{\mathrm{in}_{1,\cdots,n-1}}\left(\bra{\mathrm{EPR}}_{\mathrm{in}_n,\mathrm{out}_1} \left(\Pi_{\yt{s}}\rho_{\yt{s}\,\mathrm{in}}\Pi_{\yt{s}} d\ketbra{\mathrm{EPR}}{\mathrm{EPR}}_{\mathrm{out}_{1,2}} \right)\ket{\mathrm{EPR}}_{\mathrm{in}_n,\mathrm{out}_1}\right)\\
    =&\Tr_{\mathrm{in}_{1,\cdots,n-1}}{\rho_{\yt{s}\,\mathrm in}}.
\end{align*}

\subsection{Optimal Sample Complexity}
\label{subsec:fidelity_supp}
After constructing the optimal QPA protocol, we now turn our attention to quantifying the sample complexity it can achieve. To achieve this, we first obtain the overall fidelity by summing all contributions from \( \yd{\varsigma} \), while accounting for the degeneracy \( g^{\yd{\varsigma}} \):
\begin{equation}
    \mathcal{F} = \sum_{\yd{\varsigma}}g^{\yd{\varsigma}} s^{\yd{\varsigma}} \cdot \frac{f^{\yd{\varsigma}}}{s^{\yd{\varsigma}}}.
 \end{equation}
By the same argument as in Ref.~\cite{OW17}, we rescale functions of the YDs as functions of their normalized form, namely, \(\yd{\overline{\varsigma}} = [\varsigma_1/n, \dots, \varsigma_d/n]\), allowing for asymptotic analysis. For input states \(\rho\) with spectra \(  p_{1} \leq \dots \leq p_{d-1}  < p_d \), as \( n \to \infty \), the weights \( g^{\yd{\overline{\varsigma}}} s^{\yd{\overline{\varsigma}}} \) concentrate around a point measure \( \delta\left(\yd{\overline{\varsigma}} - \yd{\overline{\varsigma^\circ}}\right) \), where \( \yd{\overline{\varsigma^\circ}} = \yd{p_d, \dots, p_1} \), while the tail contributions decay exponentially. Since \( \frac{f^{\yd{\overline{\varsigma}}}}{s^{\yd{\overline{\varsigma}}}} \), the normalized fidelity, admits a continuum asymptotic limit~\cite[Chapter 6.2]{H05}, the overall fidelity \( \mathcal{F} \) is dominated by the typical fidelity \( \frac{f^{\yd{\varsigma^\circ}}}{s^{\yd{\varsigma^\circ}}} \), which corresponds to the configuration \( \yd{\varsigma^\circ}=\yd{np_d, \dots, np_1} \). Note that \( np_1 \) and other arguments in \( \yd{\varsigma^\circ} \) should be understood as their integer floor values, ensuring they remain valid YDs. This rounding does not alter the qualitative behavior of the asymptotics.

For the case with depolarized states, the spectra of the inputs are given by \(a, \dots, a,b\). By~\cref{equ:C_and_Cprime}, we calculate the optimal fidelity for each sector \( \yd{\varsigma} \):
\begin{equation}
\label{equ:fidelity}
\begin{aligned}
    f^{\yd{\mu_i}} =&\Tr(C^{\yd{\varsigma}}T^{\yd{\varsigma}})\\ =& A \Tr(C^{\prime\yd{\varsigma}}T^{\yd{\varsigma}}) + B \Tr(\rho^{\yd{\varsigma}})/d\\
    =&- A S^{\yd{\mu_i}}(a,\cdots,a,b)d^{\yd{\varsigma}}+ \frac{B}{d} s^{\yd{\varsigma}}(a,\ldots,a,b).
\end{aligned}
\end{equation}
The typical fidelity is given by
 \begin{equation}
 \begin{aligned}
    \frac{f^{\yd{\varsigma^\circ}}}{s^{\yd{\varsigma^\circ}}} =1 + \frac{\lambda}{d(1-\lambda)} \left(1-b\frac{ S^{\yd{\mu^\circ_1}} (a,\cdots,a,b)}{ S^{\yd{\varsigma^\circ}} (a,\cdots,a,b)}
  \right).
 \end{aligned}
  \end{equation}
Here, $\yd{\mu^\circ_1}$ is the new YT obtained by removing one box from the first row of $\yd{\varsigma^\circ}$. The object of interest consists a ratio of two normalized Schur polynomials:
\begin{equation}
\begin{aligned}
    \frac{S^{\yd{\mu^\circ_1}} }{ S^{\yd{\varsigma^\circ}} }=\frac{s^{\yd{\mu^\circ_1}} }{s^{\yd{\varsigma^\circ}} }\frac{d^{\yd{\varsigma^\circ}}}{d^{\yd{\mu^\circ_1}}}
    =\frac{s^{\yd{cn-1,0,\cdots,0}} }{s^{\yd{cn,0,\cdots,0}}}\frac{cn+d-1}{cn},
\end{aligned}
\end{equation}
where we have defined $c=b-a=1-\lambda$ for convenience. Moreover, by explicitly studying the recursive construction of Schur polynomials,
\begin{equation}
\begin{aligned}
s^{\yd{cn,0,\cdots,0}}=&\sum_{k=1}^{cn}\binom{cn-k+d-2}{d-2}a^{cn-k}b^k+\binom{cn+d-2}{d-2}a^{cn}\\=& b s^{\yd{cn-1,0,\cdots,0}}+\binom{cn+d-2}{d-2} a^{cn}.
\end{aligned}
\end{equation}
Taking the ratio,
\begin{equation}
\begin{aligned}
\frac{s^{\yd{cn-1,0,\cdots,0}}}{s^{\yd{cn,0,\cdots,0}}}=\frac{1}{b}\left(1-\binom{cn+d-2}{d-2} \frac{a^{cn}}{s^{\yd{cn,0,\cdots,0}}}\right).
\end{aligned}
\end{equation}
Since $s^{\yd{cn,0,\cdots,0}}$ is lower bounded by $b^{cn}$, $\frac{a^{cn}}{s^{\yd{cn,0,\cdots,0}}}\in O\left(\mathrm{exp}(-n)\right)$,
\begin{equation}
\begin{aligned}
\frac{s^{\yd{cn-1,0,\cdots,0}}}{s^{\yd{cn,0,\cdots,0}}}=\frac{1}{b}+O\left(\mathrm{exp}(-n)\right).
\end{aligned}
\end{equation}
Note that here $O\left(\mathrm{exp}(-n)\right)$ is used to denote an arbitrary exponentially or super-exponentially decaying function, formally defined as $\bigcup_{0 < c < 1} O(c^n)$. Plugging into the original expression for overall fidelity~\footnote{There appears to be a discrepancy in the calculations between Eqs. (5.56) and (5.57) in Ref.~\cite{F16}, where the average fidelity is given as \(1 - \frac{\lambda}{3(1-\lambda)^2(N+1)}\). Our calculations suggest that the correct expression should be \(1 - \frac{2\lambda}{3(1-\lambda)^2(N+1)}\).},
\begin{equation}
\begin{aligned}
    \mathcal{F}=&1+\frac{\lambda}{d(1-\lambda)}\left(1-b\frac{s^{\yd{cn-1,0,\cdots,0}}}{s^{\yd{cn,0,\cdots,0}}}\frac{cn+d-1}{cn}\right)\\
    =&1-\frac{\lambda}{(1-\lambda)^2}\frac{d-1}{d}\frac{1}{n}+O\left(\mathrm{exp}(-n)\right).
\end{aligned}
\end{equation}
Finally, setting $\mathcal{F} = 1 - \delta$ yields the optimal sample complexity.

\subsection{Generic Input States}
\label{subsec:gen_inputs}
By studying $\rho^{\yd{\varsigma}}$ directly in the GT basis (see \cref{subsubsec:Schur} for a reminder), we provide an alternative expression of the optimal fidelity, which appeared as Eq.~(4.67) in Ref.~\cite{F16}:  
\begin{equation} 
\begin{aligned}
    f^{\yd{\mu_i}}=&\frac{d^{\yd{\varsigma}}}{d^{\yd{\mu_i}}}\mathrm{tr}\left(U_{\mathrm{CG}}\ketbra{d}{d}\otimes\rho^{\yd{\varsigma}}U^\dagger_{\mathrm{CG}}\Pi^{\yd{\mu_i}}\right)\\
=&\frac{d^{\yd{\varsigma}}}{d^{\yd{\mu_i}}}\sum_{\wt{m}}\dfrac{\prod_{j=1}^{d-1}(m_{j,d-1}-j-\varsigma_i+i)}{\prod_{j=1,\,  j \neq i}^d (\varsigma_j-j-\varsigma_i+i)}p_d^{\#_d}\cdots p_1^{\#_1}.
\end{aligned}
\end{equation}
For the first equality, we choose the basis states \( \ket{i} \) for \( i = 1, \dots, d \) as the eigenbasis of \( \rho \), corresponding to eigenvalues \( p_1, \dots, p_d \), without loss of generality. In this basis, \( \rho^{\yd{\varsigma}} \) is diagonal, with its values determined by the weights \( \#_1, \dots, \#_d \) (i.e., the number of times each number appears in the WT). For the second equality, we have explicitly expressed the CG coefficients in terms of the entries of the GT pattern, $\wt{m_{i,j}}$. 

By the concentration of measure described in the last section, we can focus on the typical configuration $\yd{\varsigma^\circ}= [np_d, \dots, np_1]$. Moreover, we assume the spectra are non-degenerate. Although the fidelity is difficult to evaluate directly, we can analyze its asymptotic behavior order by order. To do this effectively, we first introduce an alternative parametrization of $\wt{m}$ with parameters $\{t_{i,j}\}$ for $1\leq i<j \leq d$ via the recursive definition:
\begin{equation}\label{eq:recursive_definition}
m_{i-1,j-1} - m_{i,j} = t_{j - i + 1,j}.
\end{equation}
To understand the parametrization graphically, we interpret GT pattern as a sequence of nested YDs, each obtained by adding boxes to the previous one. In the GT pattern, the \((d-j+1)\)-th row corresponds to the rows of the YD filled with \( j \), while the \((d-j+2)\)-th row corresponds to the rows of the YD filled with \( j-1 \). We define the \((i-1)\)-th row of the smaller YD to be longer than the \( i \)-th row of the larger YD by exactly \( t_{j-i+1, j} \) boxes. 
\begin{figure}[H]
\begin{center}
\vspace{0.5cm}
\includegraphics[scale=0.25]{"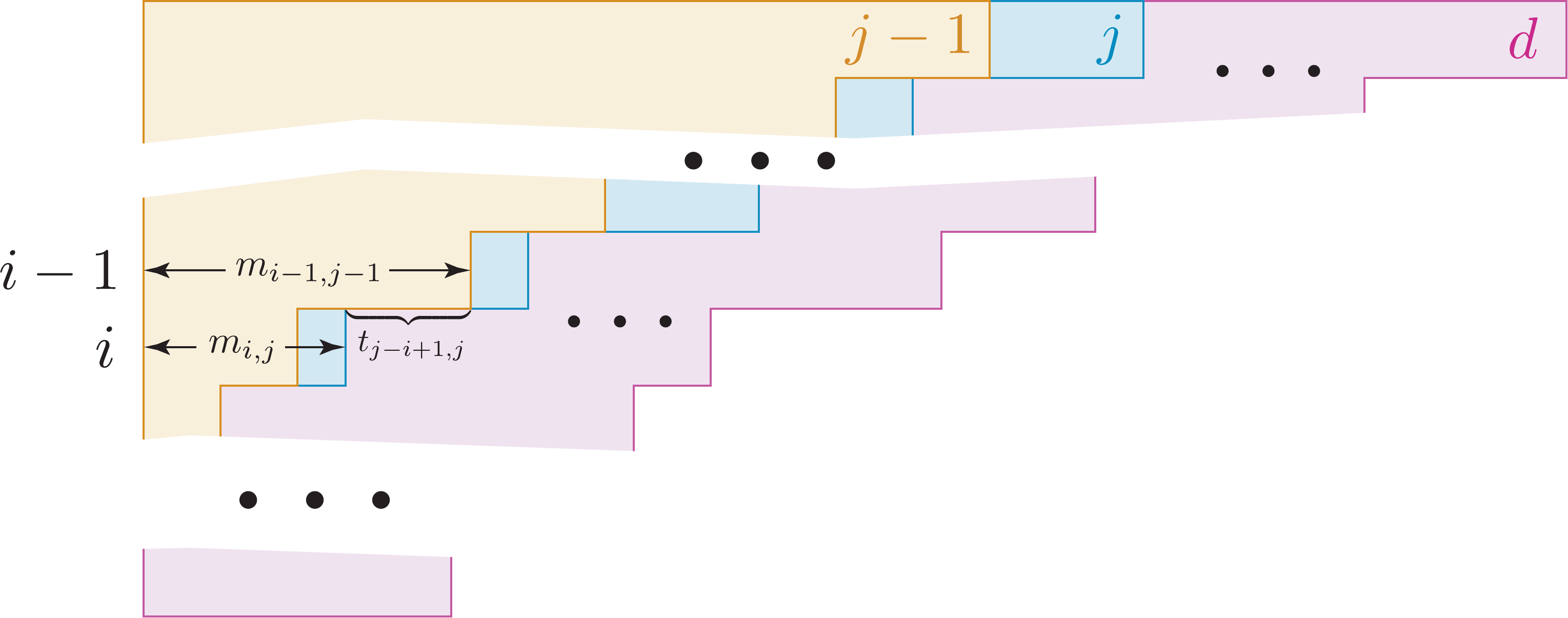"}
\vspace{0.5cm}
\caption{A WT $\wt{m}$ of the typical YD $\yt{\varsigma^\circ}$ for inputs with non-degenerate spectra. The WT is represented as a nested structure of YDs, each filled with different numbers: \( d \) in pink, \( j \) in blue, and \( j-1 \) in yellow, using ellipses to mark the numbers we skipped over. Rows \( i-1 \) and \( i \) are marked, illustrating the relationships between \( m_{i,j} \) and \( t_{i,j} \).}
\label{fig:config}
\end{center}
\end{figure}
From this recursion, it immediately follows that:
\begin{equation}\label{eq:m_ij_expression}
m_{i,j} = np_{j - i + 1} + t_{j - i + 1,d} + t_{j - i + 1,d-1} + \dots + t_{j - i + 1,j+1}.
\end{equation}  
For the configuration parametrized by $t_{ij}$, the number of boxes filled with $k$ is given by:
\begin{equation}
\#_k = \sum_{i=1}^k m_{i,k}-\sum_{i=1}^{k-1} m_{i,k-1}=n p_k + t_{k,d} + t_{k,d-1} + \cdots + t_{k,k+1} - t_{k-1,k} - t_{k-2,k} - \cdots - t_{1,k}.
\end{equation}

With this parametrization, we can quantify the optimality of our protocol asymptotically by explicitly considering the typical fidelity for each \( \yd{\mu_i^\circ} \) irrep. To turn the Schur polynomial into a well-behaved form for study, we normalize it using the monomial contribution from the WT filled with the maximal number of \( p_d \), then \( p_{d-1} \), and so forth, which we call the \emph{maximal WT}. The contribution from this maximal WT configuration is given by $p_d^{n p_d}\cdots p_1^{n p_1}$. Taking the ratio between the Schur polynomial $s^{\yd{\mu_i^\circ}}$ and this configuration, we obtain:
\begin{equation}
\frac{s^{\yd{\mu_i^\circ}}}{p_d^{n p_d} \cdots p_1^{n p_1}}
= \sum_{\{t_{ij}\}} 
\prod_{i < j} \left(\frac{p_i}{p_j}\right)^{t_{i,j}}.
\end{equation} Here, the sum over each $t_{i,j}$ is understood to be taken over all valid configurations. This series is uniformly bounded by the multinomial geometric series obtained by extending each $t_{i,j}$ to range over $\{0,1,2,\dots\}$. Since each ratio $\frac{p_i}{p_j}$ satisfies $\left| \frac{p_i}{p_j} \right| < 1$ for $i < j$, absolute convergence in each variable guarantees convergence of the entire series to $
\prod_{i < j} \frac{1}{1 - \frac{p_i}{p_j}}.
$
Therefore, by the dominated convergence theorem, we can move the limit as $n \to \infty$ inside the summation~\cite[Lemma 6.2]{KW01}. In this scenario, since the upper limit of each summation over $t_{i,j}$ grows as $O(n)$, the constraint becomes asymptotically inactive as $n \to \infty$. As a result, the configuration space of each $t_{ij}$ effectively covers all non-negative integers, recovering the multinomial geometric series.

Next, we apply a similar normalization to the fidelity using the maximal WT. Taking the ratio, we obtain:
\begin{equation}
\frac{d^{\yd{\mu_l^\circ}}}{d^{\yd{\varsigma^\circ}}}\frac{f^{\yd{\mu_l^\circ}}}{p_d^{np_d}\cdots p_1^{np_1}}
= 
\sum_{\{t_{i,j}\}}
\frac{\prod_{i=1}^{d-1}(np_{d - i} + t_{d - i, d} - np_{d-l+1}+l-i)}{\prod_{i \neq l}(np_{d-i+1} - np_{d-l+1}+l-i)}
\, \prod_{i<j}\left(\frac{p_i}{p_j}\right)^{t_{i,j}}.
\end{equation}
For the case \( l = 1 \), the expression simplifies to  
\begin{equation}
\begin{aligned}
& \sum_{\{t_{i,j}\}} \prod_{i=1}^{d-1} \left(1 - \frac{t_{d-i,d}+1}{n p_d - n p_{d-i}+i}\right)\cdot 
\prod_{i<j} \left(\frac{p_i}{p_j}\right)^{t_{i,j}},
\end{aligned}
\end{equation}
which is dominated by the same multinomial geometric series as previously discussed, with each factor proportional to $\left(\frac{p_i}{p_j}\right)^{t_{i,j}}$. 
Moreover, the series converges to 1 factorwise, implying that the asymptotic limit is also $
    \prod_{i < j} \frac{1}{1-\frac{p_i}{p_j}}$. For the general case where \( l > 1 \),
{\small\begin{equation}
\begin{aligned}
& \frac{\prod_{i=1}^{l-2}\left( n p_{d-i} + t_{d-i,d}- n p_{d-l+1}+l-i\right)}{n p_{d} - n p_{d-l+1} + l-1} 
\frac{t_{d-l+1,d}+1}{\prod_{i=2}^{l-1}\left( n p_{d-i+1}-n p_{d-l+1}+l-i \right)} 
\frac{\prod_{i=l}^{d-1} \left( n p_{d-i} + t_{d-i,d}- n p_{d-l+1}+l-i\right)}{\prod_{i=l+1}^{d}\left( n p_{d-i+1}-n p_{d-l+1} +l-i\right)}\prod_{i<j} \left(\frac{p_i}{p_j}\right)^{t_{i,j}} \\
=& \prod_{i=1}^{l-2} \left(1 - \frac{t_{d-i,d}+1}{n p_{d-l+1} - n p_{d-i}-l+i+1}\right) 
\frac{t_{d-l+1,d}+1}{n p_{d} - n p_{d-l+1}-l+1} 
\prod_{i=l}^{d-1} \left(1 - \frac{t_{d-i,d}+1}{n p_{d-l+1} - n p_{d-i}-l+i-1}\right) \prod_{i<j} \left(\frac{p_i}{p_j}\right)^{t_{i,j}}\\ =& \prod_{i=1,\,i \neq l-1}^{d-1} \left(1 - \frac{t_{d-i,d}+1}{n p_{d-l+1} - n p_{d-i}-l+i-1}\right) 
\frac{t_{d-l+1,d}+1}{n p_{d} - n p_{d-l+1}-l+1}\prod_{i<j} \left(\frac{p_i}{p_j}\right)^{t_{i,j}},
\end{aligned}
\end{equation}}
where a similar dominant series can be found to provide a uniform upper bound for the factors. The factor \( \frac{t_{d-l,d}+1}{n p_{d} - n p_{d-l+1}-l+1} \) converges pointwise to 0, while the other factors approach 1. Consequently, the overall expression evaluates to 0. Since \( l \) ranges over all feasible indices, which label all unitary invariant protocols, this establishes that our protocol is also asymptotically optimal. Now, we analyze the overall fidelity by computing the dimension ratio:
\begin{equation}
\begin{aligned}
\frac{d^{\yd{\mu^\circ_1}}}{d^{\yd{\varsigma^\circ}}}=&\frac{\prod_{1\leq j< k\leq d}\left(\mu^\circ_{1\,j}-j-\mu^\circ_{1\,k}-k\right)}{\prod_{1\leq j< k\leq d}\left(\varsigma^\circ_j-j-\varsigma^\circ_k-k\right)}\\
=&\prod_{k=1,k\neq l}^d \left(1+\frac{1}{n\left(p_{d-k+1}-p_{d-l+1}\right)-k+l}\right)\\
=&1-\frac{1}{n}\sum_{k=1}^{d-1}\frac{1}{p_d-p_k}.
\end{aligned}
\end{equation}
Moreover, we apply a similar method to extract the first-order behavior of the fidelity. Consider \( l=1 \) again. Subtracting the zeroth-order term and multiplying by \( n \), we expand to order \( O(1/n) \), yielding:  
\begin{equation}
\begin{aligned}
& \sum_{\{t_{i,j}\}} \sum_{i=1}^{d-1} \left(- \frac{t_{d-i,d}+1}{p_d -  p_{d-i}+i}\right) \cdot
\prod_{i<j} \left(\frac{p_i}{p_j}\right)^{t_{i,j}}\\
=&\sum_{\{t_{i,j}\}} \sum_{i=1}^{d-1}\left( \left(- \frac{t_{d-i,d}+1}{p_d - p_{d-i}+i}\right) \left(\frac{p_{d-i}}{p_d}\right)^{t_{d-i,d}}
\prod_{j<k, j\neq d-i, k\neq d } \left(\frac{p_j}{p_k}\right)^{t_{j,k}}\right)\\
=-&\sum_{i=1}^{d-1}\frac{p_d}{(p_d-p_{d-i})^2}\cdot
\prod_{i<j}\frac{1}{1-\frac{p_i}{p_j}}, \end{aligned}
\end{equation}
In the last line, we evaluate the summation using the identity $\sum_{t=0}^\infty t x^t = \frac{x}{(1 - x)^2}$, which converges for $|x| < 1$ and is derived by differentiating the geometric series. Finally, assembling the parts, we obtain
\begin{equation}
\begin{aligned}
&\mathcal{F}=\frac{f^{\yd{\mu_1^\circ}}}{s^{\yd{\varsigma^\circ}}}=\frac{d^{\yd{\varsigma^\circ}}}{d^{\yd{\mu^\circ_1}}}\left(\frac{d^{\yd{\mu^\circ_1}}}{d^{\yd{\varsigma^\circ}}}\frac{f^{\yd{\mu_1^\circ}}}{p_d^{n p_d} \cdots p_1^{n p_1}}\right)
\frac{p_d^{n p_d} \cdots p_1^{n p_1}}{s^{\yd{\varsigma^\circ}}}+O\left(\frac{1}{n^2}\right)\\
=& 1-\frac{1}{n}\left(\sum_{k=1}^{d-1}\frac{p_d}{(p_d-p_k)^2}-\sum_{k=1}^{d-1}\frac{1}{p_d-p_k}\right)+O\left(\frac{1}{n^2}\right)\\
=& 1-\frac{1}{n}\sum_{k=1}^{d-1}\frac{p_k}{(p_d-p_k)^2}+O\left(\frac{1}{n^2}\right).
\end{aligned}
\end{equation}
It turns out that this result also holds for cases with degenerate non-principle eigenvalues, one can perform an additional change of variables on \(\{t_{i,j}\}\) to ensure that all sums range from 0 to a value on the order of \(O(n)\). However, this procedure does not offer further insight and is omitted here. As long as the principal eigenstate is nondegenerate, the degeneracy does not affect the maximal expression. Note that by setting \(p_d = b\) and \(p_{\neq d} = a\), one recovers the familiar form \(\mathcal{F} = 1 - \frac{1}{n}\frac{d-1}{d}\frac{\lambda}{(1-\lambda)^2} + O\bigl(\frac{1}{n^2}\bigr)\) in the depolarizing case. Furthermore, a higher-order series expansion can be carried out to obtain additional correction terms, which we omit for brevity.
 
\section{Correctness of the Efficient Implementation}
\label{sec:GQPE-QPA} 
\subsection{Schur sampling}
To understand the algorithm, we first introduce GQPE as a primitive. It begins by preparing the state in the control registers as the trivial irrep state, i.e., \(\ket{[n,0,\cdots,0]}_{\Lambda} \ket{[1\,2\cdots n]}_L \ket{[1\,2\cdots n]}_R\), followed by applying the inverse Fourier transform $F^{-1}$, a controlled permutation $C\!P$, and then the Fourier transform $F$. This primitive corresponds to the following operator:

\begin{equation}
\begin{aligned}
 &F C\!P_{\text{ctrl},\text{data}}  F^{-1}\ket{[n,0,\cdots,0]}_{\Lambda} \ket{[1\, 2\cdots n]}_L \ket{[1\,2 \cdots n]}_R \\
           =& F C\!P \frac{1}{\sqrt{n!}} \sum_g \ket{g}_{\text{ctrl}}  \\
           =& \frac{1}{\sqrt{g^{\yd{\varsigma}}}} F \sum_g  \ket{g} P_{g, \text{data}}\\ 
           =& \sum_{\yd{\varsigma}}\sum_{\yt{s},\yt{t}}\frac{1}{\sqrt{g^{\yd{\varsigma}}}} \ket{\substack{\yd{\varsigma} \\ \yt{s} \yt{t}}} \Pi^{\;\;\yd{\varsigma}}_{\yt{s}\yt{t}\,\text{data}}.
\end{aligned}
\end{equation}

To analyze whether our protocol correctly operates on the data register, we work in the Schur-transformed picture of the data register. The data register decomposes into three registers: \(\Lambda^\prime\), which encodes the irrep, \(S\) for the Schur module, and \(W\) for the Weyl module. 
\vspace{0.2cm}
\begin{center}
\begin{minipage}[c]{0.37\textwidth}
\begin{center}
\includegraphics[scale=0.29]{"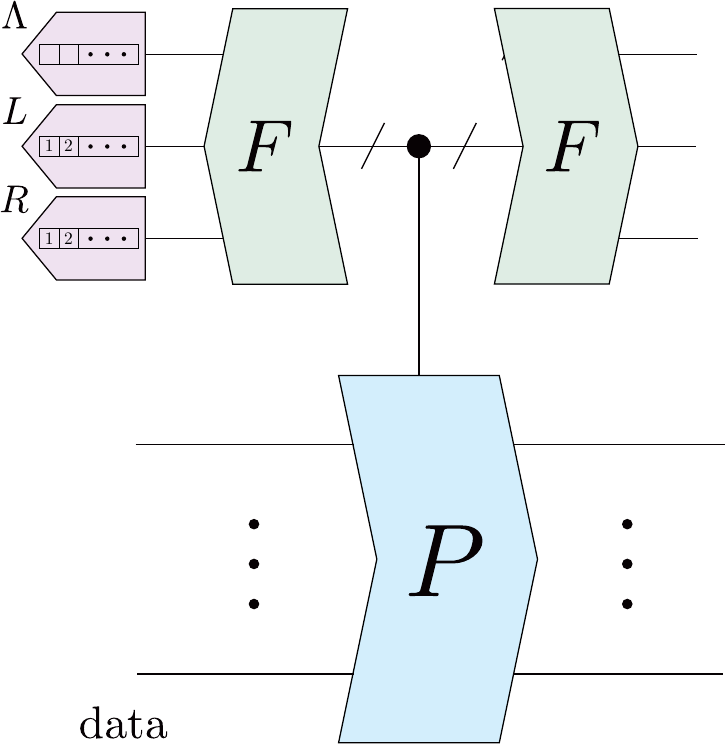"}
\end{center}
\end{minipage}%
\begin{minipage}[c]{0.05\textwidth}
\begin{center}
\raisebox{0.1cm}{
\(\xlongequal{\quad}\)}
\end{center}
\end{minipage}%
\begin{minipage}[c]{0.58\textwidth}
\begin{equation}
\sum_{\yd{\varsigma}}\sum_{\yt{s},\yt{t}}\frac{1}{\sqrt{g^{\yd{\varsigma}}}}
\ket{\substack{\yd{\varsigma} \\ \yt{s} \yt{t}}}_{\textrm{ctrl}} 
\ket{\yd{\varsigma}}\bra{\yd{\varsigma}}_{\Lambda^\prime}\ket{s}\bra{t}_S(\mathbb{I}_{d^{\yd{\varsigma}}})_W.
\label{eq:projection_2}
\end{equation}
\end{minipage}%
\hfill
\end{center}
\vspace{0.2cm}

Here we use boxes with angled edges to indicate Hermitian conjugation. A box representing an operator $A$ acting on a ket $|v\rangle$ is drawn with right-pointing edges $\rangle$, with the input on the left and the output on the right. Its adjoint $A^\dagger$ is drawn with left-pointing edges $\langle$, reflecting the box to reverse the direction of the edges.

Now, we are ready to analyze the first step of Schur sampling. We begin by performing QPE on the data register $\rho^{\otimes n}$. In the Schur-transformed picture, this corresponds to the state $\sum_{\yd{\varsigma}}\ket{\yd{\varsigma}}\bra{\yd{\varsigma}}_{\Lambda^\prime}\mathbb{I}_{g^{\yd{\varsigma}}\,S}\, \rho^{\yd{\varsigma}}_W.$ Applying QPE, which corresponds to the tensor contraction depicted below:
\vspace{0.2cm}
\begin{center}
\begin{minipage}[c]{0.25\textwidth}
\begin{center}
\includegraphics[scale=0.3]{"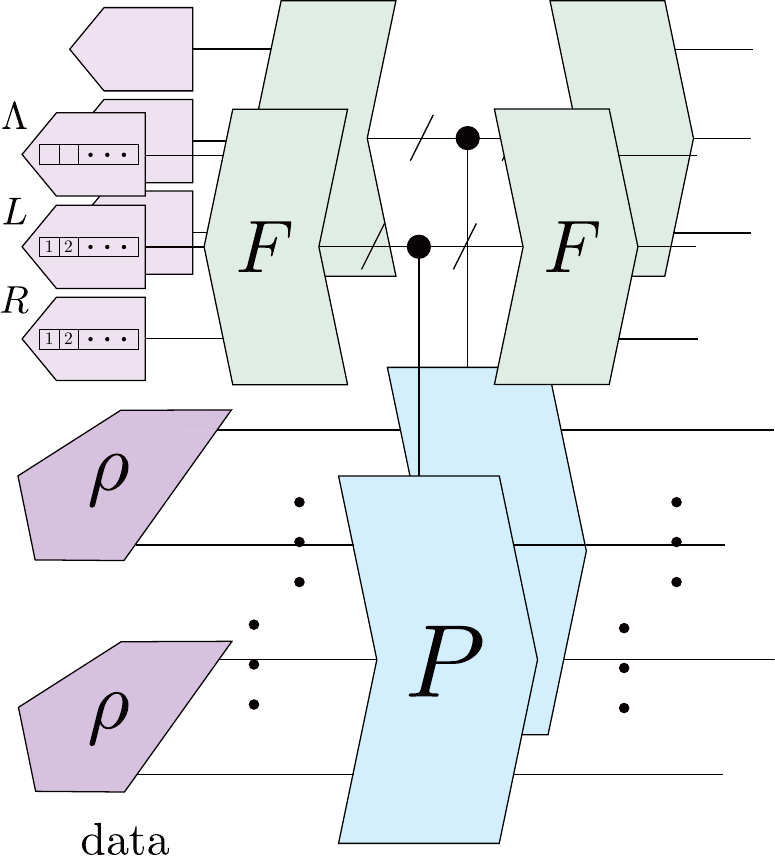"}
\end{center}
\end{minipage}%
\begin{minipage}[c]{0.05\textwidth}
\begin{center}
\raisebox{1cm}{
\(\xlongequal{\quad}\)}
\end{center}
\end{minipage}%
\begin{minipage}[c]{0.7\textwidth}
\begin{equation}\begin{aligned}
&\sum_{\yd{\varsigma}}\sum_{(r),(s),(t),(u)} \frac{1}{g^{\yd{\varsigma}}} \ket{\substack{\yd{\varsigma} \\ \yt{r} \yt{s}}}\bra{\substack{\yd{\varsigma} \\ \yt{t} \yt{u}}}_{\textrm{ctrl}}\ket{\yd{\varsigma}}\bra{\yd{\varsigma}}_{\Lambda^\prime} 
\ket{(r)}\langle(s)|(u)\rangle\bra{(t)}_S\,\rho_W^{\yd{\varsigma}}\\
=&\sum_{\yd{\varsigma}}\sum_{(r),(t)} \frac{1}{g^{\yd{\varsigma}}} \sum_{(s)} \ket{\substack{\yd{\varsigma} \\ \yt{r} \yt{s}}}\bra{\substack{\yd{\varsigma} \\ \yt{t} \yt{s}}}_{\textrm{ctrl}}\ket{\yd{\varsigma}}\bra{\yd{\varsigma}}_{\Lambda^\prime} 
\ket{(r)}\bra{(t)}_S\,\rho_W^{\yd{\varsigma}}.
\label{eq:projection}
\end{aligned}
\end{equation}
\end{minipage}%
\hfill
\end{center}
\vspace{0.2cm}

Now, we simply measure the $\Lambda$ register in the computational basis and stores the result in the classical register $A$. This results in $\yd{\varsigma}$ appearing in both the $\Lambda$ and $\Lambda^\prime$ registers. The probability of measuring a specific $\yd{\varsigma}$ corresponds to the probability of the Schur sampling outcome, $g^{\yd{\varsigma}}s^{\yd{\varsigma}}$. In the remaining registers, we obtain the state:

\begin{align}
\frac{1}{g^{\yd{\varsigma}}}\sum_{(r),(s)} \ket{(r)}\bra{(s)}_L\frac{\mathbb{I}_{g^{\yd{\varsigma}}}}{g^{\yd{\varsigma}}}_R\ket{(r)}\bra{(s)}_S\,\frac{\rho^{\yd{\varsigma}}}{s^{\yd{\varsigma}}}_W.
\end{align}
The circuit transfers the state of \( S \) into \( R \), at the cost of entangling \( S \) with \( L \). This allows us to manipulate the state originally held in the \( S \) register by performing operations on \( R \). We then trace out the $R$ register, as the specific irrep labeled by the YT is not relevant since we will later impose the column-ordered configuration regardless. This results in $\frac{1}{g^{\yd{\varsigma}}}\sum_{(r),(s)} \ket{(r)}\bra{(s)}_L\ket{(r)}\bra{(s)}_S\,\frac{\rho^{\yd{\varsigma}}}{s^{\yd{\varsigma}}}_W$ in the remaining system, successfully performing Schur sampling to extract the $\yd{\varsigma}$ information while preserving the data originally in $W$.

\subsection{Correction}
To implement the corrections, we need a mechanism to reintroduce the optimal column-ordered configuration \( (\lambda^\diamond) \) back into the system. We achieve this by first preparing the state $(\lambda^\diamond)$ based on the measured value of \( \yd{\varsigma} \) stored in $A$, which resets our state to $\frac{1}{g^{\yd{\varsigma}}}\sum_{(r),(s)} \ket{(r)}\bra{(s)}_L\ket{(\varsigma^\diamond)}\bra{(\varsigma^\diamond)}_R\ket{(r)}\bra{(s)}_S\,\frac{\rho^{\yd{\varsigma}}}{s^{\yd{\varsigma}}}_W$.
Next, we apply the inverse GQPE by sequentially applying \( F \), followed by \( C\!P^{-1} \), and then \( F^{-1} \). This process yields  

\begin{center}
\begin{minipage}[c]{0.26\textwidth}
\begin{center}
\includegraphics[scale=0.3]{"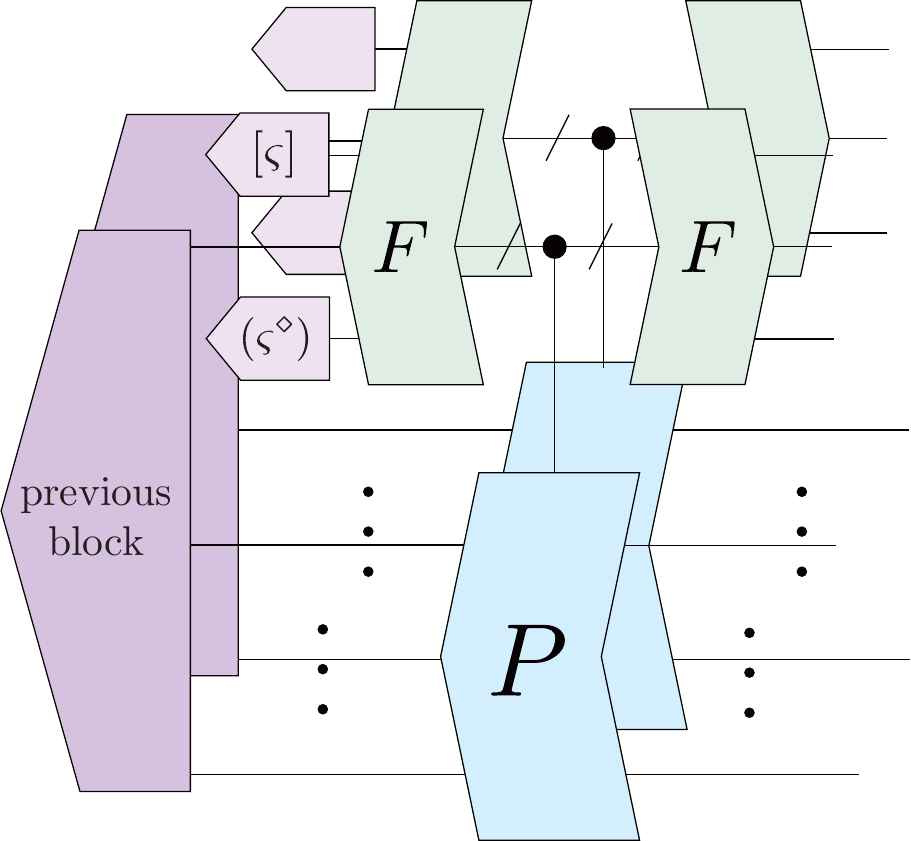"}
\end{center}
\end{minipage}%
\begin{minipage}[c]{0.025\textwidth}
\begin{center}
\raisebox{0.1cm}{
\(\xlongequal{\quad}\)}
\end{center}
\end{minipage}%
\begin{minipage}[c]{0.7\textwidth}
\begin{equation}
\begin{aligned}
&\sum_{\substack{(r),(s),(t)\\ (u),(v),(w)}}\!\!\!
{\left.\dfrac{\Tr\ket{\substack{\yd{\varsigma} \\ \yt{r} \yt{\varsigma^\diamond}}}\!\left\langle\substack{\yd{\varsigma} \\ \yt{s} \yt{\varsigma^\diamond}}\right|\left.
\substack{\yd{\varsigma} \\ \yt{t} \yt{u}}\right\rangle\!\bra{\substack{\yd{\varsigma} \\ \yt{v} \yt{w}}}}{g^{\yd{\varsigma}\,2}}\right.}_{\textrm{ctrl}}
\langle(v)|(r)\rangle\langle(s)|(t)\rangle\ket{(w)}\!\bra{(u)}_S\,\frac{\rho^{\yd{\varsigma}}}{s^{\yd{\varsigma}}}_W  \\
=&\ket{(\varsigma^\diamond)}\bra{(\varsigma^\diamond)}_S\,\frac{\rho^{\yd{\varsigma}}}{s^{\yd{\varsigma}}}_W.
\label{eq:projection_3}
\end{aligned}
\end{equation}
\end{minipage}%
\hfill
\end{center}

Since the state is normalized at this stage of the operation, measuring the control register always yields the trivial irrep state. This confirms that the process defines a valid quantum channel, as it remains trace-preserving. Finally, note that we are working within the Schur-transformed picture. Transforming back under the Schur transform, we obtain the state \( \rho_{\yt{\varsigma^\diamond}} \).

\section{SWAPNET}
\label{subsec:circuit}

Here, we demonstrate that the SWAPNET algorithm implements the optimal QPA protocol for three effective qudits. For a SWAP test, the measurement outcome of the ancilla determines the applied projection. The corresponding projections are:
\begin{equation} \left\{ \begin{aligned} \Pi_{\tiny\Yvcentermath1\young(12)} &= \frac{1}{2}\big(\id + P_{(12)}\big), \quad \text{if measuring 0}, \\ \Pi_{\tiny\Yvcentermath1\young(1,2)} &= \frac{1}{2}\big(\id - P_{(12)}\big), \quad \text{if measuring 1}. \end{aligned} \right. \end{equation}
where $\Pi_{\tiny\Yvcentermath1\young(12)}$ projects onto the symmetric subspace, and $\Pi_{\tiny\Yvcentermath1\young(1,2)}$ projects onto the antisymmetric subspace.

During the first SWAP test, if $q_1$ and $q_2$ are antisymmetrized, then the first two registers are in the the $\scriptsize\Yvcentermath1\young(1,2)$ configuration. No matter what the outcome involving $q_3$ is, whether it is $\scriptsize\Yvcentermath1\young(13,2)$ or $\scriptsize\Yvcentermath1\young(1,2,3)$, $q_1$ and $q_2$ are going to be discarded eventually. Therefore we only need to return $q_3$ as our amplified qudit.

For the followiwng SWAP tests, as in~\cref{alg:SWAPNET}, let $l+1=N$ represents the total number of SWAP tests applied. Suppose the antisymmetric configuration is measured during the $l+1$-th SWAP test. This effectively applies a channel defined by the Krauss operator
 $K_l=\Pi_{\tiny\Yvcentermath1\young(1,2)}(P_{(23)}\Pi_{\tiny\Yvcentermath1\young(12)})^l.$
It turns out that $K_l$ is proportional to the transition operator $\Pi_{\tiny\Yvcentermath1\young(13,2)\,\tiny\Yvcentermath1\young(12,3) }$. 
Specifically,
\begin{equation}
    K_l=(-1)^{(l+1)}\frac{\sqrt{3}}{2^l}\Pi_{\tiny\Yvcentermath1\young(13,2)\,\tiny\Yvcentermath1\young(12,3)}.
\end{equation}
We prove this equation by induction. 
    Base case: 
    \begin{align}
    \Pi_{\tiny\Yvcentermath1\young(1,2)}\big(P_{(23)}\Pi_{\tiny\Yvcentermath1\young(12)}\big)= \frac{1}{4}\big(-(13)+(23)-(123)+(132)\big)=\frac{\sqrt{3}}{2}\Pi_{\tiny\Yvcentermath1\young(13,2)\,\tiny\Yvcentermath1\young(12,3) }
    \end{align} by definition.
    Inductive hypothesis:
    $\Pi_{\tiny\Yvcentermath1\young(1,2)}\big(P_{(23)}\Pi_{\tiny\Yvcentermath1\young(12)}\big)^{(l-1)}=(-1)^l\frac{1}{2^l}(-(13)+(23)-(123)+(132))$.\\Suppose this is true, then
    \begin{align}
\nonumber\Pi_{\tiny\Yvcentermath1\young(1,2)}\big(P_{(23)}\Pi_{\tiny\Yvcentermath1\young(12)}\big)^{l}&= (-1)^l\frac{1}{2^l}\big(-(13)+(23)-(123)+(132)\big)(23)\big(()+(12)\big)\\
         &=(-1)^{(l+1)}\frac{1}{2^{(l+1)}}\big(-(13)+(23)-(123)+(132)\big).
    \end{align} 
When $\Pi_{\tiny\Yvcentermath1\young(13,2)\,\tiny\Yvcentermath1\young(12,3) }$ acts on $\rho^{\otimes 3}$, it turns the state $\rho_{\;\tiny\young(12,3)}$ unitarily into $\rho_{\;\tiny\young(13,2)}$, which can be dealt with by tracing out the first two registers.

Summing up all the branches defined by the $K_l$'s for all $l$ from $0$ to $\infty$. By the definition of transition operators, we obtain $\sum_l K_l\rho^{\otimes 3}K^{\dag}_l = \rho_{\;\tiny\young(13,2)}$, which asymptotically gives the correct projection on the $\tiny\Yvcentermath1\young(13,2)$ irrep. This shows that none of the information encoded in the inputs is lost, as the circuit remains effectively reversible.

Moreover, this also implies the following: Suppose no anti-symmetric configuration has ever been observed before we terminate the protocol at some iteration. In this case, an effective channel consisting of only symmetrizers, $P_{(23)}\Pi_{\tiny\Yvcentermath1\young(12)}(P_{(23)}\Pi_{\tiny\Yvcentermath1\young(12)})^l$, is applied. Asymptotically, this effective channel tends towards the totally symmetric projector $\Pi_{\tiny\young(123)}$. Therefore, at the end of the algorithm we also trace out the first two qudits and return $q_3$.

\section{Conjectured Construction of optimal Multi-output QPA protocol}
\label{sec:conjecture}
Following Ref.~\cite{KW01}, we are interested in studying the higher-dimensional analogue for the optimal QPA protocol with multiple outputs. In this case, the cost matrix can be chosen as 
\begin{equation}
    C^{\yd{\varsigma}}=\int(\sigma^{\otimes m})^\top\otimes\rho^{\yd{\varsigma}}.
\end{equation} We conjecture the optimal construction is given by the projector onto the smallest irrep under the YD partial order. The smallest YD in this order is obtained from the following removal routine:
\begin{enumerate}
    \item First, start from the first row, remove at most $\varsigma_1-\varsigma_2$ boxes.
    \item Then, start from the second row, remove at most $\varsigma_2-\varsigma_3$ boxes.
    \item Proceed by repeating the same procedure for all the rows until the last row.
    \item For the last row, remove however many boxes.
\end{enumerate}
This procedure is summarized in the following diagram:
\begin{figure}[H]
\begin{center}
\includegraphics[scale=0.17]{"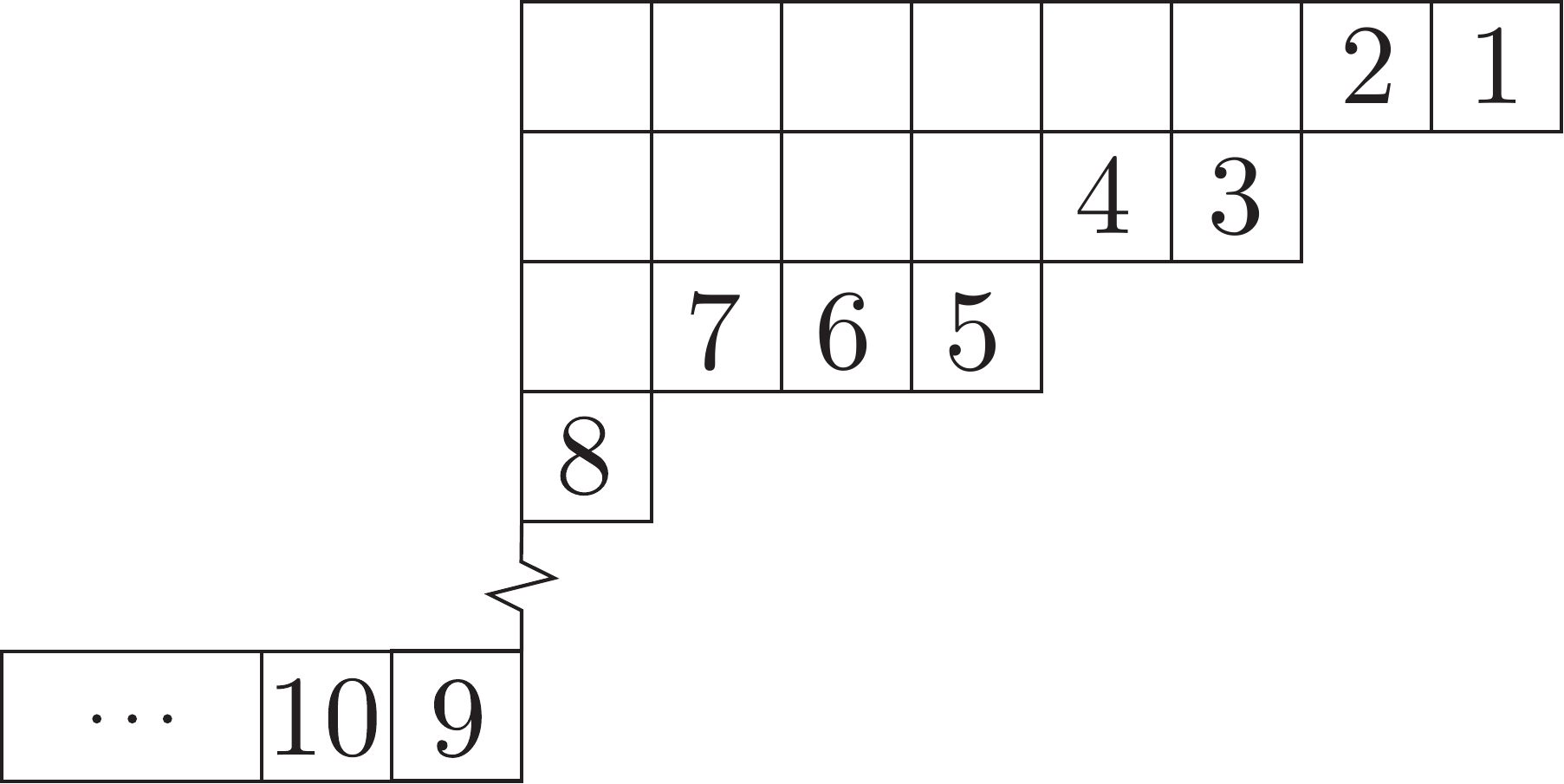"}
\end{center}
\caption{Diagram summarizing the conjectured optimal procedure for multi-output QPA, where the numbers indicate the order of removal.}
\end{figure}

This is suggested by the possible existence (though not yet proven) of an ordering relation for generalized Schur polynomials. For instance, take $d=3$ with $n=3$ and $m=2$. Consider the branch $\yd{\varsigma}=\overline{\mathbf{15}}=\Yvcentermath1\scriptsize\yng(3,2)$, and the target representation $\mathbf{10}=\Yvcentermath1\scriptsize\yng(3)$. $T^{\yd{\varsigma}}$ can be decomposed into a block diagonal form corresponding to the representations:
\ytableausetup
 {centertableaux,mathmode, boxframe=normal, boxsize=1em}
\begin{equation}
\ydiagram{3+3,3+2,3}\oplus\ydiagram{2+3,2+1,2}\oplus\ydiagram{1+3,1+0,1}\oplus\ydiagram{2+2,2+2,2}\oplus\ydiagram{1+2,1+1,1}\oplus\ydiagram{2,0,0}.
\end{equation}Note that in the above formula, the ordering obeys the majorization relation $\succeq$ except for the third and fourth tableau which are incomparable. However, there is a unique smallest tableaux that is given by $\yd{2,0,0}=\scriptsize\Yvcentermath1\ydiagram{2}$. 
\vspace{0.5mm}
\section{Implementation Details and Numerical Simulations}
\label{sec:numerical}
\subsection{Simulating QPA for Trotterized Hamiltonian evolution}

We take \( H \) to be a transverse-field Ising Hamiltonian of the form \( -J\sum_{\langle i,j\rangle} Z_iZ_j - h\sum_{i=1}^k X_i \), with \( J= h =1\) representing the interaction and field strengths, respectively, \( \langle i,j \rangle \) indicating nearest-neighbor pairs, and $k=4$.
Circuit-level noise is simulated with Qiskit Aer using the EstimatorV2 primitive with $102400$ shots in each case. 

We include additional simulation results with various values of $k$ to demonstrate the robustness of QPA with respect to scaling.
\begin{figure}[H]
    \centering
    \begin{subfigure}[t]{0.48\textwidth}
        \centering
        \includegraphics[width=\textwidth]{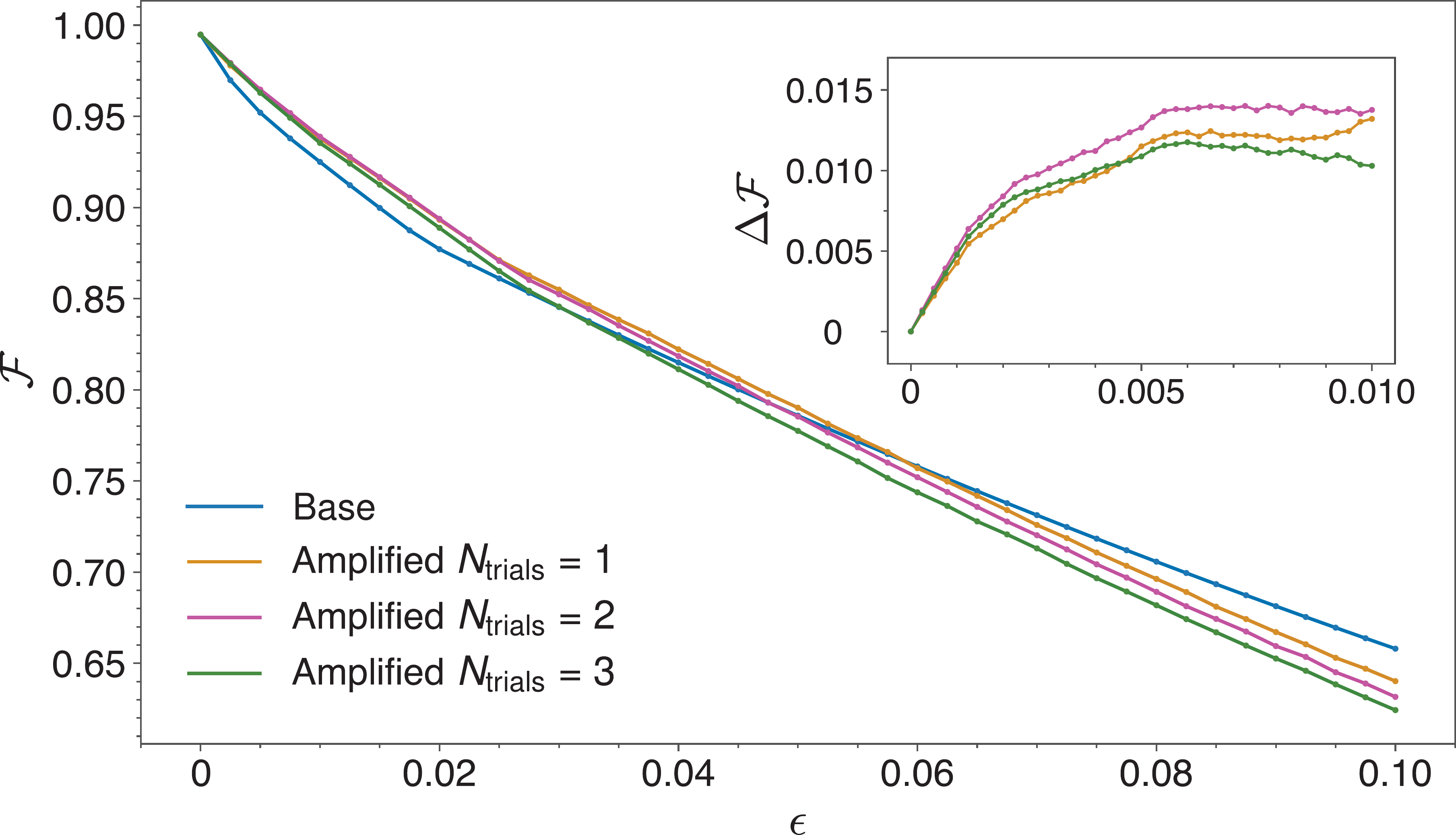}
        \caption{Base and amplified fidelities for $N_\mathrm{trials} = 1, 2, 3$ as functions of the gate error rate $\epsilon$ from  Hamiltonian simulation with $k = 2$. Inset shows the fidelity gain.}
        \label{fig:curves_k2}
    \end{subfigure}
    \hfill
    \begin{subfigure}[t]{0.48\textwidth}
        \centering
        \includegraphics[width=\textwidth]{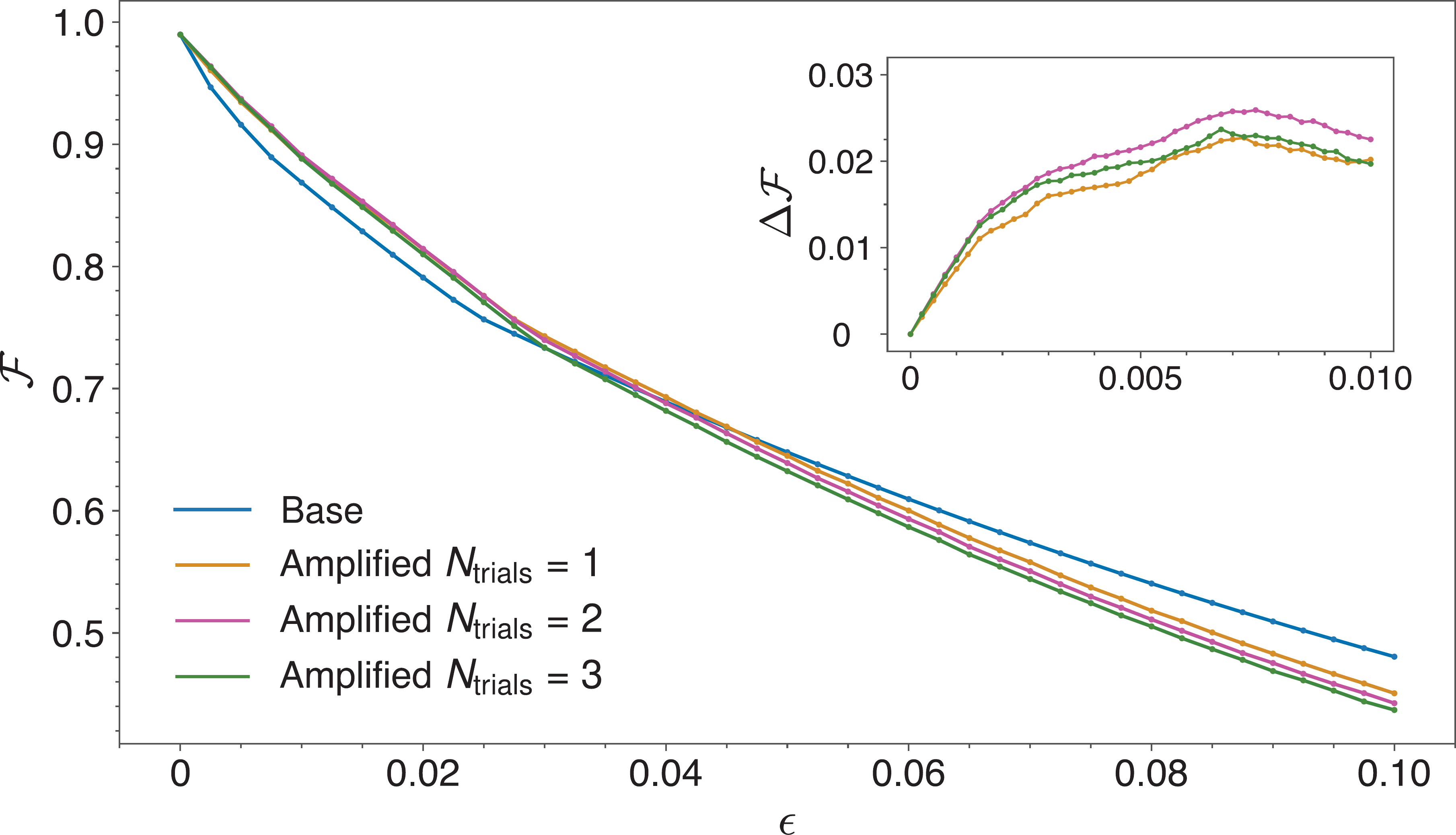}
        \caption{Base and amplified fidelities for $N_\mathrm{trials} = 1, 2, 3$ as functions of the gate error rate $\epsilon$ from Hamiltonian simulation with $k = 3$. Inset shows the fidelity gain.}
        \label{fig:curves_k3}
    \end{subfigure}
\end{figure}

\subsection{Simulating QPA for Adiabatic State Preparation}
The native Hamiltonian of Rydberg simulators is given by
$
H = \frac{\Omega}{2} \sum_i \left( |g\rangle\langle r|_i + |r\rangle\langle g|_i \right) - \sum_i \Delta\, n_i + \sum_{\langle i,j\rangle} V_0\, n_i n_j,
$ where $n_i = |r\rangle\langle r|_i$ is the Rydberg number operator at site $i$, and $V_0$ is the nearest-neighbor interaction strength between atoms in the Rydberg state. We assume periodic boundary conditions throughout. In addition, we include noise modeled by local jump operators $L_i = \sqrt{2\gamma}\,n_i$ with coherence time $T_2$. Starting from the initial state $|g\rangle^{\otimes 4}$, the Rabi frequency $\Omega$ is ramped up and down, while the detuning $\Delta$ is gradually tuned from negative to positive to prepare the one-dimensional $Z_2$ state $\frac{1}{\sqrt{2}} \left( \ket{grgr} + \ket{rgrg} \right)$.

The ramp profile consists of a linear rise of $\Omega$ from 0 to $\Omega_{\text{max}}$, followed by a flat hold period during which $\Delta$ is tuned linearly from $-\Delta_\mathrm{start}$ to $\Delta_\mathrm{start}$, and then a linear fall of $\Omega$ back to 0. The three segments have relative durations in the ratio $0.33:0.34:0.33$, with a total evolution time $T_\mathrm{tot}$.
This specific choice has minimal impact on the phase transition, as long as the evolution remains adiabatic~\cite{bernien2017probing}. To optimize the parameters, we first performed a coarse grid scan over $\Omega_{\max}/2\pi \in (0, 3]$ MHz, $\Delta_{\text{start}}/2\pi \in (0, 30]$ MHz, and total evolution duration $T_{\text{tot}} \in (1, 10]$ $\mu s$. We further assume $T_2 = 6\,\mu\text{s}$ and $V_0 = 40\,\text{MHz}$. We then refined the search near the optimal region using a finer grid: $\Omega_{\max}/2\pi \in (2.80, 3.00]\,\text{MHz}$ with a resolution of $0.02$ MHz, $\Delta_{\text{start}}/2\pi \in (2.0, 6.0]$ MHz with a resolution of $0.2$ MHz, and $T_{\text{tot}} \in [1.0, 2.0]\,\mu\text{s}$ with a resolution of $0.1\,\mu\text{s}$. The resulting fidelity with the target state is shown as a heatmap over $\Omega_{\max}/2\pi$ and $\Delta_{\text{start}}/2\pi$, maximized over $T_{\mathrm{tot}}$, in Supplementary~\cref{fig:fidelity_heatmap}. The best-performing instance, presented in Supplementary~\cref{fig:drive_profile}, occurs at $T_{\mathrm{tot}} = 1.0\,\mu\text{s}$, with $\Omega_{\max} = 3.00 \times 2\pi\,\text{MHz}$ and $\Delta_{\mathrm{start}} = 4.0 \times 2\pi\,\text{MHz}$, achieving a fidelity of 0.828. The resulting density matrix is far from depolarized and exhibits a sharply peaked spectrum, with eigenvalues ordered as $p_{16} \gg p_{15} \gg p_{14},\, p_{13},\, p_{12} \gg p_{11}$, as shown in Supplementary~\cref{fig:spectrum_rho}. However, the bias, defined as the fidelity between the principal eigenstate and the ideal state, exceeds $0.997$, making QPA possible. The theoretical optimum is obtained by analytically evaluating the QPA channel on the inputs.

We truncate the resulting density matrix to its top five eigenstates, which together capture over $0.9949$ of the total probability. These eigenstates are used as inputs to the QPA protocol, yielding 125 possible combinations. Each combination is simulated using the EstimatorV2 primitive with $102400$ shots. For each input state, the fidelity is computed individually, and the overall fidelity is obtained by a weighted average, where the weights are given by the product of eigenvalues corresponding to each combination.

We present numerical results obtained by simulating the analog adiabatic evolution, and we summarize the findings in~\cref{fig:rydb_figures}.

\begin{figure}[H]
    \centering
\includegraphics[width=0.9\textwidth]{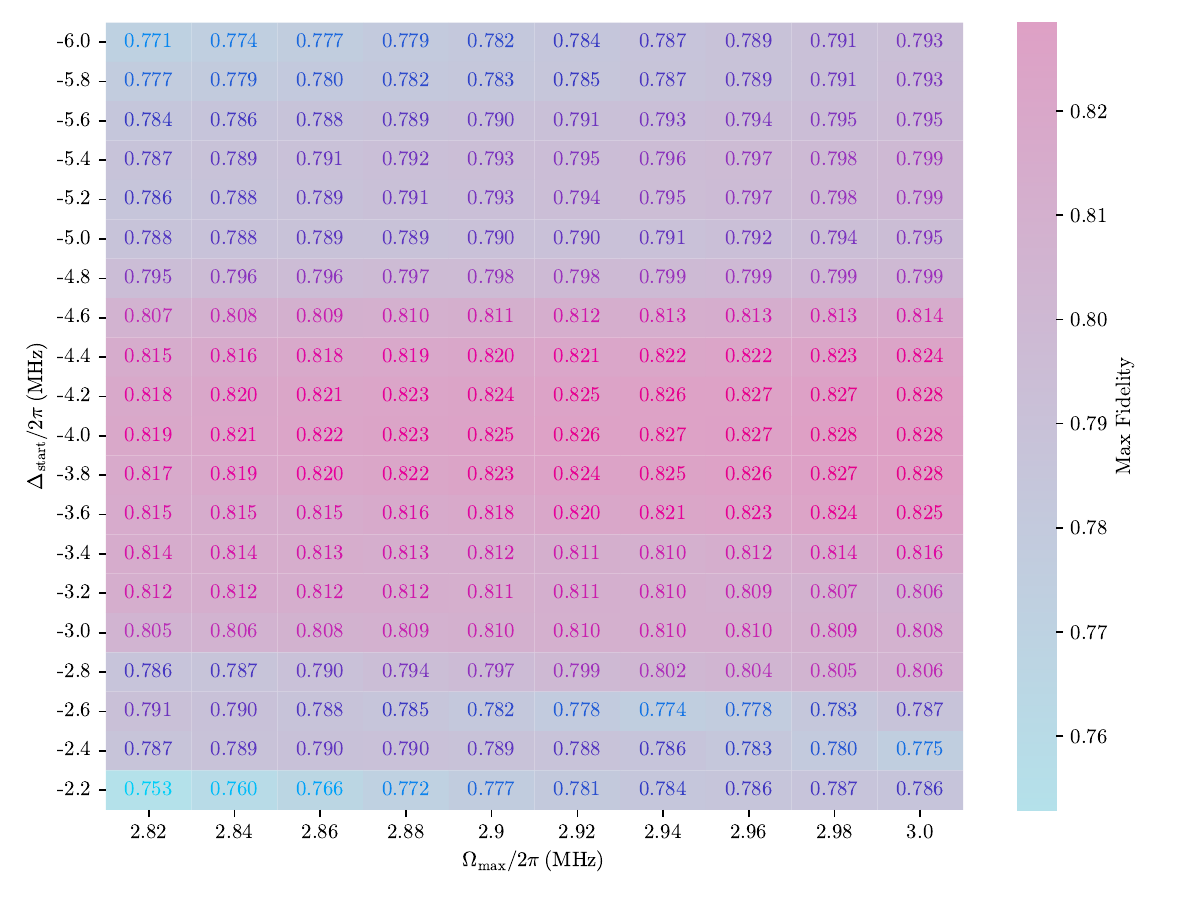}
    \caption{Fidelity heatmap from the refined parameter scan over maximum Rabi frequency and starting detuning. Each grid point shows the highest fidelity achieved over all total evolution times $T_\text{tot}$.}
    \label{fig:fidelity_heatmap}
\end{figure}

\begin{figure}[H]
    \begin{subfigure}{0.48\textwidth}
        \centering
        \includegraphics[width=\textwidth]{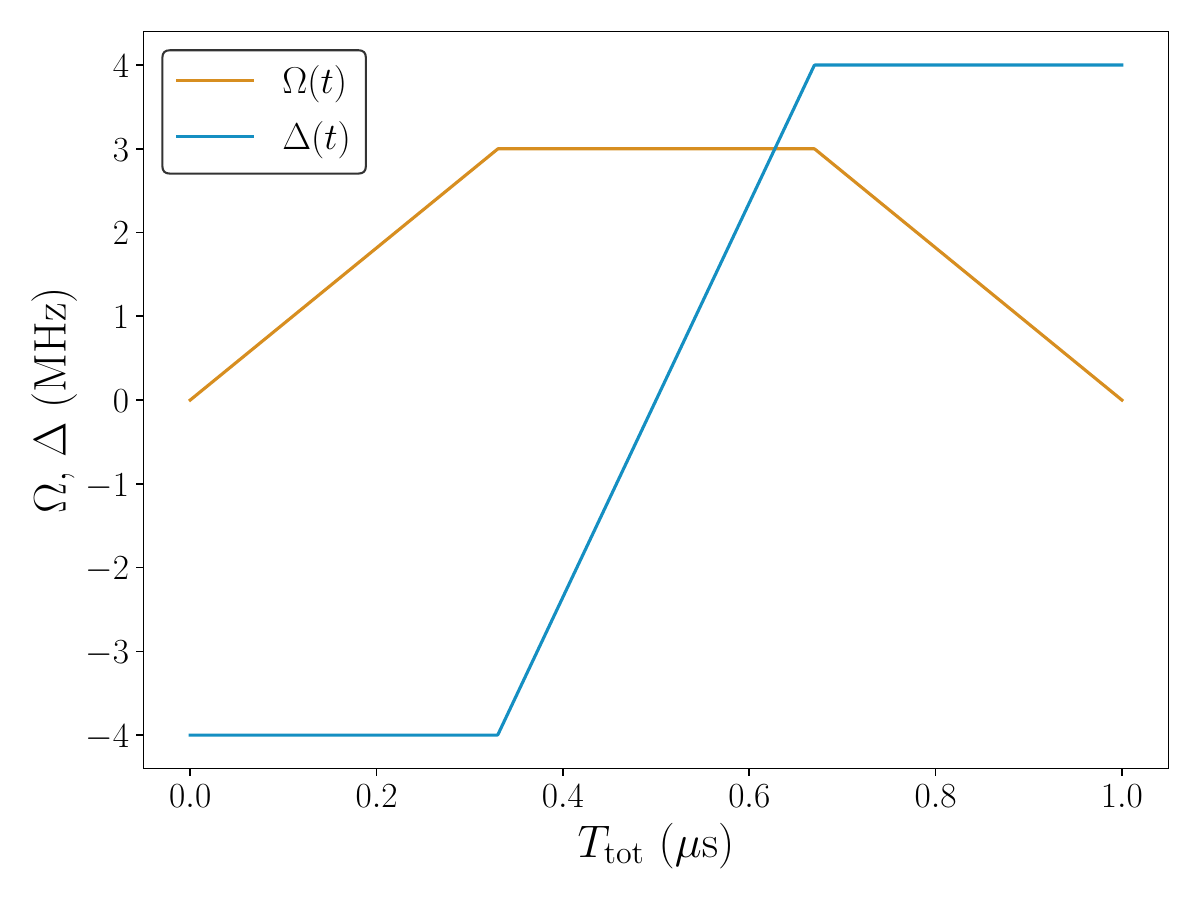}
        \caption{Ramp profile of the Rabi frequency $\Omega$ and detuning $\Delta$ used in the adiabatic protocol.}
        \label{fig:drive_profile}
    \end{subfigure}
    \hfill
    \begin{subfigure}{0.48\textwidth}
        \centering
        \includegraphics[width=\textwidth]{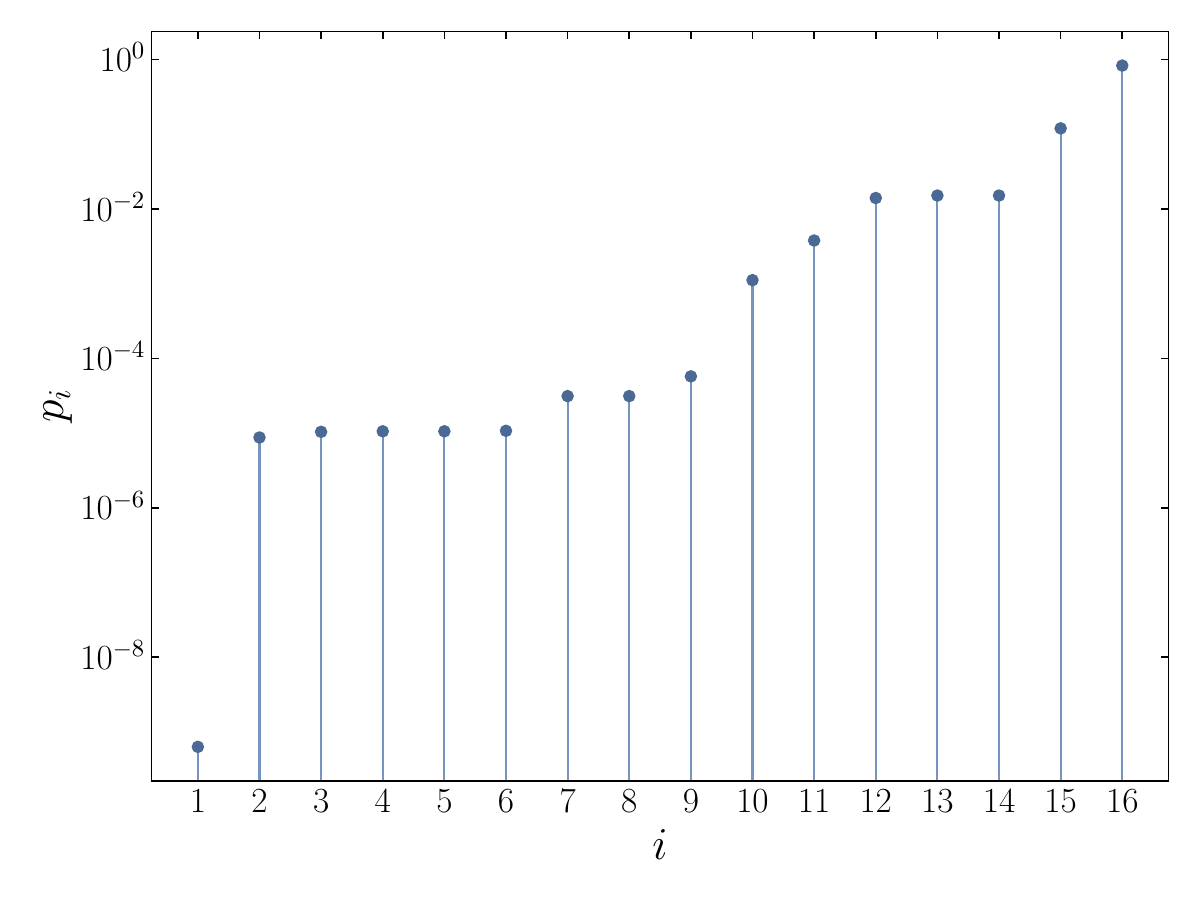}
        \caption{Eigenvalue spectrum of the final state $\rho$, displayed in ascending order.}
        \label{fig:spectrum_rho}
    \end{subfigure}
    \phantomcaption
    \label{fig:rydb_figures}
\end{figure}

\subsection{Noisy state preparation using QPA on superconducting hardware.}  
We model tunable global depolarizing noise by applying a random Pauli operator \(P \in \mathcal{P}_k = \{I, X, Y, Z\}^{\otimes k}\) to each input  with probability \(\lambda / 4^k\), where \(k = 2\). The base fidelity is given by \(\mathcal{F} = 1 - \lambda \left(1 - \frac{1}{2^k}\right)\) and is used across all three experiments, as circuit noise for one layer of gates negligible. 
We analytically evaluate the noiseless QPA channel on the depolarized inputs to obtain the fidelities shown in \cref{fig:curves}c. The theoretical optimal fidelity corresponding to $N_{\mathrm{trials}} = \infty$, as derived from Supplementary~\ref{equ:fidelity}, is given by $\mathcal{F} = \frac{1}{8} (-2 + \lambda)(1 + \lambda)(-4 + 3\lambda)$.

Following the analytical benchmarks, we perform experiments on IBM’s Marrakesh superconducting heavy-hex processor using the SamplerV2 primitive. 
For simplicity, we take the ideal state to be $\ket{0}^{\otimes 2}$, and estimate the amplified fidelity as the fraction of measurements yielding this outcome. To avoid mid-circuit classical control, we run separate circuits for each $n$-th iteration of SWAPNET and study measurement outcomes that satisfy $z_n = 1$ and $z_{<n} = 0$ (and also on $z_n = 0$ and $z_{<n} = 0$ for the final iteration). 
We first run on the simulated FakeMarrakesh processor, which uses a static noise profile. For each setting, we generate 100000 randomized circuit instances with two shots each (\cref{fig:curves}d). We then run on the experimental processor, performing 22000 randomized instances with two shots each (\cref{fig:curves}e).

We include additional simulation results for $k = 4$ and $k = 5$, corresponding to qudit dimensions $d = 2^4$ and $d = 2^5$, respectively, in the setting with global noise only.
\begin{figure}[H]
    \centering
    \begin{subfigure}[t]{0.48\textwidth}
        \centering
        \includegraphics[width=\textwidth]{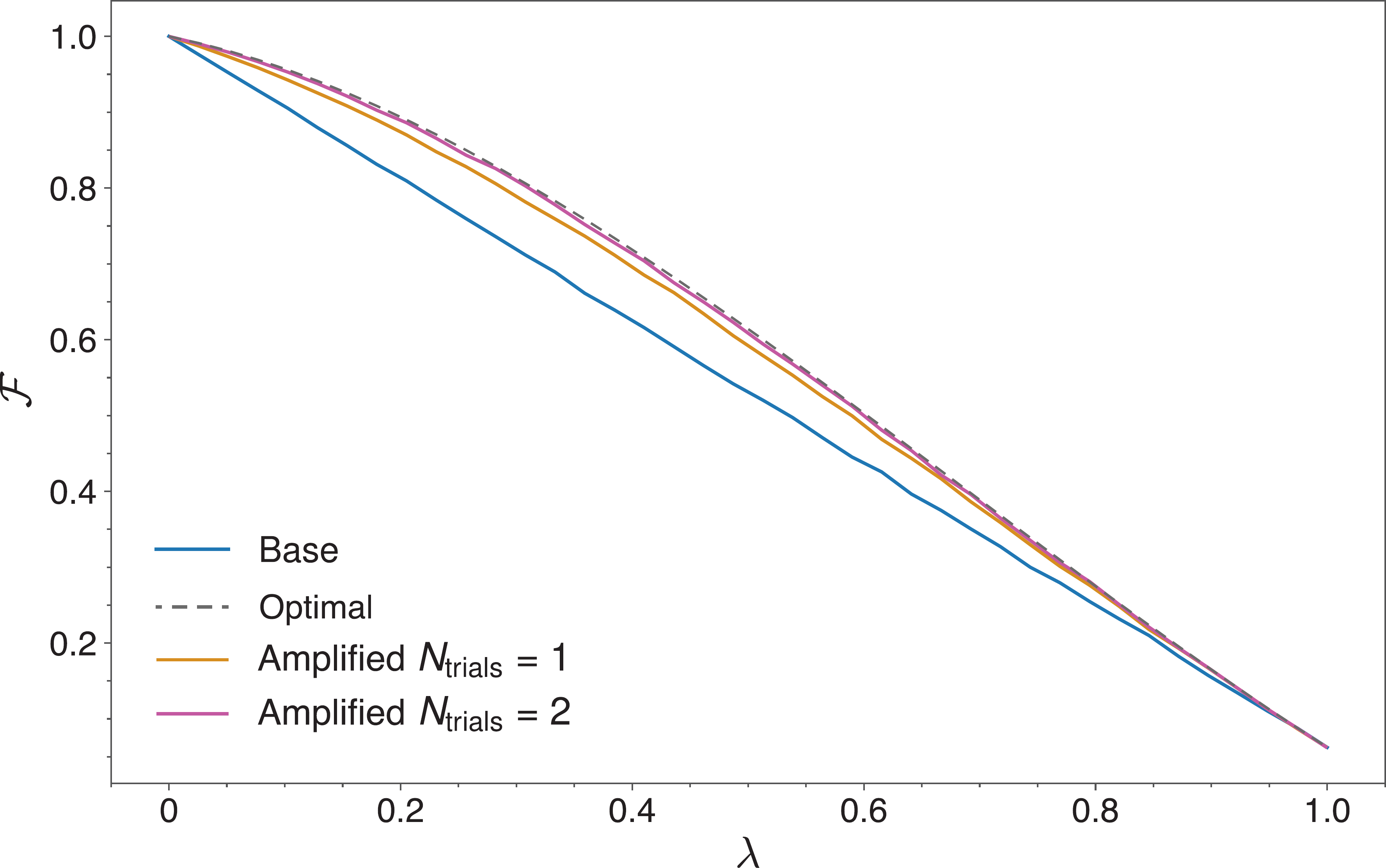}
        \caption{Base, amplified, and optimal fidelities for $N_\mathrm{trials} = 1, 2$ as functions of the depolarization strength $\lambda$ from Hamiltonian simulation with $k = 4$.}
        \label{fig:curves_k4}
    \end{subfigure}
    \hfill
    \begin{subfigure}[t]{0.48\textwidth}
        \centering
        \includegraphics[width=\textwidth]{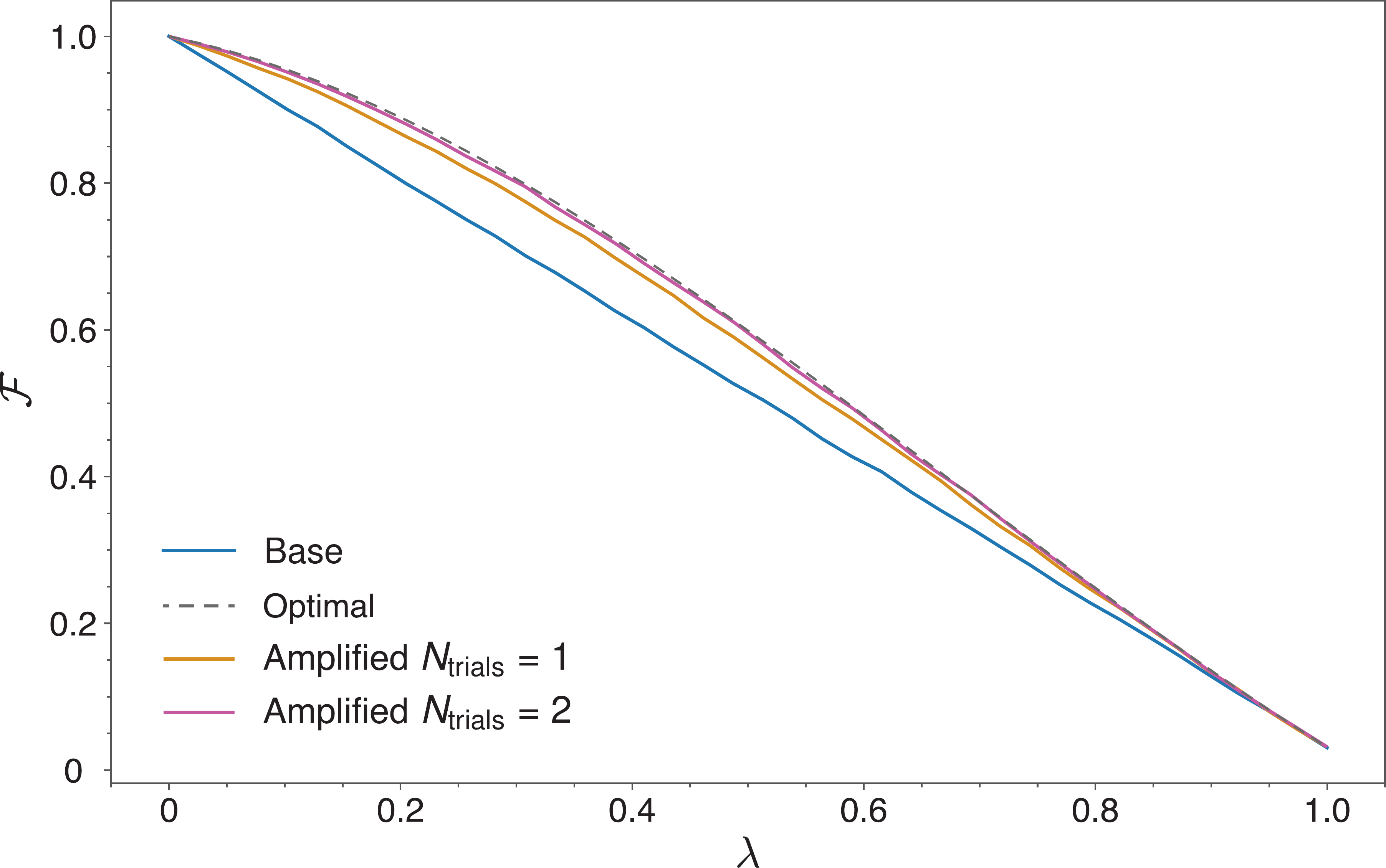}
        \caption{Base, amplified, and optimal fidelities for $N_\mathrm{trials} = 1, 2$ as functions of the depolarization strength $\lambda$ from Hamiltonian simulation with $k = 5$.}
        \label{fig:curves_k5}
    \end{subfigure}
\end{figure}
\bibliography{main.bib}
\end{document}